\documentclass[journal,onecolumn]{IEEEtran}

\usepackage{amsmath,amssymb,amsfonts,amsthm,mathtools,bm}
\usepackage[noadjust]{cite}
\usepackage{enumerate}
\usepackage{tikz}
\usetikzlibrary{arrows.meta,calc,positioning,decorations.pathreplacing}
\usepackage{hyperref}
\usepackage{cleveref}

\newtheorem{theorem}{Theorem}[section]
\newtheorem{lemma}{Lemma}[section]
\newtheorem{proposition}{Proposition}[section]
\newtheorem{corollary}{Corollary}[section]
\theoremstyle{definition} 
\newtheorem{definition}{Definition}[section]

\crefname{construction}{Construction}{Constructions}
\Crefname{construction}{Construction}{Constructions}

\newtheorem{remark}{Remark}[section]
\newtheorem{assumption}{Assumption}
\crefname{assumption}{Assumption}{Assumptions}
\Crefname{assumption}{Assumption}{Assumptions}

\newtheorem*{clm}{Claim}
\newtheorem{example}{Example}[section]

\newcommand{\cfcite}[2]{\textit{cf.}~\cite[#2]{#1}}

\newcommand\nc\newcommand
\nc{\cA}{\mathcal{A}}\nc{\cB}{\mathcal{B}}\nc{\cC}{\mathcal{C}}\nc{\cD}{\mathcal{D}}
\nc{\cE}{\mathcal{E}}\nc{\cF}{\mathcal{F}}\nc{\cG}{\mathcal{G}}\nc{\cH}{\mathcal{H}}
\nc{\cI}{\mathcal{I}}\nc{\cJ}{\mathcal{J}}\nc{\cK}{\mathcal{K}}\nc{\cL}{\mathcal{L}}
\nc{\cM}{\mathcal{M}}\nc{\cN}{\mathcal{N}}\nc{\cO}{\mathcal{O}}\nc{\cP}{\mathcal{P}}
\nc{\cQ}{\mathcal{Q}}\nc{\cR}{\mathcal{R}}\nc{\cS}{\mathcal{S}}\nc{\cT}{\mathcal{T}}
\nc{\cU}{\mathcal{U}}\nc{\cV}{\mathcal{V}}\nc{\cW}{\mathcal{W}}\nc{\cX}{\mathcal{X}}
\nc{\cY}{\mathcal{Y}}\nc{\cZ}{\mathcal{Z}}

\nc{\bba}{\mathbf{a}}\nc{\bbb}{\mathbf{b}}\nc{\bbc}{\mathbf{c}}\nc{\bbd}{\mathbf{d}}
\nc{\bbe}{\mathbf{e}}\nc{\bbf}{\mathbf{f}}\nc{\bbg}{\mathbf{g}}\nc{\bbh}{\mathbf{h}}
\nc{\bbi}{\mathbf{i}}\nc{\bbj}{\mathbf{j}}\nc{\bbk}{\mathbf{k}}\nc{\bbl}{\mathbf{l}}
\nc{\bbm}{\mathbf{m}}\nc{\bbn}{\mathbf{n}}\nc{\bbo}{\mathbf{o}}\nc{\bbp}{\mathbf{p}}
\nc{\bbq}{\mathbf{q}}\nc{\bbr}{\mathbf{r}}\nc{\bbs}{\mathbf{s}}\nc{\bbt}{\mathbf{t}}
\nc{\bbu}{\mathbf{u}}\nc{\bbv}{\mathbf{v}}\nc{\bbw}{\mathbf{w}}\nc{\bfx}{\mathbf{x}}
\nc{\bby}{\mathbf{y}}\nc{\bbz}{\mathbf{z}}

\nc{\bbA}{\mathbf{A}}\nc{\bbB}{\mathbf{B}}\nc{\bbC}{\mathbf{C}}\nc{\bbD}{\mathbf{D}}
\nc{\bbE}{\mathbf{E}}\nc{\bbF}{\mathbf{F}}\nc{\bbG}{\mathbf{G}}\nc{\bbH}{\mathbf{H}}
\nc{\bbI}{\mathbf{I}}\nc{\bbJ}{\mathbf{J}}\nc{\bbK}{\mathbf{K}}\nc{\bbL}{\mathbf{L}}
\nc{\bbM}{\mathbf{M}}\nc{\bbN}{\mathbf{N}}\nc{\bbO}{\mathbf{O}}\nc{\bbP}{\mathbf{P}}
\nc{\bbQ}{\mathbf{Q}}\nc{\bbR}{\mathbf{R}}\nc{\bbS}{\mathbf{S}}\nc{\bbT}{\mathbf{T}}
\nc{\bbU}{\mathbf{U}}\nc{\bbV}{\mathbf{V}}\nc{\bbW}{\mathbf{W}}\nc{\bfX}{\mathbf{X}}
\nc{\bbY}{\mathbf{Y}}\nc{\bbZ}{\mathbf{Z}}

\nc{\sA}{\mathsf{A}}\nc{\sB}{\mathsf{B}}\nc{\sC}{\mathsf{C}}\nc{\sD}{\mathsf{D}}
\nc{\sE}{\mathsf{E}}\nc{\sF}{\mathsf{F}}\nc{\sG}{\mathsf{G}}\nc{\sH}{\mathsf{H}}
\nc{\sI}{\mathsf{I}}\nc{\sJ}{\mathsf{J}}\nc{\sK}{\mathsf{K}}\nc{\sL}{\mathsf{L}}
\nc{\sM}{\mathsf{M}}\nc{\sN}{\mathsf{N}}\nc{\sO}{\mathsf{O}}\nc{\sP}{\mathsf{P}}
\nc{\sQ}{\mathsf{Q}}\nc{\sR}{\mathsf{R}}\nc{\sS}{\mathsf{S}}\nc{\sT}{\mathsf{T}}
\nc{\sU}{\mathsf{U}}\nc{\sV}{\mathsf{V}}\nc{\sW}{\mathsf{W}}\nc{\sX}{\mathsf{X}}
\nc{\sY}{\mathsf{Y}}\nc{\sZ}{\mathsf{Z}}

\nc{\bma}{\bm{a}}\nc{\bmb}{\bm{b}}\nc{\bmc}{\bm{c}}\nc{\bmd}{\bm{d}}
\nc{\bme}{\bm{e}}\nc{\bmf}{\bm{f}}\nc{\bmg}{\bm{g}}
\nc{\bmh}{\bm{h}}\nc{\bmi}{\bm{i}}\nc{\bmj}{\bm{j}}
\nc{\bmk}{\bm{k}}\nc{\bml}{\bm{l}}\nc{\bmm}{\bm{m}}
\nc{\bmn}{\bm{n}}\nc{\bmo}{\bm{o}}\nc{\bmp}{\bm{p}}
\nc{\bmq}{\bm{q}}\nc{\bmr}{\bm{r}}\nc{\bms}{\bm{s}}\nc{\bmt}{\bm{t}}\nc{\bmu}{\bm{u}}\nc{\bmv}{\bm{v}}\nc{\bmw}{\bm{w}}\nc{\bmx}{\bm{x}}\nc{\bmy}{\bm{y}}\nc{\bmz}{\bm{z}}

\nc{\bmA}{\bm{A}}\nc{\bmB}{\bm{B}}\nc{\bmC}{\bm{C}}\nc{\bmD}{\bm{D}}
\nc{\bmE}{\bm{E}}\nc{\bmF}{\bm{F}}\nc{\bmG}{\bm{G}}
\nc{\bmH}{\bm{H}}\nc{\bmI}{\bm{I}}\nc{\bmJ}{\bm{J}}
\nc{\bmK}{\bm{K}}\nc{\bmL}{\bm{L}}\nc{\bmM}{\bm{M}}
\nc{\bmN}{\bm{N}}\nc{\bmO}{\bm{O}}\nc{\bmP}{\bm{P}}
\nc{\bmQ}{\bm{Q}}\nc{\bmR}{\bm{R}}\nc{\bmS}{\bm{S}}\nc{\bmT}{\bm{T}}\nc{\bmU}{\bm{U}}\nc{\bmV}{\bm{V}}\nc{\bmW}{\bm{W}}\nc{\bmX}{\bm{X}}\nc{\bmY}{\bm{Y}}\nc{\bmZ}{\bm{Z}}

\nc{\bmalpha}{\bm{\alpha}}\nc{\bmbeta}{\bm{\beta}}\nc{\bmgamma}{\bm{\gamma}}

\nc{\Part}{\mathsf{Part}}

\newcommand{\mathset}[1]{\left\{#1\right\}}

\newcommand{\abs}[1]{ \left|#1\right|}
\newcommand{\ceilenv}[1]{\left\lceil #1 \right\rceil}
\newcommand{\floorenv}[1]{\left\lfloor #1 \right\rfloor}
\newcommand{\parenv}[1]{\left( #1 \right)}
\newcommand{\sparenv}[1]{\left[ #1 \right]}

\newcommand{\lwt}[1]{\mathsf{wt}_{L_1}\left(#1\right)}

\newcommand{\vt}{\mathsf{VT}}

\newcommand{\e}{\mathsf{e}}

\newcommand{\Lab}{\mathsf{Lab}}
\DeclareMathOperator{\lcm}{lcm}
\newcommand{\red}{\mathsf{red}_2}

\title{Substring Edit Correcting Codes and Optimal Single Burst-Deletion Correcting Codes}
\author{Zuo~Ye and Gennian~Ge%
\thanks{This research is supported by the National Key Research and Development Program of China under Grant 2025YFC3409900, the National Natural Science Foundation of China under Grant 12231014 and Grant 12501466, Beijing Scholars Program, and Xiaomi Young Scholars Program.%

Z. Ye is with the Institute of Mathematics and Interdisciplinary Sciences, Xidian University,
Xi'an 710126, China. Email: yezuo@xidian.edu.cn.%

G. Ge is with the School of Mathematical Sciences, Capital Normal University, Beijing 100048, China. Email: gnge@zju.edu.cn.
}
}

\begin{document}
\maketitle
\begin{abstract}
A $k$-substring edit in a sequence first deletes a substring of length at most $k$, and then inserts a sequence of length at most $k$ at the same position. A code that can correct a $k$-substring edit is called a $k$-substring edit code. In this paper, we develop a new localization method and use it to construct a $q$-ary $k$-substring edit correcting code with $\log n+8\log\log n+o(\log\log n)$ bits of redundancy for any fixed $q\ge2$ and $k\ge1$, where $n$ is the code length. For the binary alphabet, this improves upon the redundancy $\log n+16k\log\log n+o(\log\log n)$ obtained by Li \emph{et al}. 
When the deleted substring and the inserted sequence have different lengths, we further construct codes with redundancy $\log n+O_{q,k}(1)$, which is optimal up to an additive constant. As corollaries, for all fixed $q\ge2$ and $1\le t\le T$, we obtain $q$-ary $(\le t)$-burst-deletion correcting codes and $(t,T)$-localized deletion correcting codes with redundancies $\log n+O_{q,t}(1)$ and $\log n+O_{q,T}(1)$, respectively. To the best of our knowledge, these are the first constructions attaining optimal redundancy up to an additive constant for these two deletion models over the full range of fixed parameters. For $(\le t)$-burst-deletion correction with $t\ge2$, such redundancy had previously been achieved only for $q=t=2$ by Levenshtein in 1967.
\end{abstract}

\section{Introduction}
A $k$-substring edit replaces a substring of length at most $k$ by another sequence of length at most $k$. More precisely, over the alphabet $\Sigma_q=\mathset{0,1,\ldots,q-1}$, it transforms a sequence $\bmx=\bmu\bms\bmv$ into $\bmy=\bmu\bmt\bmv$, where $\bmu$, $\bms$, $\bmt$, and $\bmv$ are sequences over $\Sigma_q$, with the lengths of $\bms,\bmt$ being at most $k$. This operation encompasses insertions, deletions, substitutions, and combinations of these errors  within one short region. A $k$-substring edit correcting code must uniquely recover the original codeword $\bmx$ from $\bmy$ without knowing the location of the edit or the lengths of $\bms$ and $\bmt$.

Tang \emph{et al.} introduced the $k$-substring edit model in \cite{YuanyuanTang2025TCOMM}. Their analysis of source-code revisions and simulated nanopore sequencing reads provides empirical evidence that edits tend to occur in clusters. In file synchronization, for example, updating an expression in a source-code file may replace a block of characters while leaving the rest unchanged. Similarly, a cluster of insertion, deletion, and substitution errors in a stored or transmitted sequence can be represented by a substring replacement. From a coding perspective, the main feature of a $k$-substring edit is that all changes are confined to one short region, even though the symbols after that region may shift. Using this restriction can substantially reduce the redundancy compared with protecting against the same number of edits at arbitrary positions.

For the binary alphabet and fixed $k$, Tang \emph{et al.} \cite{YuanyuanTang2025TCOMM} constructed codes correcting one $k$-substring edit with redundancy $2\log n+o(\log n)$, where $n$ is the code length. Li \emph{et al.} \cite{YutingLi2025IT} subsequently improved the redundancy bound to $\log n+16k\log\log n+o(\log\log n)$. As shown later in the current paper, for fixed $q$ and $k$, any $q$-ary code of length $n$ correcting one $k$-substring edit requires at least $\log n-O_{q,k}(1)$ bits of redundancy. Thus, Li's construction achieves the optimal leading term $\log n$ in the binary case, with an additive $O_k(\log\log n)$ gap between their redundancy bound and the lower bound.

\subsection{Related Error Models}\label{sec_relatedmodel}
We review some error models which are special cases of the $k$-substring edit model.
\subsubsection{\textbf{Burst/Localized Deletions}}
A burst of exactly $t$ deletions removes $t$ consecutive symbols, whereas a burst of at most $t$ deletions removes a substring of length at most $t$. Codes correcting these errors are called $t$-burst-deletion correcting codes and $(\le t)$-burst-deletion correcting codes, respectively. The study of these codes goes back to Levenshtein's work in 1967 \cite{Levenshtein1967Burst}, where a binary $(\le 2)$-burst-deletion correcting code was constructed. Applications in DNA storage and file synchronization have motivated further studies of these codes in the past two decades \cite{Schoeny2017IT,Lenz2020ISIT,Song2023IT,Shuche2024IT,Tuan2024IT,Song2025Entropy,Yubo2025tsburstIT,ZuoYe2026IT}. A central objective is to construct $q$-ary codes whose redundancy matches the lower bound $\log n-O_{q,t}(1)$ \cite{Schoeny2017IT,Shuche2024IT} up to an additive constant, for fixed $q\ge2$ and $t\ge1$.

For a prescribed burst length $t$, Sun \emph{et al.} \cite{Yubo2025tsburstIT} constructed $q$-ary $t$-burst-deletion correcting codes with redundancy $\log n+O_{q,t}(1)$, attaining this objective. For $(\le t)$-burst-deletion correcting codes, several redundancy upper bounds are available in different parameter ranges \cite{Song2023IT,Song2025Entropy,Yubo2025tsburstIT,ZuoYe2026IT}. For arbitrary fixed $q,t\ge2$, Song, Cai, and Quek \cite{Song2025Entropy} obtained redundancy $\log n+8\log\log n+o(\log\log n)$. For even $q$, Sun \emph{et al.}\ \cite[Theorem~10]{Yubo2025tsburstIT} obtained redundancy $\log n+t\log\log n+O_{q,t}(1)$. Very recently, Ye, Sun, and Ge \cite[Theorem~17]{ZuoYe2026IT} reduced the redundancy to $\log n+(t-1)\log\log n+O_{q,t}(1)$ when $t<q$, or when $q$ is even and $t<2q$. Sharper bounds are available in some special cases. For $t=1$, optimal $q$-ary single-deletion codes have redundancy $\log n+O_q(1)$ \cite{Tenengolts1984,Tuan2024IT}. For $q=2$, Levenshtein's code achieves redundancy $\log n+O(1)$ when $t=2$ \cite{Levenshtein1967Burst}, and \cite[Remark~18]{ZuoYe2026IT} gives redundancy $\log n+\log\log n+O(1)$ when $t=3$.

A $t$-localized deletion removes at most $t$ symbols from an unknown window of $t$ consecutive symbols, without requiring the deleted positions to be consecutive. Schoeny \emph{et al.} introduced and studied this model under the name ``bursts of non-consecutive deletions'' \cite{Schoeny2017IT}.
Bitar \emph{et al.} \cite{RawadBitar2021ISIT} constructed binary codes with redundancy $\log n+O\bigl(t\log^2(t\log n)\bigr)$, which becomes $\log n+O_t((\log\log n)^2)$ for fixed $t$. For even $q$, Sun \emph{et al.} \cite[Theorem~11]{Yubo2025tsburstIT} obtained $q$-ary codes with redundancy $\log n+2t\log\log n+O_{q,t}(1)$.

More generally, a $(t,T)$-localized deletion removes at most $t$ symbols from an unknown window of $T$ consecutive symbols, where the window length $T\ge t$ is specified separately from $t$. The $t$-localized-deletion model is the special case $T=t$. For fixed $q,t$ and $T$, Ye, Sun, and Ge \cite[Theorem~20]{ZuoYe2026IT} constructed $q$-ary codes correcting $(t,T)$-localized deletions with redundancy $\log n+(t-1)\log\log n+O_{q,t,T}(1)$ when $t<q$, or when $q$ is even and $t<2q$.

Clearly, the $(t,T)$-localized-deletion model includes both the $t$-burst-deletion and $(\le t)$-burst-deletion models. Its connection with substring edits can be seen directly. Let $\bms$ be the affected length-$T$ window, and let $\bmt$ be the subsequence remaining after $d$ deletions, where $1\le d\le t$. Replacing $\bms$ by $\bmt$ is a $T$-substring edit.

\subsubsection{\textbf{The \texorpdfstring{$(t,s)$}{(t,s)}-Burst Error}}
For prescribed nonnegative integers $t$ and $s$ with $t+s>0$, a $(t,s)$-burst error replaces a length-$t$ substring by a length-$s$ sequence. The special case $(t,s)=(2,1)$ was introduced by Schoeny \emph{et al.} in \cite{Schoeny2017IT} as a component of their constructions for correcting localized deletions. Lu and Zhang subsequently formulated the general model in \cite{LuZhang2023IT}. This model also generalizes the $t$-deletion-burst error model.

Sun \emph{et al.} \cite{Yubo2025tsburstIT} later constructed $q$-ary $(t,s)$-burst correcting codes with redundancy $\log n+O_{q,t,s}(1)$ for every fixed $q\ge2$ and every fixed pair $(t,s)$ with $t+s>0$. Their redundancy is optimal up to an additive constant depending only on $q,t,s$.

Every $(t,s)$-burst error with $\max\{t,s\}\le k$ is a $k$-substring edit. A $(t,s)$-burst correcting code is required to correct errors for one prescribed pair $(t,s)$, whereas a $k$-substring edit correcting code must handle all such pairs simultaneously. For a codeword of length $n$, the received sequence has length $n-t+s$, which determines $s-t$ but generally does not determine $t$ and $s$ individually. Consequently, the optimal redundancy result for each fixed pair does not by itself give a single code with the same redundancy for the more general substring edit model.

\subsubsection{\textbf{The \texorpdfstring{$t$}{t}-Burst-Edit}}
A $t$-burst-edit, as defined in \cite{Yubosun2025IT}, is one of three operations: deleting exactly $t$ consecutive symbols, inserting a sequence of length $t$ at one position, or replacing a length-$t$ substring by another length-$t$ substring. When $t=1$, this reduces to the single-edit model, whose binary case was studied by Levenshtein \cite{Levenshtein1966}. Sun, Xi, and Ge \cite{SunXiGe2023IT} considered the burst-edit channel in the sequence reconstruction setting. Sun and Ge \cite{Yubosun2025IT} subsequently studied codes correcting these errors. The objective is to construct a single code that protects against both synchronization errors and length-preserving substitutions.

For a single binary edit, redundancy $\log n+O(1)$ is achievable \cite{Levenshtein1966}. For arbitrary fixed alphabet size $q\ge2$, Gabrys \emph{et al.} \cite{Ryan2023IT} obtained single-edit correcting codes with redundancy $\log n+O_q(\log\log n)$. For fixed $t\ge1$ and fixed $q\ge2$ that is a power of two, Sun and Ge \cite[Theorem~12]{Yubosun2025IT} constructed $q$-ary $t$-burst-edit correcting codes with redundancy $\log n+\log\log n+O_{q,t}(1)$. This redundancy is within an additive $\log\log n+O_{q,t}(1)$ term of the $\log n-O_{q,t}(1)$ lower bound inherited from $t$-burst-deletion correcting codes.

Writing $(a,b)$ for the lengths of the removed and replacement substrings, the three operations above correspond to $(t,0)$, $(0,t)$, and $(t,t)$. Thus, every $t$-substring edit correcting code also corrects a $t$-burst-edit. The two error models coincide when $t=1$. For $t\ge2$, however, the substring edit model also allows pairs $(a,b)$ with $1\le a,b\le t$ and $a\ne b$, as well as bursts of fewer than $t$ deletions or insertions. Correcting a $t$-burst-edit alone therefore does not guarantee correction of all these additional errors.

\subsubsection{\textbf{Single Deletion or Single Adjacent Transposition}}
An adjacent transposition exchanges two neighboring symbols, replacing $ab$ by $ba$, where $a,b\in\Sigma_q$ and $a\ne b$. Motivated by DNA-based storage, Gabrys, Yaakobi, and Milenkovic \cite{Gabrys2018Damerau} initiated the study of codes correcting deletions and adjacent transpositions. A basic case requires a single code to correct either one deletion or one adjacent transposition, thereby providing protection against both symbol loss and a local change in symbol order.

For the binary alphabet, Gabrys \emph{et al.} \cite[Theorem~2]{Ryan2023IT} constructed codes correcting either one deletion or one adjacent transposition with redundancy $\log n+O(\log\log n)$, which is within an additive $O(\log\log n)$ of the lower bound $\log n-O(1)$. For arbitrary fixed $q\ge2$, Khuat and Kim \cite{KhuatKim2023} constructed $q$-ary codes for the stronger model allowing one insertion, deletion, substitution, or adjacent transposition, with redundancy $2\log n+O_q(1)$ bits.

A single deletion replaces a length-one substring by the empty sequence, whereas an adjacent transposition replaces a length-two substring by another sequence of length two. Both are $2$-substring edits. Consequently, taking $k=2$ in our first construction gives $q$-ary codes for this model, and also for the stronger one-edit-or-transposition model.

\subsection{Our Contributions and Main Ideas}
This paper continues the study of $k$-substring edit correcting codes. For fixed $q\ge2$ and $k\ge1$ and sufficiently large code length $n$, we obtain the following two constructions.
\begin{itemize}
\item We construct $q$-ary $k$-substring edit correcting codes with redundancy at most $\log n+8\log\log n+8\log\log\log n+O_{q,k}(1)$. This extends the leading redundancy term $\log n$ achieved by the binary construction in \cite{YutingLi2025IT} to arbitrary fixed alphabets. Moreover, the coefficient of $\log\log n$ is reduced from $16k$ to $8$ and is independent of both $q$ and $k$. 

\item When the removed substring and its replacement have different lengths, we construct $q$-ary codes correcting one such edit with redundancy $\log n+O_{q,k}(1)$.
\end{itemize}

These constructions also yield codes for several related error models mentioned in \Cref{sec_relatedmodel}. Applying \Cref{thm_codefine} with $k=2$ gives $q$-ary codes correcting one insertion, deletion, substitution, or adjacent transposition with redundancy $\log n+8\log\log n+o(\log\log n)$ for every fixed $q\ge2$.
Applying \Cref{thm_uneqfinalcode} with $k=T$ gives $q$-ary $(t,T)$-localized deletion correcting codes with redundancy $\log n+O_{q,T}(1)$ for all fixed $q\ge2$ and $1\le t\le T$. Taking $T=t$ gives, in particular, $(\le t)$-burst-deletion correcting codes with redundancy $\log n+O_{q,t}(1)$. For both deletion models, the redundancy is optimal up to an additive constant, since correcting even a single deletion requires $\log n-O_q(1)$ bits of redundancy.

We next outline the first construction. An overview of the second construction is given in \Cref{sec_outlineuneq}.

The first construction builds on the partition and integer-valued mapping of Li \emph{et al} \cite{YutingLi2025IT}. Their method partitions a binary sequence into segments at occurrences of a prescribed pattern and maps each segment to an integer. It first locates the induced edit within $O_k(\log n)$ consecutive entries of the resulting integer sequence. Since each entry represents a segment of length $O_k(\log n)$, this locates the edit in the original sequence within an interval of $O_k((\log n)^2)$ symbols.

We extend the partition and mapping to $q$-ary sequences and sharpen the localization bound by accounting for the actual lengths of the segments. When the edit changes the number of segments, we estimate the VT-syndrome difference associated with each possible edit location. As this location moves across a segment, the estimate changes monotonically by at least a fixed positive multiple of the segment length. The bound on the estimation error then implies that the segments between candidate edit locations have total length $O_{q,k}(\log n)$. When the number of segments is unchanged, we show that the induced edit changes at most two adjacent entries of the integer sequence. The VT syndrome, together with two additional modular sums, identifies the first changed entry exactly. In either case, we obtain an interval of $O_{q,k}(\log n)$ symbols covering the edit region. The general localization framework is developed in \Cref{sec_framework}.

To complete recovery, we adapt the distributed graph coloring method of Li, Gabrys, and Farnoud \cite{LiGabrysFarnoud2026IT} to obtain a local syndrome requiring $4\log m+4\log\log m+O_{q,k}(1)$ bits for a block of length $m$. We use two partitions into blocks of length $m=\Theta_{q,k}(\log n)$, with their boundaries shifted by half a block. This ensures that the localized interval obtained above lies entirely within one block of at least one partition. For each partition, we store modular sums of the components of the block syndromes. Since the other blocks in the selected partition can be read from the received sequence, their contributions can be subtracted to recover the syndrome of the affected block. The two sets of sums, together with the $\log n+O_{q,k}(1)$ bits used for localization, yield the stated redundancy bound.

\subsection{Organization}
\Cref{sec_pre} defines the notation and error model, and establishes bounds on code size. \Cref{sec_framework} develops the general localization framework, and \Cref{sec_instantiation} gives its $q$-ary instantiation. \Cref{sec_substringeditcode} supplies the local recovery component and the first code construction. \Cref{sec_unequal} presents the unequal-length construction and its application to localized/burst deletions. \Cref{sec_conlcusion} concludes the paper.

\section{Preliminaries}\label{sec_pre}
For any $q\ge2$, denote by $\Sigma_q$ the $q$-ary alphabet $\{0,1,\ldots,q-1\}$. We write $\mathbb Z$ for the set of integers and $\mathbb{Z}_{\ge0}$ for the set of nonnegative integers. For any alphabet $\Sigma$, sequences over $\Sigma$ are written in bold lowercase letters, and symbols in $\Sigma$ are written in plain lowercase letters. The length of a sequence $\bmx$, denoted by $\abs{\bmx}$, is the number of symbols in $\bmx$. Note that the empty sequence, denoted by $\epsilon$, is the unique sequence of length $0$.

The notation $|\cdot|$ is also used for the cardinality of a finite set and the absolute value of a real number. Its meaning will always be clear from the argument.
For $n\in\mathbb{Z}_{\ge0}$, let $\Sigma^n$ be the set of all sequences of length $n$ over $\Sigma$. Define $\Sigma^{*}\triangleq\cup_{n\ge0}\Sigma^{n}$. By convention, the $i$-th symbol of sequence $\bmx\in\Sigma^n$ is denoted by $x_i$, for any $1\le i\le n$. Thus, the sequence $\bmx$ of length $n$ can be written as $x_1\cdots x_n$ or $\parenv{x_1,\ldots,x_n}$. The concatenation of two sequences $\bmx$ and $\bmy$ over the same alphabet is denoted by $\bmx\bmy$. For any sequence $\bmx$ over $\mathbb{Z}_{\ge0}$ or $\Sigma_q$, define
\begin{equation}\label{eq_L1VT}
    \lwt{\bmx}\triangleq\sum_{i=1}^{\abs{\bmx}} x_i,\quad \vt(\bmx)\triangleq\sum_{i=1}^{\abs{\bmx}} i\cdot x_i.
\end{equation}
We call $\lwt{\bmx}$ the $L_1$-weight of $\bmx$, and $\vt(\bmx)$ the VT syndrome of $\bmx$.

For integers $a$ and $b$, define $[a,b]=\mathset{i\in\mathbb{Z}:a\le i\le b}$, with $[a,b]=\varnothing$ when $a>b$. For $n\in\mathbb{Z}_{\ge0}$, abbreviate $[1,n]$ as $[n]$. Let $\bmx$ be a sequence of length $n$. For a set $I=\mathset{i_1,i_2,\ldots,i_r}$ with $1\le i_1<i_2<\cdots<i_r\le n$, denote $\bmx_{I}=x_{i_1}x_{i_2}\cdots x_{i_r}$. We set $\bmx_I=\epsilon$ when $I$ is empty. In particular, if $I=[i,j]$ for some $1\le i\le j\le n$, the sequence $\bmx_I$ is called a \emph{substring} of $\bmx$.

\begin{definition}[$k$-substring edit]\label{def_substringedit}
Let $\Sigma$ be an alphabet and $\bmx,\bmy\in\Sigma^*$. We say that $\bmy$ is obtained from $\bmx$ by a \emph{$k$-substring edit} if
\begin{equation}\label{eq_editbasic}
 \bmx=\bmu\bms\bmv,
 \quad
 \bmy=\bmu\bmt\bmv,
\end{equation}
for some sequences $\bmu,\bms,\bmt,\bmv\in\Sigma^{*}$ with $\abs{\bms},\abs{\bmt}\le k$ and $\abs{\bms}+\abs{\bmt}>0$.
In other words, the edit replaces $\bms$ by $\bmt$. We write $\bmx\to_k\bmy$ whenever such a representation exists, and call the edit \emph{nontrivial} if
$\bmx\ne\bmy$.
\end{definition}

In the above definition, if $\bms=\epsilon$, the edit inserts $\bmt$ between $\bmu$ and $\bmv$. If $\bmt=\epsilon$, the edit deletes $\bms$.

When $\Sigma=\Sigma_q$ and the $k$-substring edit $\bmx\to_k\bmy$ is nontrivial, we always choose the representation in \eqref{eq_editbasic} so that the replaced substring is minimal. Concretely, when both $\bms$
and $\bmt$ are nonempty, we always assume that the first symbols of $\bms$ and $\bmt$ are different, and that the last symbols of $\bms$ and $\bmt$ are different, i.e.,
\begin{equation}\label{eq_minimalends}
 s_1\ne t_1,
 \quad
 s_{\abs{\bms}}\ne t_{\abs{\bmt}}.
\end{equation}
Pure insertions (when $\bms=\epsilon$) and pure deletions (when $\bmt=\epsilon$) already have this form.

For $\bmx\in\Sigma^*$, let $\cE_{b,k}(\bmx)$ be the set of all sequences obtainable from $\bmx$ by at most $b$ $k$-substring edits.

\begin{definition}
A nonempty set $\cC\subseteq\Sigma_q^n$ is a
\emph{$(b,k)$-substring-edit correcting code} if
$\cE_{b,k}(\bmx)\cap\cE_{b,k}(\bmy)=\varnothing$ for all distinct $\bmx,\bmy\in\cC$.  When $b=1$, we simply call $\cC$ a \emph{$k$-substring-edit correcting code}.
\end{definition}

The bit redundancy of a code $\cC\subseteq\Sigma_q^n$ is
$\red(\cC)\triangleq\log(q^n/|\cC|)$, where $\log$ denotes the base-$2$ logarithm. A subscript on $O(\cdot)$ or $\Omega(\cdot)$ lists the fixed parameters on which the implicit constants may depend.

Regarding the maximum size (or the minimum redundancy) of a $(b,k)$-substring-edit correcting code, we have the next theorem.
\begin{theorem}
    Let $q\ge2$, $b\ge1$, and $1\le k\le n$ be integers. Set $\tau=\min\mathset{b,\floorenv{n/k}}$ and
    \begin{equation*}
        V_{q,k,b}(n)=\sum_{j=0}^{\tau}\binom{n-(k-1)j}{j}\parenv{(q-1)q^{k-1}}^j.
    \end{equation*}
    Let $\cC\subseteq\Sigma_q^n$ be a $(b,k)$-substring edit correcting code with maximum size. Then 
    \begin{equation}\label{eq_uplowbound}
        \frac{q^{n-2kb}}{(n+kb)^{2b}(k+1)^{4b}}\le\abs{\cC}\le\min\mathset{q^{\max\mathset{0,n-2kb}},\frac{q^n}{V_{q,k,b}(n)}}.
    \end{equation}
    Equivalently, it holds that
    \begin{equation*}
       \max\mathset{\min\mathset{n,2kb}\log q,\log V_{q,k,b}(n)}\le\red(\cC)\le 2b\log(n+kb)+4b\log(k+1)+2kb\log q.
    \end{equation*}
    Suppose that $q$ and $b$ are fixed. When $n\ge2kb$, we have
    \begin{equation}\label{eq_lowredfixedqandb}
        \red(\cC)\ge\max\mathset{2kb\log q,b\log n+kb\log q}-O_{q,b}(1).
    \end{equation}
\end{theorem}
\begin{IEEEproof}
The lower bound in \eqref{eq_uplowbound} is given in \cite[Lemma 2]{YuanyuanTang2025TCOMM}. We only have to prove the upper bound.

We first prove the first upper bound in \eqref{eq_uplowbound}. Put $m=\min\mathset{n,2kb}$, and fix an interval $I=[i,i+m-1]\subseteq[n]$. Let
$r=\lceil m/k\rceil$. Then $r\le2b$. For $j\in[r-1]$, denote $I_j=\sparenv{i+(j-1)k,i+jk-1}$. Denote $I_r=\sparenv{i+(r-1)k,i+m-1}$. Clearly, $I_1,\ldots,I_r$ are sub-intervals of $I$ of length at most $k$. Moreover, they form a partition of $I$.

Suppose that codewords $\bmx,\bmy\in\cC$ agree outside $I$, i.e., $\bmx_{[n]\setminus I}=\bmy_{[n]\setminus I}$. Form $\bmz\in\Sigma_q^n$ as follows:
\begin{equation*}
    \begin{array}{c}
       \bmz_{[n]\setminus I}=\bmx_{[n]\setminus I}=\bmy_{[n]\setminus I},\\
        \bmz_{I_j}=\bmx_{I_j}\text{ for }j\le\floorenv{r/2},\\
        \bmz_{I_j}=\bmy_{I_j}\text{ for }j>\floorenv{r/2}.
    \end{array}
\end{equation*}
Since $r\le 2b$, we have $\floorenv{r/2}, r-\floorenv{r/2}\le b$. Then it follows that $\bmz\in\cE_{b,k}(\bmx)\cap\cE_{b,k}(\bmy)$. This implies $\bmx=\bmy$ because $\cC$ is a $(b,k)$-substring-edit correcting code. Therefore, it holds that $\bmx_{[n]\setminus I}\ne\bmy_{[n]\setminus I}$ for any distinct $\bmx,\bmy\in\cC$, which gives
$\abs{\cC}\le q^{n-m}=q^{\max\mathset{n-2kb,0}}$.

We next prove the second upper bound. Since $\cC$ is a $(b,k)$-substring-edit correcting code, by the sphere-packing argument, it is sufficient to establish the following lower bound on the size of $\cE_{b,k}(\bmx)\cap\Sigma_q^n$.
\begin{clm}
    For every $\bmx\in\Sigma_q^n$, it holds that $\abs{\cE_{b,k}(\bmx)\cap\Sigma_q^n}\ge V_{q,k,b}(n)$.
\end{clm}
\noindent\textit{Proof of Claim:}
It is clear that $\bmx\in\cE_{b,k}(\bmx)\cap\Sigma_q^n$. We next estimate the size of $\cE_{b,k}(\bmx)\cap\Sigma_q^n\setminus\{\bmx\}$.

For each integer $j\in[\tau]$, choose $j$ pairwise disjoint intervals
\begin{equation}\label{eq_constraintofI}
 I_h=[i_h,i_h+k-1],\quad
 1\le i_1<\cdots<i_j\le n-k+1,
 \quad i_{h+1}\ge i_h+k\quad(1\le h<j).
\end{equation}
For each choice of the tuple $\parenv{I_1,\ldots,I_j}$, define $j+1$ integers $a_1,a_2,\ldots,a_{j+1}$ as follows:
\begin{equation}\label{eq_itoa}
    \begin{array}{c}
        a_1=i_1-1,\\
        a_h=i_h-i_{h-1},\text{ for every }2\le h\le j,\\
        a_{j+1}=n+1-i_j.
    \end{array}
\end{equation}
Clearly, the tuple $\parenv{a_1,\ldots,a_{j+1}}$ satisfies
\begin{equation}\label{eq_constraintofa}
   \begin{array}{c}
        a_1+a_2+\cdots+a_{j+1}=n,\\
        a_1\ge0,\\
        a_h\ge k,\text{ for every }2\le h\le j+1.
   \end{array}
\end{equation}
On the other hand, it follows from \eqref{eq_itoa} that each tuple $\parenv{a_1,\ldots,a_{j+1}}\in\mathbb{Z}^{j+1}$ satisfying \eqref{eq_constraintofa} uniquely determines a tuple $\parenv{I_1,\ldots,I_j}$ satisfying \eqref{eq_constraintofI}. Therefore, for each $j\in[\tau]$, the number of the tuples $\parenv{I_1,\ldots,I_j}$ satisfying \eqref{eq_constraintofI} equals the number of tuples $\parenv{a_1,\ldots,a_{j+1}}\in\mathbb{Z}^{j+1}$ satisfying \eqref{eq_constraintofa}, which is
$\binom{n-(k-1)j}{j}$. 

For every $j\in[\tau]$ and every tuple $\parenv{I_1,\ldots,I_j}$ that satisfies \eqref{eq_constraintofI}, form a sequence $\bmz\in\Sigma_q^n$ from $\bmx$ as follows: on each $I_h$, replace $\bmx_{I_h}$ with an arbitrary $\bmt^{(h)}\in\Sigma_q^k$ satisfying $t^{(h)}_1\ne x_{i_h}$. There are $(q-1)q^{k-1}$ choices for each replacement. By construction, sequences $\bmz$ and $\bmx$ differ at at least $j$ coordinates.

Distinct tuples $\parenv{j,I_1,\ldots,I_j,\bmt^{(1)},\ldots,\bmt^{(j)}}$ produce distinct $\bmz$. Indeed, for an output $\bmz$,
let $D(\bmx,\bmz)=\mathset{i\in[n]:x_i\ne z_i}$, which is not empty because $\bmx\ne\bmz$. Then $i_1=\min D(\bmx,\bmz)$ and thus, the interval $I_1$ is determined. After removing the elements of
$D(\bmx,\bmz)$ that lie in $I_1$, the smallest remaining element, if one exists, is $i_2$. This gives $I_2$.  Repeating this procedure recovers $j$ and all the intervals $I_1,\ldots,I_j$. Finally, we have $\parenv{\bmt^{(1)},\ldots,\bmt^{(j)}}=\parenv{\bmz_{I_1},\ldots,\bmz_{I_j}}$.

The preceding argument shows that $\abs{\cE_{b,k}(\bmx)\cap\Sigma_q^n\setminus\{\bmx\}}$ is at least the number of tuples $\parenv{j,I_1,\ldots,I_j,\bmt^{(1)},\ldots,\bmt^{(j)}}$, which is $\sum_{j=1}^{\tau}\binom{n-(k-1)j}{j}\parenv{(q-1)q^{k-1}}^j=V_{q,k,b}(n)-1$. This completes the proof of the claim.
\hfill$\square$

Finally, if $n\ge 2kb$, the term with $j=b$ gives
\begin{align*}
 V_{q,k,b}(n)&\ge
 \binom{n-(k-1)b}{b}\parenv{(q-1)q^{k-1}}^{b}\\
 &\ge\parenv{\frac{n-(k-1)b}{b}}^b\parenv{(q-1)q^{k-1}}^b\\
 &\ge\parenv{\frac{n}{2b}}^b\parenv{(q-1)q^{k-1}}^b.
\end{align*}
Now \eqref{eq_lowredfixedqandb} follows.
\end{IEEEproof}

\bigskip
Throughout the rest of this paper, parameters $q$ and $k$ are assumed to be fixed integers which are independent of the code-length $n$. Our objective is to construct $k$-substring edit correcting codes with small redundancy.

The following two definitions will make the exposition in this paper clearer and more precise.
\begin{definition}[edit region]\label{def_editregion}
    Let $\bmx$ and $\bmy$ be represented as in \eqref{eq_editbasic} and \eqref{eq_minimalends}. Set $i=\abs{\bmu}$. If $\abs{\bms}>0$, the \emph{edit region} in $\bmx$ is the interval $[i+1,i+\abs{\bms}]$. If $\abs{\bms}=0$, it must be that $\abs{\bmt}>0$. In this case, the edit region in $\bmx$ is the \emph{gap} $g=i+1$ between $x_i$ and $x_{i+1}$, with the conventions that $g=1$ precedes the first symbol $x_1$ and $g=\abs{\bmx}+1$ follows the last symbol $x_{\abs{\bmx}}$. Equivalently, if $\abs{\bms}>0$, the edit region is the coordinates occupied by $\bms$; if $\abs{\bms}=0$, the edit region is the position where $\bmt$ is inserted.
    
    The edit region in $\bmy$ is defined analogously using $\bmt$.
\end{definition}

By definition, the edit region in $\bmx$ is a sub-interval of $\sparenv{1,\abs{\bmx}}$ when $\bms$ is not empty, or a gap between coordinates when $\bms$ is empty. The same observation holds for $\bmy$. \Cref{fig_editregion} illustrates \Cref{def_editregion}.

\begin{figure}[!t]
\centering
\begin{tikzpicture}[x=0.8cm,y=0.8cm,>=Latex,font=\small]
  \begin{scope}
    \node[anchor=west] at (-0.5,1.75) {(a) $\bms\ne\epsilon$, $\bmt\ne\epsilon$};
    \node[left] at (0,0.9) {$\bmx$};
    \node[left] at (0,-1.05) {$\bmy$};
    \draw (0,0.6) rectangle (1.5,1.2);
    \draw[fill=red!12] (1.5,0.6) rectangle (3,1.2);
    \draw (3,0.6) rectangle (4.5,1.2);
    \node at (0.75,0.9) {$\bmu$};
    \node at (2.25,0.9) {$\bms$};
    \node at (3.75,0.9) {$\bmv$};
    \draw[decorate,decoration={brace,mirror,amplitude=3pt}] (1.5,0.48)--(3,0.48);
    \node[below] at (2.25,0.35) {$[i+1,i+|\bms|]$};
    \draw (0,-1.35) rectangle (1.5,-0.75);
    \draw[fill=red!12] (1.5,-1.35) rectangle (3.5,-0.75);
    \draw (3.5,-1.35) rectangle (5,-0.75);
    \node at (0.75,-1.05) {$\bmu$};
    \node at (2.5,-1.05) {$\bmt$};
    \node at (4.25,-1.05) {$\bmv$};
    \draw[decorate,decoration={brace,mirror,amplitude=3pt}] (1.5,-1.47)--(3.5,-1.47);
    \node[below] at (2.5,-1.6) {$[i+1,i+|\bmt|]$};
  \end{scope}
  \begin{scope}[xshift=5.15cm]
    \node[anchor=west] at (-0.5,1.75) {(b) $\bms=\epsilon$, $\bmt\ne\epsilon$};
    \node[left] at (0,0.9) {$\bmx$};
    \node[left] at (0,-1.05) {$\bmy$};
    \draw (0,0.6) rectangle (1.5,1.2);
    \draw (1.5,0.6) rectangle (3,1.2);
    \node at (0.75,0.9) {$\bmu$};
    \node at (2.25,0.9) {$\bmv$};
    \draw[very thick] (1.5,0.55)--(1.5,1.25);
    \draw[->] (1.5,0.05)--(1.5,0.51);
    \node[below] at (1.5,0.03) {gap $g=i+1$};
    \draw (0,-1.35) rectangle (1.5,-0.75);
    \draw[fill=red!12] (1.5,-1.35) rectangle (3.5,-0.75);
    \draw (3.5,-1.35) rectangle (5,-0.75);
    \node at (0.75,-1.05) {$\bmu$};
    \node at (2.5,-1.05) {$\bmt$};
    \node at (4.25,-1.05) {$\bmv$};
    \draw[decorate,decoration={brace,mirror,amplitude=3pt}] (1.5,-1.47)--(3.5,-1.47);
    \node[below] at (2.5,-1.6) {$[i+1,i+|\bmt|]$};
  \end{scope}
  \begin{scope}[xshift=10.3cm]
    \node[anchor=west] at (-0.5,1.75) {(c) $\bms\ne\epsilon$, $\bmt=\epsilon$};
    \node[left] at (0,0.9) {$\bmx$};
    \node[left] at (0,-1.05) {$\bmy$};
    \draw (0,0.6) rectangle (1.5,1.2);
    \draw[fill=red!12] (1.5,0.6) rectangle (3,1.2);
    \draw (3,0.6) rectangle (4.5,1.2);
    \node at (0.75,0.9) {$\bmu$};
    \node at (2.25,0.9) {$\bms$};
    \node at (3.75,0.9) {$\bmv$};
    \draw[decorate,decoration={brace,mirror,amplitude=3pt}] (1.5,0.48)--(3,0.48);
    \node[below] at (2.25,0.35) {$[i+1,i+|\bms|]$};
    \draw (0,-1.35) rectangle (1.5,-0.75);
    \draw (1.5,-1.35) rectangle (3,-0.75);
    \node at (0.75,-1.05) {$\bmu$};
    \node at (2.25,-1.05) {$\bmv$};
    \draw[very thick] (1.5,-1.4)--(1.5,-0.7);
    \draw[->] (1.5,-1.93)--(1.5,-1.44);
    \node[below] at (1.5,-1.95) {gap $g=i+1$};
  \end{scope}
\end{tikzpicture}
\caption{Edit regions for the specified decomposition $\bmx=\bmu\bms\bmv$, $\bmy=\bmu\bmt\bmv$, where $i=|\bmu|$: (a) replacement of a nonempty substring by a nonempty sequence; (b) pure insertion; (c) pure deletion. In (c), the interval occupied by $\bms$ in $\bmx$ becomes a gap in $\bmy$.}
\label{fig_editregion}
\end{figure}
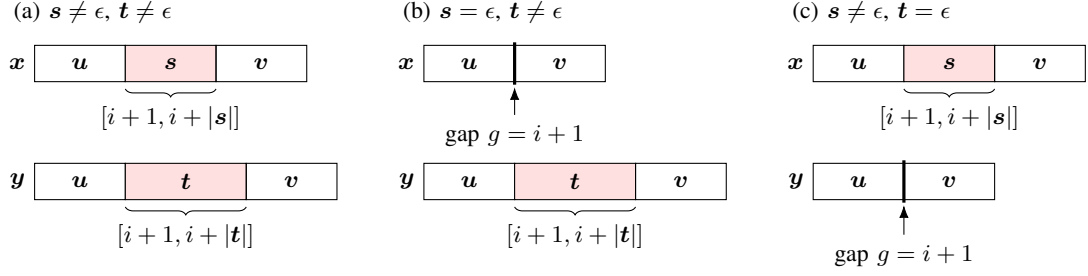
\begin{definition}[covering an edit region]\label{def_covereditregion}
    For a coordinate interval $J=[a,b]\subseteq[1,\abs{\bmx}]$, we say that $J$ or $\bmx_J$ \emph{covers} the edit region in $\bmx$ if $[i+1,i+\abs{\bms}]\subseteq J$ when $\abs{\bms}>0$, and if $a\le i+1\le b+1$ when $\abs{\bms}=0$. Coverage in $\bmy$ is defined in the same way.
\end{definition}
Notice that the above definition allows a gap at either endpoint of a substring to be regarded as covered by that substring.

\begin{example}
Here we give some examples to illustrate \Cref{def_editregion,def_covereditregion}. Each edit region refers to the specified decomposition in \eqref{eq_editbasic}.
\begin{enumerate}[$(i)$]
\item Let $\bmu=01$, $\bms=10$, $\bmt=01$, and $\bmv=01$. Replacing $\bms$ by $\bmt$ changes $\bmx=011001$ to $\bmy=010101$. Both edit regions are $[3,4]$. The interval $[2,5]$ covers either region, whereas $[3,3]$ does not.
\item Let $\bmu=01$, $\bms=10$, $\bmt=\epsilon$, and $\bmv=01$. Deleting $\bms$ changes $\bmx=011001$ to $\bmy=0101$. The edit region in $\bmx$ is $[3,4]$, and the edit region in $\bmy$ is the gap $g=3$. In $\bmy$, each of $[1,2]$, $[3,4]$ and $[2,3]$ covers this gap. The interval $[1,1]$ does not cover it.
\item Reversing the edit in (ii) inserts $10$ into $\bmx=0101$ and gives $\bmy=011001$. The edit region in $\bmx$ is now the gap $g=3$, and that in $\bmy$ is $[3,4]$.
\item For $\bmx=01$, insertion before the first symbol has gap $g=1$, covered by $[1,1]$. Insertion after the last symbol has gap $g=3$, covered by $[2,2]$. Thus endpoint insertions are included in \Cref{def_covereditregion}.
\end{enumerate}
\end{example}

In this paper, we measure computational complexity in bit operations. A bit operation reads, writes, or applies a Boolean operation to a fixed number of bits. The following standard bounds for elementary integer arithmetic will be used in our syndrome-computation and decoding analyses. For an integer $a$, define its binary length by $\operatorname{len}(0)=1$ and $\operatorname{len}(a)=\floorenv{\log|a|}+1$ when $a\ne0$.
\begin{proposition}\cite[Theorem 3.3]{Shoup2008}\label{prop_bitarithmetic}
Let $a,b\in\mathbb Z$. Then the following computations can be performed deterministically.
\begin{enumerate}[$(i)$]
\item Addition and subtraction, $a\pm b$, require $O\parenv{\operatorname{len}(a)+\operatorname{len}(b)}$ bit operations.
\item Multiplication, $ab$, requires $O\parenv{\operatorname{len}(a)\operatorname{len}(b)}$ bit operations.
\item If $b>0$, the quotient $c=\floorenv{a/b}$ and the remainder $r=a-bc\in[0,b-1]$ can be computed in $O\parenv{\operatorname{len}(b)\operatorname{len}(c)}$ bit operations. In particular, this is $O\parenv{\operatorname{len}(a)\operatorname{len}(b)}$.
\end{enumerate}
\end{proposition}
Thus, for integers represented using at most $\lambda$ bits, addition, subtraction, and comparison cost $O(\lambda)$ bit operations, whereas multiplication and division with remainder cost $O(\lambda^2)$ bit operations.

\section{A General Framework for Locating One \texorpdfstring{$k$}{k}-Substring Edit Error}\label{sec_framework}
Let $\bmx\in\Sigma_q^n$ and  $\bmy$ be obtained from $\bmx$ by a nontrivial $k$-substring edit. Represent $\bmx$ and $\bmy$ as in \eqref{eq_editbasic} with the convention in \eqref{eq_minimalends}. This section establishes sufficient conditions for determining an interval $I\subseteq[n]$ with $|I|=O_{q,k}(\log n)$ that covers the edit region in $\bmx$. The main result is presented in \Cref{thm_framework}.

Let $\bmx_{[i_1,j_1]}$ and $\bmx_{[i_2,j_2]}$ be two substrings of $\bmx$. We say that they are \emph{overlapping} (or $\bmx_{[i_1,j_1]}$ overlaps with $\bmx_{[i_2,j_2]}$) if $\sparenv{i_1,j_1}\cap\sparenv{i_2,j_2}\ne\varnothing$ and \emph{non-overlapping} otherwise. 

A \emph{partition rule} $\cP$ is a rule under which we divide every nonempty sequence $\bma$ into several non-overlapping nonempty substrings whose concatenation is $\bma$. Each of these substrings is called a \emph{segment} of $\bma$ under the partition rule. Let $m_{\cP}(\bma)$ denote the number of segments of $\bma$ under partition rule $\cP$. Denote
\begin{equation*}
    \Part_{\cP}(\bma)=\parenv{\bma_1^{\cP},\bma_2^{\cP},\ldots,\bma_{m_{\cP}(\bma)}^{\cP}},
\end{equation*}
where $\bma_i^{\cP}$ is the $i$-th segment of $\bma$. By definition,  we have $\bma=\bma_1^{\cP}\bma_2^{\cP}\cdots\bma_{m_{\cP}(\bma)}^{\cP}$.

For a function $f:\Sigma_q^{*}\to\mathbb{Z}_{\ge0}$, define 
\begin{equation*}
    \bmF_{\cP}(\bma)=\parenv{f\parenv{\bma_1^{\cP}},f\parenv{\bma_2^{\cP}},\ldots,f\parenv{\bma_{m_{\cP}(\bma)}^{\cP}}}.
\end{equation*}
In other words, $\bmF_{\cP}(\bma)$ is the sequence over $\mathbb{Z}_{\ge0}$ obtained by applying $f$ to every segment of $\bma$.

Let $\bmx$ and $\bmy$ be given at the beginning of this section. We further make the following assumption throughout this section.
\begin{assumption}\label{assump1}
    There exist integers $p_1,p_2\ge0$ and $1\le l_1,l_2\le K$ such that
 \begin{equation*}
 \begin{aligned}
  \Part_{\cP}(\bmx)
   &=\parenv{\bme_1,\ldots,\bme_{p_1},
            \bma_1,\ldots,\bma_{l_1},
            \bmh_1,\ldots,\bmh_{p_2}},\\
  \Part_{\cP}(\bmy)
   &=\parenv{\bme_1,\ldots,\bme_{p_1},
            \bmb_1,\ldots,\bmb_{l_2},
            \bmh_1,\ldots,\bmh_{p_2}}.
 \end{aligned}
 \end{equation*}
 The edit region in $\bmx$ is covered by the substring
 $\bma_1\cdots\bma_{l_1}$, and the edit region in $\bmy$ is covered by the substring $\bmb_1\cdots\bmb_{l_2}$.
\end{assumption}

Under this assumption, denote
\begin{equation*}
    \begin{aligned}
        \bmu^\prime=\parenv{f\parenv{\bme_1},\ldots,f\parenv{\bme_{p_1}}},\;\bms^\prime=\parenv{f\parenv{\bma_1},\ldots,f\parenv{\bma_{l_1}}},\\
        \bmv^\prime=\parenv{f\parenv{\bmh_1},\ldots,f\parenv{\bmh_{p_2}}},\; \bmt^\prime=\parenv{f\parenv{\bmb_1},\ldots,f\parenv{\bmb_{l_2}}}.
    \end{aligned}
\end{equation*}
Then it holds that
\begin{equation}\label{eq_divideFxy}
    \bmF_{\cP}(\bmx)=\bmu^\prime\bms^\prime\bmv^\prime,\quad \bmF_{\cP}(\bmy)=\bmu^\prime\bmt^\prime\bmv^\prime.
\end{equation}
Clearly, $\bmF_{\cP}(\bmy)$ is obtained from $\bmF_{\cP}(\bmx)$ by a $K$-substring edit. To locate the $k$-substring edit error in $\bmx$, the idea is to locate the $K$-substring edit error in $\bmF_{\cP}(\bmx)$.

Notice that both $\bmF_{\cP}(\bmx)$ and $\bmF_{\cP}(\bmy)$ are sequences over the alphabet $\mathbb{Z}_{\ge0}$. We now study the induced edit independently of its underlying partition. In the rest of this section, assume that $\bmz\in\mathbb{Z}_{\ge0}^{*}$ and that $\bmw$ is obtained from $\bmz$ by a nontrivial $K$-substring edit for some $K\ge2$. Denote
\begin{equation}\label{eq_dEDelta}
    d=\abs{\bmw}-\abs{\bmz},\; E=\lwt{\bmz}-\lwt{\bmw},\; \Delta_{\rm VT}=\vt(\bmz)-\vt(\bmw),
\end{equation}
where the two functions $\lwt{\cdot}$ and $\vt(\cdot)$ are defined in \eqref{eq_L1VT}.

We first consider the case $\abs{\bmz}=\abs{\bmw}$. For $r\in\{0,1,2\}$, define the function $G_r:\mathbb{Z}_{\ge0}^{*}\to\mathbb{Z}_{\ge0}$ by
\begin{equation}\label{eq_functionG}
    G_{r}(\bma)=\sum_{\substack{i\in\sparenv{\abs{\bma}}\\i\equiv r\pmod{3}}}a_i.
\end{equation}
\begin{proposition}\label{prop_equallength}
    Let $d$, $E$ and $\Delta_{\rm VT}$ be defined in \eqref{eq_dEDelta}. Suppose that $d=0$, and that $\bmz$ and $\bmw$ differ in one coordinate or two consecutive coordinates. For $0\le r\le 2$, set $c_r=G_r(\bmz)-G_r(\bmw)$. If $E\ne0$, and all the values of $E$, $\Delta_{\rm VT}$, $c_0$ and $c_1$ are known, we can determine exactly the first index at which $\bmz$ and $\bmw$ differ.
\end{proposition}
\begin{IEEEproof}
    It is clear that $c_0+c_1+c_2=E$. Thus, we can obtain $c_2$ when $c_0$, $c_1$ and $E$ are given. Since $\bmz$ and $\bmw$ differ in one or two consecutive coordinates, exactly one or two of $c_0$, $c_1$ and $c_2$ are nonzero.

    If only one of $c_0$, $c_1$ and $c_2$ is nonzero, there is a unique $j$ such that $z_j\ne w_j$. The definitions of $E$ and $\Delta_{\rm VT}$ give $E=z_j-w_j$ and $\Delta_{\rm VT}=j(z_j-w_j)$. Therefore, $j=\Delta_{\rm VT}/E$.

    Now suppose that exactly two of $c_0$, $c_1$ and $c_2$ are nonzero. Then there are two indices $j$ and $j+1$ at which $\bmz$ and $\bmw$ differ. Let $r$ be the unique element of $\{0,1,2\}$ satisfying $c_r=0$. It follows that $j\equiv r+1\pmod{3}$ and $z_{j+1}-w_{j+1}=c_{(r+2)\bmod 3}$. Moreover, 
    \begin{align*}
        \Delta_{\rm VT}&=j(z_j-w_j)+(j+1)\parenv{z_{j+1}-w_{j+1}}\\
        &=j\parenv{z_j-w_j+z_{j+1}-w_{j+1}}+z_{j+1}-w_{j+1}\\
        &=j\cdot E+z_{j+1}-w_{j+1}.
    \end{align*}
    Consequently, it holds that $j=\parenv{\Delta_{\rm VT}-c_{(r+2)\bmod 3}}/E$. This completes the proof.
\end{IEEEproof}

Next we consider the case $\abs{\bmz}\ne\abs{\bmw}$. Let $d$ and $E$ be given in \eqref{eq_dEDelta}. Since $\bmz\to_{K}\bmw$, there are sequences $\bmu^\prime$, $\bmv^\prime$, $\bms^\prime$, and $\bmt^\prime$ with $\abs{\bms^\prime},\abs{\bmt^\prime}\le K$, such that
\begin{equation}\label{eq_splitzw}
    \bmz=\bmu^\prime\bms^\prime\bmv^\prime,\; \bmw=\bmu^\prime\bmt^\prime\bmv^\prime.
\end{equation}
For $r\in\sparenv{0,\abs{\bmw}}$, define
\begin{equation}\label{eq_reta}
    r^{*}=\abs{\bmu^\prime},\quad\eta(r)=-d\sum_{i=r+1}^{\abs{\bmw}}w_i+r\cdot E.
\end{equation}

\Cref{lem_monotonicity,lem_Deltaetagap} appeared in the proof of \cite[Lemma 4]{YutingLi2025IT}. We present their proofs here for completeness.
\begin{lemma}\label{lem_monotonicity}
    If $d\ne0$ and $\abs{E}\le w_i$ for all $i\in\sparenv{\abs{\bmw}}$, the function $\eta(r)$ is monotone in $r$.
\end{lemma}
\begin{IEEEproof}
    By definition, for $r,r-1\in\sparenv{0,\abs{\bmw}}$, it holds that $\eta(r)-\eta(r-1)=d\cdot w_r+E$. Since $d$ is a non-zero integer and $\abs{E}\le w_i$ for all $i\in\sparenv{\abs{\bmw}}$, we have $\eta(r)-\eta(r-1)\ge 0$ if $d>0$ and $\eta(r)-\eta(r-1)\le 0$ if $d<0$.
\end{IEEEproof}

For a positive integer $A$, a sequence $\bma\in\mathbb Z_{\ge0}^*$ is called \emph{$A$-bounded} if $a_i\le A$ for every $i\in[\abs{\bma}]$.
\begin{lemma}\label{lem_Deltaetagap}
    Let $\Delta_{\rm VT}$ be defined in \eqref{eq_dEDelta}, and $r^{*}$ and $\eta(\cdot)$ be defined in \eqref{eq_reta}. If both $\bmz$ and $\bmw$ are $A$-bounded, then $\abs{\Delta_{\rm VT}-\eta(r^{*})}<2K^2A$.
\end{lemma}
\begin{IEEEproof}
    It follows from \eqref{eq_splitzw} that $\bmz_{\sparenv{1,r^{*}}}=\bmw_{\sparenv{1,r^{*}}}=\bmu^\prime$, $\bmz_{\sparenv{r^{*}+\abs{\bms^\prime}+1,\abs{\bmz}}}=\bmw_{\sparenv{r^{*}+\abs{\bmt^\prime}+1,\abs{\bmw}}}=\bmv^\prime$, $\bmz_{\sparenv{r^{*}+1,r^{*}+\abs{\bms^\prime}}}=\bms^\prime$ and $\bmw_{\sparenv{r^{*}+1,r^{*}+\abs{\bmt^\prime}}}=\bmt^\prime$. This implies that $d=\abs{\bmw}-\abs{\bmz}=\abs{\bmt^\prime}-\abs{\bms^\prime}$ and $E=\lwt{\bmz}-\lwt{\bmw}=\lwt{\bms^\prime}-\lwt{\bmt^\prime}$. Therefore, we have
    \begin{align}
            \Delta_{\rm VT}&=\vt(\bmz)-\vt(\bmw)\notag\\
            &=\sum_{i=r^{*}+1}^{r^{*}+\abs{\bms^\prime}}iz_i+\sum_{i=r^{*}+\abs{\bms^\prime}+1}^{\abs{\bmz}}iz_i-\sum_{i=r^{*}+1}^{r^{*}+\abs{\bmt^\prime}}iw_i-\sum_{i=r^{*}+\abs{\bmt^\prime}+1}^{\abs{\bmw}}iw_i\notag\\
            &=\parenv{\sum_{i=r^{*}+1}^{r^{*}+\abs{\bms^\prime}}iz_i-\sum_{i=r^{*}+1}^{r^{*}+\abs{\bmt^\prime}}iw_i}+\parenv{\sum_{i=r^{*}+\abs{\bms^\prime}+1}^{\abs{\bmz}}iz_i-\sum_{i=r^{*}+\abs{\bmt^\prime}+1}^{\abs{\bmw}}iw_i}\notag\\
            &=\vt(\bms^\prime)-\vt(\bmt^\prime)+r^{*}\parenv{\lwt{\bms^\prime}-\lwt{\bmt^\prime}}-\parenv{\abs{\bmt^\prime}-\abs{\bms^\prime}}\sum_{i=r^{*}+\abs{\bmt^\prime}+1}^{\abs{\bmw}}w_i\label{eq_VT}\\
            &=-d\sum_{i=r^{*}+1}^{\abs{\bmw}}w_i+r^{*}\cdot E+\vt(\bms^\prime)-\vt(\bmt^\prime)+\parenv{\abs{\bmt^\prime}-\abs{\bms^\prime}}\sum_{i=r^{*}+1}^{r^{*}+\abs{\bmt^\prime}}w_i\notag\\
            &=\eta\parenv{r^{*}}+\vt(\bms^\prime)-\vt(\bmt^\prime)+\parenv{\abs{\bmt^\prime}-\abs{\bms^\prime}}\lwt{\bmt^\prime}.\notag
    \end{align}
    Since both $\bmz$ and $\bmw$ are $A$-bounded, $K\ge2$, and $\abs{\bms^\prime},\abs{\bmt^\prime}\le K$, it holds that $\abs{\vt(\bms^\prime)-\vt(\bmt^\prime)}<K^2A$, $\abs{\abs{\bmt^\prime}-\abs{\bms^\prime}}\le K$ and $\lwt{\bmt^\prime}\le KA$. Then it follows that $\abs{\Delta_{\rm VT}-\eta(r^{*})}\le\abs{\vt(\bms^\prime)-\vt(\bmt^\prime)}+\abs{\abs{\bmt^\prime}-\abs{\bms^\prime}}\lwt{\bmt^\prime}<2K^2A$.
\end{IEEEproof}

\vspace{5pt}
We next retain the unequal-length assumption $d\ne0$ and
associate a positive number $\ell_i$ with $w_i$ for every $i$.
\begin{proposition}\label{prop_unequallength}
    Suppose that $\bmz$ and $\bmw$ are $A$-bounded. Let $d$, $E$, $\Delta_{\rm VT}$ be defined by \eqref{eq_dEDelta}. Assume that $d\ne0$ and $\abs{E}\le w_i$ for every $i\in[\abs{\bmw}]$. Let $\ell_1,\ldots,\ell_{\abs{\bmw}}$ be positive numbers. Suppose that there are positive numbers $C^{(1)}$, $C^{(2)}$ and $C^{(3)}$ satisfying:
    \begin{enumerate}[$(i)$]
        \item $w_i\ge C^{(1)}\ell_i$ for every $i\in[\abs{\bmw}]$; 
        \item $\ell_i\ge C^{(2)}$ for every $i\in[\abs{\bmw}]$; 
        \item $\abs{E}\le C^{(3)}$;
        \item $C^{(3)}<C^{(1)}C^{(2)}$.
    \end{enumerate}
    Let $r^{*}$ and $\eta(\cdot)$ be defined as in \eqref{eq_reta}, $\cR=\mathset{r\in\sparenv{0,\abs{\bmw}}:\abs{\Delta_{\rm VT}-\eta(r)}<2K^2A}$, $r_1=\min\cR$ and $r_2=\max\cR$. Then $r^{*}\in\cR$, $\cR=[r_1,r_2]$, and
    \begin{equation}\label{eq_sumupper}
        \sum_{i=r_1+1}^{r_2}\ell_i<\frac{4K^2A}{C^{(1)}-\frac{C^{(3)}}{C^{(2)}}}.
    \end{equation}
\end{proposition}
\begin{IEEEproof}
    It follows from \Cref{lem_Deltaetagap} that $r^{*}\in\cR$. By \Cref{lem_monotonicity}, the function $\eta(r)$ is monotone in $r\in\sparenv{0,\abs{\bmw}}$. This implies $\cR=[r_1,r_2]$. Since $r_1,r_2\in\cR$, we have
    \begin{equation}\label{eq_diffupper}
        \abs{\eta(r_2)-\eta(r_1)}<4K^2A.
    \end{equation}
    By the definition of $\eta(\cdot)$, it holds that $\eta(r)-\eta(r-1)=d\cdot w_r+E$ whenever $r,r-1\in\sparenv{0,\abs{\bmw}}$, which gives
    \begin{equation}\label{eq_lowstep}
        \begin{aligned}
            \abs{\eta(r)-\eta(r-1)}&\ge \abs{d}w_r-\abs{E}\\
            &\overset{(a)}{\ge} w_r-\abs{E}\\
            &\overset{(b)}{\ge} C^{(1)}\ell_r-C^{(3)}\\
            &=\parenv{C^{(1)}-\frac{C^{(3)}}{\ell_r}}\ell_r\\
            &\overset{(c)}{\ge} \parenv{C^{(1)}-\frac{C^{(3)}}{C^{(2)}}}\ell_r,
        \end{aligned}
    \end{equation}
    where (a) uses the assumption $d\ne0$, (b) follows from conditions (i) and (iii), and (c) follows from condition (ii). Since $\eta(r)$ is monotone in $r$, we conclude from \eqref{eq_lowstep} that
    \begin{equation}\label{eq_difflower}
        \abs{\eta(r_2)-\eta(r_1)}=\abs{\sum_{r=r_1+1}^{r_2}\parenv{\eta(r)-\eta(r-1)}}=\sum_{r=r_1+1}^{r_2}\abs{\eta(r)-\eta(r-1)}\ge \parenv{C^{(1)}-\frac{C^{(3)}}{C^{(2)}}}\sum_{r=r_1+1}^{r_2}\ell_r.
    \end{equation}
    Now \eqref{eq_sumupper} follows from \eqref{eq_diffupper}, \eqref{eq_difflower} and condition (iv).
\end{IEEEproof}

Both \Cref{prop_equallength} and \Cref{prop_unequallength} use the value of $\Delta_{\rm VT}$. The next lemma gives the bound required later in \Cref{lem_dEDeltaVT}. It is claim (iii) in the proof of \cite[Lemma 4]{YutingLi2025IT}. We reproduce its proof in the present notation.
\begin{lemma}\label{lem_absVT}
    Let $\bmz\in\mathbb{Z}_{\ge0}^{*}$ and $\bmz\to_{K}\bmw$ for some $K\ge2$. Suppose that both $\bmz$ and $\bmw$ are $A$-bounded, and that $\abs{\lwt{\bmz}-\lwt{\bmw}}\le w_i$ for all $i\in[\abs{\bmw}]$. Then $\abs{\Delta_{\rm VT}}<K^2\parenv{\lwt{\bmz}+A}$.
\end{lemma}
\begin{IEEEproof}
Represent $\bmz$ and $\bmw$ as in \eqref{eq_splitzw}. Since $\bmw_{\sparenv{r^{*}+\abs{\bmt^\prime}+1,\abs{\bmw}}}=\bmv^\prime$, we conclude from \eqref{eq_VT} that 
    \begin{align*}
    \abs{\Delta_{\rm VT}}&\le\abs{\vt(\bms^\prime)-\vt(\bmt^\prime)}+\abs{\bmu^\prime}\abs{\lwt{\bms^\prime}-\lwt{\bmt^\prime}}+\abs{\abs{\bmt^\prime}-\abs{\bms^\prime}}\lwt{\bmv^\prime}\\
    &<K^2A+\abs{\bmu^\prime}\abs{\lwt{\bmz}-\lwt{\bmw}}+\abs{\abs{\bmt^\prime}-\abs{\bms^\prime}}\lwt{\bmv^\prime}\\
    &\le K^2A+\sum_{i=1}^{\abs{\bmu^\prime}}\abs{\lwt{\bmz}-\lwt{\bmw}}+K\lwt{\bmv^\prime}\\
    &\le K^2A+\sum_{i=1}^{\abs{\bmu^\prime}}w_i+K\lwt{\bmv^\prime}\\
    &\overset{(a)}{=}K^2A+\sum_{i=1}^{\abs{\bmu^\prime}}z_i+K\lwt{\bmv^\prime}\\
    &\le K^2A+\sum_{i=1}^{\abs{\bmu^\prime}}z_i+K\sum_{i=\abs{\bmu^\prime}+1}^{\abs{\bmz}}z_i\\
    &\le K^2\parenv{\lwt{\bmz}+A},
    \end{align*}
    where (a) follows from the fact that $z_i=w_i$ for every $1\le i\le \abs{\bmu^\prime}$.
\end{IEEEproof}

Although it is permitted that $\abs{E}=w_i$ in \Cref{lem_monotonicity} and \Cref{lem_absVT}, \eqref{eq_lowstep} implies that we must have $\abs{E}<w_i$ for all $i$ under the assumption in \Cref{prop_unequallength}. The following definition collects the properties of the $K$-substring edit $\bmz\to_{K}\bmw$ required in \Cref{prop_equallength}, \Cref{prop_unequallength} and \Cref{lem_absVT}.
\begin{definition}\label{def_locatable}
Let $\bmz,\bmw\in\mathbb{Z}_{\ge0}^{*}$. Suppose that $\bmw$ is obtained from $\bmz$ by a nontrivial $K$-substring edit. We call this $K$-substring edit \emph{strongly locatable} if the following three conditions are satisfied:
\begin{enumerate}[$(i)$]
\item $\abs{\bmz}\ne\abs{\bmw}$ or $\lwt{\bmz}\ne\lwt{\bmw}$;
\item if $\abs{\bmz}=\abs{\bmw}$, then $\bmz$ and $\bmw$ differ in one coordinate or two consecutive coordinates;
\item $|\lwt{\bmz}-\lwt{\bmw}|< w_i$ for every $i\in\sparenv{\abs{\bmw}}$.
\end{enumerate}
\end{definition}

\begin{remark}
Li \emph{et al.} introduced the definition of a \emph{locatable} substring edit in \cite[Definition~3]{YutingLi2025IT}, which satisfies only conditions (i) and (iii) in \Cref{def_locatable}. We additionally include condition (ii) in the definition. This is the first reason why we call the edit \emph{strongly} locatable. The second reason is that while the locatability property helps to determine the edit $\bmx\to_{k}\bmy$ to within a substring of $\bmx$ of length $O_{q,k}\parenv{(\log n)^2}$ (see Lemma 4 and the proof of Theorem 19 in \cite{YutingLi2025IT}), the strong locatability property enables us to determine the edit $\bmx\to_{k}\bmy$ to within a substring of $\bmx$ of length $O_{q,k}(\log n)$ (see \Cref{thm_framework} and \Cref{thm_specificlocator} in this paper).
\end{remark}

\vspace{5pt}
With these preliminary results, we are now ready to return to our goal in this section, that is, finding an interval $I\subseteq[n]$ with $|I|=O_{q,k}(\log n)$ that covers the edit region in $\bmx$.

Let $\bmx$, $\bmy$, $\cP$ and $f$ be as before. Then $\bmF_{\cP}(\bmx)$ and $\bmF_{\cP}(\bmy)$ have the forms in \eqref{eq_divideFxy} for some $K\ge2$. To explain how \Cref{prop_equallength,prop_unequallength} are helpful, assume that the $K$-substring edit $\bmF_{\cP}(\bmx)\to_{K}\bmF_{\cP}(\bmy)$ is strongly locatable. Denote $\bmz=\bmF_{\cP}(\bmx)$ and $\bmw=\bmF_{\cP}(\bmy)$.

If all the assumptions in \Cref{prop_equallength} are satisfied, then we can determine the first index, say $j$, at which $\bmF_{\cP}(\bmx)$ and $\bmF_{\cP}(\bmy)$ differ. This implies $\bmx_j^{\cP}\ne\bmy_j^{\cP}$. Notice that it is possible that $\bmx_i^{\cP}\ne\bmy_i^{\cP}$ but $f\parenv{\bmx_i^{\cP}}=f\parenv{\bmy_i^{\cP}}$. Set
$a_0=\max\{1,j-K+1\}$ and $b_0=\min\{\abs{\bmw},j+K-1\}$. By \eqref{eq_divideFxy} and the fact $\abs{\bmt^\prime}\le K$, the substring $\bmt^\prime$ is contained in the substring $\bmF_{\cP}(\bmy)_{[a_0,b_0]}$. Then by \Cref{assump1} and the definition of $\bmt^\prime$, the edit region in $\bmy$ is covered by $J$, where $\bmy_J=\bmy_{a_0}^{\cP}\cdots\bmy_{b_0}^{\cP}$.

For every $1\le i\le m_{\cP}(\bmy)$, let $\ell_i=\abs{\bmy_i^{\cP}}$. If all the assumptions in \Cref{prop_unequallength} are satisfied and all the values of $d$, $E$ and $\Delta_{\rm VT}$ are given, then we can find an interval $[r_1,r_2]\subseteq[0,\abs{\bmw}]$ that contains $\abs{\bmu^\prime}$. Specifically, initialize $\eta(\abs{\bmw})=\abs{\bmw}E$ and sequentially compute $\eta(\abs{\bmw}-1),\eta(\abs{\bmw}-2),\ldots,\eta(0)$ using $\eta(r-1)=\eta(r)-dw_{r}-E$ when $r\ge1$. Record the first and last $r$ satisfying $\abs{\Delta_{\rm VT}-\eta(r)}<2K^2A$. These are $r_2$ and $r_1$, respectively.

Set $b_0=\min\{\abs{\bmw},r_2+K\}$. By \eqref{eq_divideFxy} and the fact $\abs{\bmt^\prime}\le K$, the substring $\bmt^\prime$ is contained in $\bmF_{\cP}(\bmy)_{[r_1+1,b_0]}$, which implies that the edit region in $\bmy$ is covered by $J$, where $\bmy_J=\bmy_{r_1+1}^{\cP}\cdots\bmy_{b_0}^{\cP}$. 

If each segment of $\bmy$ has length at most $L$, then $\abs{J}\le (2K-1)L$ in the case of \Cref{prop_equallength}, or
\begin{equation*}
    \abs{J}<\frac{4K^2A}{C^{(1)}-\frac{C^{(3)}}{C^{(2)}}}+\abs{\bmy_{r_2+1}^{\cP}\cdots\bmy_{b_0}^{\cP}}\le \frac{4K^2A}{C^{(1)}-\frac{C^{(3)}}{C^{(2)}}}+KL
\end{equation*}
in the case of \Cref{prop_unequallength}.
Assume $J=[a,b]$. By the definition of a $k$-substring edit, the edit region in $\bmx$ is covered by $[a-k,b+k]\cap[n]$, which is the interval $I$ we are looking for.

The above analysis yields the following theorem.
\begin{theorem}\label{thm_framework}
    Let $\bmx\in\Sigma_q^n$, where $n>k$, and let $\bmy$ be obtained from $\bmx$ by a nontrivial $k$-substring edit. Let $\cP$ be a partition rule and $f:\Sigma_q^{*}\to\mathbb{Z}_{\ge0}$ be a function. Set $\bmz=\bmF_{\cP}(\bmx)$ and $\bmw=\bmF_{\cP}(\bmy)$. Suppose that the following conditions hold.
    \begin{enumerate}[$(i)$]
        \item \Cref{assump1} holds.
        \item The two sequences $\bmz$ and $\bmw$ are $A$-bounded, where $A\ge1$ is an integer.
        \item The $K$-substring edit $\bmz\to_{K}\bmw$ is strongly locatable, where $K\ge2$ is an integer.
        \item Each segment of $\bmy$ under $\cP$ has length at most $L$.
        \item The values of $A$, $K$ and $L$ are known.
        \item There are known positive numbers $C^{(1)},C^{(2)},C^{(3)}$ such that if $d\ne 0$,  they satisfy, for $\ell_i=|\bmy_i^{\cP}|$ and every $i\in[\abs{\bmw}]$, $w_i\ge C^{(1)}\ell_i$, $\ell_i\ge C^{(2)}$, and $\abs{E}\le C^{(3)}<C^{(1)}C^{(2)}$, where $d$ and $E$ are defined in \eqref{eq_dEDelta}.
    \end{enumerate} 
    Let $G_0(\cdot)$, $G_1(\cdot)$ and $\Delta_{\rm VT}$ be defined as before. Suppose that $d$, $E$, $\Delta_{\rm VT}$, together with $G_0(\bmz)-G_0(\bmw)$ and $G_1(\bmz)-G_1(\bmw)$ when $d=0$, are given. Then one can determine from $\bmy$ an interval $I\subseteq[n]$ that covers the edit region in $\bmx$ and satisfies
    \begin{equation}\label{eq_geninterval}
    \abs{I}\le\max\mathset{(2K-1)L,\frac{4K^2A}{C^{(1)}-\frac{C^{(3)}}{C^{(2)}}}+KL}+2k.
    \end{equation}
    Given $\Part_{\cP}(\bmy)$ and $\bmw$, this interval can be found using $O_{q,k}\parenv{n\lambda^2}$ bit operations, where $\lambda=\ceilenv{\log\parenv{(n+k+1)^2K^2A}}$.
\end{theorem}
\begin{IEEEproof}
    The preceding analysis shows how to find such an interval. We only need to count bit operations. Since $\bmz\to_K\bmw$, we have $|d|\le K$. Since every segment is nonempty, we obtain $\abs{\bmz}\le n$ and $\abs{\bmw}\le n+k$. Then the $A$-bounded property of $\bmz$ and $\bmw$ gives $\abs{E}\le KA$, $\abs{\Delta_{\rm VT}}\le A(n+k+1)^2$, and $\abs{G_r(\bmz)-G_r(\bmw)}\le A(n+k)$ for every $0\le r\le 2$. Each value of $\eta(r)$ has absolute value at most $2KA(n+k)$. We need to compare each $\eta(r)$ with $\Delta_{\rm VT}-2K^2A$ and $\Delta_{\rm VT}+2K^2A$. Thus all integers used below, including $2K^2A$, $\Delta_{\rm VT}-\eta(r)$, segment indices, and coordinate sums, have $O(\lambda)$ bits.

    If $d=0$, the proof of \Cref{prop_equallength} computes the index $j$ by a constant number of additions, subtractions, and comparisons, followed by one exact division by $E$. Hence it costs $O(\lambda^2)$ bit operations. Set $a_0=\max\mathset{1,j-K+1}$ and $b_0=\min\mathset{\abs{\bmw},j+K-1}$. Computing these two indices costs $O(\lambda)$ further bit operations.

    Now suppose that $d\ne0$. Compute the threshold $2K^2A$ and initialize $\eta\parenv{\abs{\bmw}}=\abs{\bmw}E$. These computations cost $O(\lambda^2)$ bit operations. For $r=\abs{\bmw},\abs{\bmw}-1,\ldots,1$, compute $\eta(r-1)=\eta(r)-dw_r-E$. At every index $r\in[0,\abs{\bmw}]$, test whether $\abs{\Delta_{\rm VT}-\eta(r)}<2K^2A$. For each $r$, the process uses a fixed number of multiplications, additions, subtractions, and comparisons. Therefore, scanning all $\abs{\bmw}+1$ indices costs $O\parenv{(n+k+1)\lambda^2}$ bit operations. Set $a_0=r_1+1$ and $b_0=\min\mathset{\abs{\bmw},r_2+K}$, which costs $O(\lambda)$ further bit operations.

    In either case, the selected substring of $\bmy$ is $\bmy_{J}=\bmy_{a_0}^{\cP}\cdots\bmy_{b_0}^{\cP}$. Determining $J$ and $I$ requires additional $O(n+k)$ bit operations. Therefore, the total number of bit operations is $O_{q,k}\parenv{n\lambda^2}$.
\end{IEEEproof}

Thus the task is to choose $\cP$ and $f$ so that all the hypotheses of \Cref{thm_framework} hold and the right-hand side of \eqref{eq_geninterval} is $O_{q,k}(\log n)$. This is done in the next section.

\section{Instantiation of the Framework}\label{sec_instantiation}
To apply the general framework in \Cref{thm_framework}, we need to design a suitable partition rule $\cP$ and a function $f$ that satisfy all the hypotheses. We must also determine the exact values of $K$, $A$, $L$, $C^{(1)}$, $C^{(2)}$, $C^{(3)}$, $d$, $E$, $\Delta_{\rm VT}$, $G_0\parenv{\bmF_{\cP}(\bmx)}-G_0\parenv{\bmF_{\cP}(\bmy)}$, and $G_1\parenv{\bmF_{\cP}(\bmx)}-G_1\parenv{\bmF_{\cP}(\bmy)}$.

This section accomplishes this task. The partition rule and the function $f$ are specified in \Cref{subsec_partitionandf}, where the values of $K$, $A$, $L$, $C^{(1)}$, $C^{(2)}$ and $C^{(3)}$ are also determined. In \Cref{sec_locatableverify}, we show that the induced edit $\bmF_{\cP}(\bmx)\to_{K}\bmF_{\cP}(\bmy)$ is strongly locatable if $\bmx\ne\bmy$. Finally, \Cref{sec_finallocator} shows how to determine the values of $d$, $E$, $\Delta_{\rm VT}$, $G_0\parenv{\bmF_{\cP}(\bmx)}-G_0\parenv{\bmF_{\cP}(\bmy)}$ and $G_1\parenv{\bmF_{\cP}(\bmx)}-G_1\parenv{\bmF_{\cP}(\bmy)}$, and presents a code which can locate a $k$-substring edit within a window of length $O_{q,k}(\log n)$.

\subsection{The Partition Rule \texorpdfstring{$\cP$}{P} and the Function \texorpdfstring{$f$}{f}}\label{subsec_partitionandf}
In \cite[Section III]{YutingLi2025IT}, Li \emph{et al.} introduced a partition rule $\cP$ and a function $f$ with the property that if $\bmx\to_{k}\bmy$ and $\bmx\ne\bmy$, then $\bmF_{\cP}(\bmy)$ is obtained from $\bmF_{\cP}(\bmx)$ by a \emph{locatable} $3$-substring edit. In this section, we show that the same $\cP$ and $f$, with suitable modifications, give a strongly locatable edit $\bmF_{\cP}(\bmx)\to_{K}\bmF_{\cP}(\bmy)$.

In the rest of this paper, denote $\bmp_{k}\triangleq 0^k1^k$ for every integer $k\ge1$. We call $\bmp_{k}$ the \emph{pattern}. For a nonempty sequence $\bma\in\Sigma_q^{*}$, if $\bma_{[i,i+2k-1]}=\bmp_{k}$ for some $1\le i\le \abs{\bma}-2k+1$, we say that $i$ is a \emph{pattern start} and that $\bma_{[i,i+2k-1]}$ is an \emph{occurrence} of the pattern. It is straightforward to see that two occurrences of $\bmp_{k}$ cannot overlap.

Denote by $m_{k}(\bma)$ the number of occurrences of $\bmp_{k}$ in $\bma$. The partition rule $\cP$ in \cite[Section III]{YutingLi2025IT} is defined as follows.
\begin{itemize}
    \item If $m_{k}(\bma)=0$, then $\Part_{\cP}(\bma)$ consists of the single segment $\bma$.
    \item If $m_{k}(\bma)\ge1$, for $1\le i\le m_{k}(\bma)$, let $\tau_i$ be the start of the $i$-th pattern in $\bma$. Set $\tau_0=1$ and $\tau_{m_{k}(\bma)+1}=\abs{\bma}+1$. 
    When $\tau_1>1$, there are $m_{k}(\bma)+1$ segments in $\bma$. For $1\le i\le m_{k}(\bma)+1$, the $i$-th segment is $\bma^{(k)}_{i}\triangleq\bma_{[\tau_{i-1},\tau_{i}-1]}$. When $\tau_1=1$, there are $m_{k}(\bma)$ segments in $\bma$. For $1\le i\le m_{k}(\bma)$, the $i$-th segment is $\bma^{(k)}_{i}\triangleq\bma_{[\tau_{i},\tau_{i+1}-1]}$.
\end{itemize}
For clarity, we use $\Part_{k}(\bma)$ for $\Part_{\cP}(\bma)$ throughout this section.

When $m_{k}(\bma)\ge1$ and $\tau_1>1$, it holds that $\abs{\bma_i^{(k)}}\ge 2k$ for every $2\le i\le m_{k}(\bma)+1$. But the first segment $\bma_1^{(k)}$ may have length smaller than $2k$. In this paper, we apply this rule to the virtually extended sequence
$$
\widehat{\bma}\triangleq\bmp_{k}\bma\bmp_{k}
$$
rather than to $\bma$ itself. It is important to notice that the two \emph{virtual} patterns at both ends of $\widehat{\bma}$ are used only to compute the partition and the function $f_k$ defined below. They are not transmitted through the channel. We do not assume that a transmitted $\bma$ begins with a pattern or ends with a pattern to save $4k\log q$ redundant bits.

\begin{remark}
Before proceeding, we explain why we consider the partition of $\widehat{\bma}$ instead of $\bma$. Clearly, the sequence $\widehat{\bma}$ is obtained by adding a copy of $\bmp_{k}$ to both ends of $\bma$. Both copies are needed for our construction.  The left copy starts at coordinate~$1$.  Consequently every segment of $\widehat{\bma}$ begins with a complete occurrence of $\bmp_k$ and hence has length at least
$2k$. Without this copy, the first segment could have length in $[1,2k-1]$, so the inequality $\ell_i\ge C^{(2)}=2k$ used in \Cref{prop_unequallength,thm_framework} would be false.  In addition, for a $k$-substring edit before the first $\bmp_{k}$ inside $\bma$, the unchanged pattern start $\alpha$ chosen in the proof of \Cref{lem_segmentreplacement}~(ii) would not exist.

The right copy begins at coordinate $|\widehat{\bma}|-2k+1$. It guarantees an unchanged pattern occurrence to the right of every edit in $\bma$. Without it, a $k$-substring edit after the last $\bmp_{k}$ inside $\bma$ could leave no such occurrence, and the pattern start $\beta$ in the proof of \Cref{lem_segmentreplacement}~(ii) would not exist.  The decompositions in \eqref{eq_segmentreplacement}, and hence condition~(i) of \Cref{thm_framework}, would then require separate boundary cases and would not follow from the stated argument.

Thus the left and right
copies respectively supply the two pattern starts that bound the edit from the left and from the right. Omitting either copy breaks a specific hypothesis used by the present proof.
\end{remark}

Now suppose $\bmx\in\Sigma_q^{n}$ and $\bmx\to_{k}\bmy$. Use the representation in \eqref{eq_editbasic}, with the convention \eqref{eq_minimalends} whenever the edit is nontrivial. Then
\begin{equation}\label{eq_extendededit}
 \widehat{\bmx}=\bmp_{k}\bmu\bms\bmv\bmp_{k},
 \;
 \widehat{\bmy}=\bmp_{k}\bmu\bmt\bmv\bmp_{k}.
\end{equation}
Thus the $k$-substring edit $\widehat{\bmx}\to_k\widehat{\bmy}$ replaces the same pair $\bms,\bmt$. Let $\rho=\abs{\bmp_{k}\bmu}=2k+\abs{\bmu}$ and $d=\abs{\bmt}-\abs{\bms}$. Coordinates outside the edit regions are related by the bijection $\phi:\sparenv{\rho}\cup\sparenv{\rho+\abs{\bms}+1,\abs{\widehat{\bmx}}}\to\sparenv{\rho}\cup\sparenv{\rho+\abs{\bmt}+1,\abs{\widehat{\bmy}}}$ defined by
\begin{equation}\label{eq_functionphi}
\phi(i)=
    \begin{cases}
        i,\mbox{ if }i\le\rho,\\
        i+d,\mbox{ if }i\ge\rho+\abs{\bms}+1.
    \end{cases}
\end{equation}
For a set $I\subseteq \sparenv{\rho}\cup\sparenv{\rho+\abs{\bms}+1,\abs{\widehat{\bmx}}}$, write $\phi(I)=\mathset{\phi(i):i\in I}$. Clearly, we have
\begin{equation}\label{eq_correspondxy}
    \widehat{\bmx}_{I}=\widehat{\bmy}_{\phi(I)}.
\end{equation}
The coordinate relation in \eqref{eq_correspondxy} is illustrated in \Cref{fig_extendededit}.

\begin{figure}[!t]
\centering
\begin{tikzpicture}[x=0.82cm,y=0.82cm,>=Latex,font=\small]
  \node[left] at (0,0.35) {$\widehat{\bmx}$};
  \node[left] at (0,-2.05) {$\widehat{\bmy}$};
  \draw (0,0) rectangle (4,0.7);
  \draw[fill=red!12] (4,0) rectangle (6,0.7);
  \draw (6,0) rectangle (11,0.7);
  \draw (1.3,0)--(1.3,0.7); \draw (9.7,0)--(9.7,0.7);
  \node at (0.65,0.35) {$\bmp_k$}; \node at (2.65,0.35) {$\bmu$};
  \node at (5,0.35) {$\bms$}; \node at (7.85,0.35) {$\bmv$};
  \node at (10.35,0.35) {$\bmp_k$};
  \draw (0,-2.4) rectangle (4,-1.7);
  \draw[fill=red!12] (4,-2.4) rectangle (7.2,-1.7);
  \draw (7.2,-2.4) rectangle (12.2,-1.7);
  \draw (1.3,-2.4)--(1.3,-1.7); \draw (10.9,-2.4)--(10.9,-1.7);
  \node at (0.65,-2.05) {$\bmp_k$}; \node at (2.65,-2.05) {$\bmu$};
  \node at (5.6,-2.05) {$\bmt$}; \node at (9.05,-2.05) {$\bmv$};
  \node at (11.55,-2.05) {$\bmp_k$};
  \draw[dashed] (4,0)--(4,-1.7);
  \draw[dashed] (6,0)--(7.2,-1.7);
  \draw[->,thick] (2.3,0)--node[left] {$\phi(i)=i$} (2.3,-1.7);
  \draw[->,thick] (9,0)--node[right] {$\phi(j)=j+d$} (10.2,-1.7);
  \fill[black] (2.3,0) circle (1.5pt); 
  \fill[black] (9,0) circle (1.5pt);
  \node[above] at (4,0.75) {$\rho$};
  \node[above] at (6,0.75) {$\rho+|\bms|$};
  \node[below] at (4,-2.45) {$\rho$};
  \node[below] at (7.2,-2.45) {$\rho+|\bmt|$};
  \node[align=center,font=\footnotesize] at (5.25,-0.85) {no map on\\edited symbols};
  \node at (6.1,-3.45) {$d=|\bmt|-|\bms|$, \quad $\widehat{x}_i=\widehat{y}_{\phi(i)}$ for every $i$ in the domain of $\phi$.};
\end{tikzpicture}
\caption{The coordinate correspondence $\phi$ in \eqref{eq_functionphi} and \eqref{eq_correspondxy} in the case $d>0$.}
\label{fig_extendededit}
\end{figure}
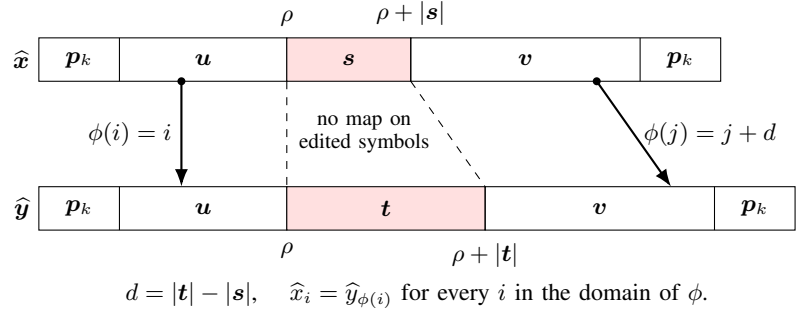

Coordinate $i$ of $\bmx$ corresponds to coordinate $i+2k$ of
$\widehat{\bmx}$. The same relation holds between $\bmy$ and $\widehat{\bmy}$. Therefore, locating errors in $\bmx$ is equivalent to locating errors in $\widehat{\bmx}$. 

Next, we introduce the function $f$ needed in \Cref{thm_framework}. For $\bmz=z_1\cdots z_r\in\Sigma_q^*$ and $a\in[k] $, let
\begin{equation}\label{eq_functiond}
 d_a(\bmz)=
 \sum_{\substack{1\le i\le r\\i\equiv a\pmod k}}z_i\pmod q.
\end{equation}
The vector $(d_1(\bmz),\ldots,d_k(\bmz))$ records the symbol sums in the $k$ residue classes modulo $k$ of the coordinate index. Since $0\le d_a(\bmz)<q$, we can store this vector as the integer
\begin{equation}\label{eq_functionD}
    D_q(\bmz)\triangleq\sum_{a=1}^k d_a(\bmz)q^{a-1}
 \in\{0,1,\ldots,q^k-1\}.
\end{equation}
Also let
\[
 N_k(\bmz)
 \triangleq\abs{\{1\le i\le\abs{\bmz}-k+1:\bmz_{[i,i+k-1]}=1^k\}}
\]
be the number of occurrences of $1^k$ in $\bmz$. For later use in this section and in \Cref{sec_substringeditcode}, set
\begin{align}
 Q=q^k, \quad D_{\max}=3(Q-1),\quad B=3Q-2,\label{eq_DandB}\\
 C_{k}=B(2k-1)+D_{\max}+1.\label{eq_TandC}
\end{align}
The choice of these parameters is justified in the proofs of \Cref{lem_segmentreplacement}~(iv) and \Cref{lem_condition1}.

With the above notations, we define the function $f_{k}:\Sigma_q^{*}\to\mathbb{Z}_{\ge0}$ by
\begin{equation}\label{eq_fmu}
 f_{k}(\bmz)=D_q(\bmz)+B N_k(\bmz)+C_{k}\abs{\bmz}.
\end{equation}

For the framework in \Cref{thm_framework}, take $f=f_k$. Applying this map to all segments of $\widehat{\bmx}$ gives the integer sequence
\begin{equation}\label{eq_Fmu}
 \bmF_{k}(\bmx)\triangleq
 \parenv{f_{k}\parenv{\widehat{\bmx}^{(k)}_1},\ldots,
       f_{k}\parenv{\widehat{\bmx}^{(k)}_{m_{k}(\widehat{\bmx})}}}.
\end{equation}
The argument of $\bmF_k$ is written as $\bmx$, because
$\widehat{\bmx}=\bmp_k\bmx\bmp_k$ is uniquely determined by $\bmx$.

In \Cref{lem_segmentreplacement} below, we show that the partition rule $\cP$ and the function $f_k$ described above satisfy conditions (i)-(ii) and (iv)-(vi) in \Cref{thm_framework}. This lemma also determines the six quantities $K$, $A$, $L$, $C^{(1)}$, $C^{(2)}$ and $C^{(3)}$. Before presenting this lemma and its proof, more notations and definitions are needed.

For a sequence $\bma$, define the set of pattern starts in it by $\cT_k(\bma)=\mathset{a\in[1,|\bma|-2k+1]:\bma_{[a,a+2k-1]}=\bmp_k}$.
Let $\widehat{\bmx}$ and $\widehat{\bmy}$ be given in \eqref{eq_extendededit}. Denote
\begin{align*}
 \cU\parenv{\widehat{\bmx}}&=\{a\in\cT_k(\widehat{\bmx}):
 a\le\rho-2k+1\ \text{or}\ a\ge\rho+\abs{\bms}+1\},\\
 \cU\parenv{\widehat{\bmy}}&=\{b\in\cT_k(\widehat{\bmy}):
 b\le\rho-2k+1\ \text{or}\ b\ge\rho+\abs{\bmt}+1\}.
\end{align*}
In other words, the set $\cU\parenv{\widehat{\bmx}}$ (resp. $\cU\parenv{\widehat{\bmy}}$) consists of the starts of patterns in $\widehat{\bmx}$ (resp. $\widehat{\bmy}$) that are entirely contained in $\bmp_{k}\bmu$ or $\bmv\bmp_{k}$. Then \eqref{eq_extendededit} implies that 
\begin{equation}\label{eq_elementUxUy}
    1\in\cU\parenv{\widehat{\bmx}}\cap[\rho-2k+1],\; n+2k+1\in\cU\parenv{\widehat{\bmx}}\cap[\rho+\abs{\bms}+1,\abs{\widehat{\bmx}}].
\end{equation}
Let $\phi$ be the mapping defined in \eqref{eq_functionphi}. Then \eqref{eq_correspondxy} gives
\begin{align}
    \cU\parenv{\widehat{\bmy}}\cap[\rho-2k+1]=\cU\parenv{\widehat{\bmx}}\cap[\rho-2k+1],\label{eq_Uleft}\\
    \cU\parenv{\widehat{\bmy}}\cap[\rho+\abs{\bmt}+1,\abs{\widehat{\bmy}}]=\phi\parenv{\cU\parenv{\widehat{\bmx}}\cap[\rho+\abs{\bms}+1,\abs{\widehat{\bmx}}]}.\label{eq_Uright}
\end{align}

\begin{definition}[destroyed and created pattern]\label{def_descrepattern}
A pattern in $\widehat{\bmx}$ with start $a\in\cD_{\rm es}\parenv{\widehat{\bmx}}\triangleq\cT_k(\widehat{\bmx})\setminus\cU\parenv{\widehat{\bmx}}$ is called a \emph{destroyed pattern}. In other words, the interval $[a,a+2k-1]$ intersects the edit region in $\widehat{\bmx}$ in the case $\abs{\bms}>0$, or crosses the edit region in $\widehat{\bmx}$ in the case $\bms=\epsilon$. 

A pattern in $\widehat{\bmy}$ with start $b\in\cC_{\rm re}\parenv{\widehat{\bmy}}\triangleq\cT_k(\widehat{\bmy})\setminus\cU\parenv{\widehat{\bmy}}$ is called a \emph{created pattern}. In other words, the interval $[b,b+2k-1]$ intersects the edit region in $\widehat{\bmy}$ in the case $\abs{\bmt}>0$, or crosses the edit region in $\widehat{\bmy}$ in the case $\bmt=\epsilon$.
\end{definition}

Thus, ``destroyed'' and ``created'' refer to a pattern affected by the specified edit, even if it has an occurrence both before and after the edit.
\begin{example}
Let $q=k=2$. Then $\bmp_k=0011$. Take $\bmu=00$, $\bms=0$, $\bmt=\epsilon$, and $\bmv=11$. This edit deletes the third symbol of $\bmx=00011$ and gives $\bmy=0011$. The extended sequences are $\widehat{\bmx}=0011\,00\,0\,11\,0011$ and $\widehat{\bmy}=0011\,00\,11\,0011$. Here $\rho=6$ and $d=-1$. By definition, we have
$$
\begin{aligned}
\cT_k(\widehat{\bmx})&=\{1,6,10\}, &\cU(\widehat{\bmx})&=\{1,10\}, &\cD_{\rm es}(\widehat{\bmx})=\{6\},\\
\cT_k(\widehat{\bmy})&=\{1,5,9\}, &\cU(\widehat{\bmy})&=\{1,9\}, &\cC_{\rm re}(\widehat{\bmy})=\{5\}.
\end{aligned}
$$
The destroyed pattern occupies $[6,9]$ and contains the deleted coordinate $7$. The created pattern occupies $[5,8]$ and crosses the edit region $g=7$ in $\widehat{\bmy}$.
\end{example}

To obtain a short interval from \Cref{thm_framework}, we need every segment of $\widehat{\bmy}$ to have length at most $L=O_{q,k}(\log n)$. We enforce this by restricting $\bmx$ to sequences satisfying a pattern-density condition. Set
\begin{equation}\label{eq_ellnk}
    \ell_{n,k}=
 2k q^{2k}\ceilenv{\log(2n)}.
\end{equation}
A sequence in $\Sigma_q^n$ is called $\parenv{\bmp_{k},\ell_{n,k}}$-\emph{dense} if every length-$\ell_{n,k}$ substring of it contains at least one copy of $\bmp_{k}$. For the construction in \Cref{sec_instantiation,sec_substringeditcode}, assume $n\ge\ell_{n,k}$. This holds for all sufficiently large $n$ because $q$ and $k$ are fixed.

To the best of our knowledge, pattern-density constraints were first introduced \cite{Lenz2020ISIT} for constructing binary burst-deletion correcting codes, and were subsequently used for constructing non-binary burst-deletion correcting codes \cite{Song2023IT,Shuche2024IT,Song2025Entropy}, localized-deletion correcting codes \cite{RawadBitar2021ISIT}, and $k$-substring edit correcting codes \cite{YuanyuanTang2025TCOMM,YutingLi2025IT}.  

Let $\cD_{q,k}(n)$ be the subset of $\Sigma_q^n$ consisting of all $\parenv{\bmp_{k},\ell_{n,k}}$-dense sequences. We have the following lemma.
\begin{lemma}\label{lem_pdense}
For fixed integers $q\ge2$ and $k\ge1$, let $n\ge\max\mathset{\ell_{n,k},q^{\frac{(6k+3)\log q-\log\e}{0.4\log q}}}$ be an integer.
\begin{enumerate}[$(i)$]
    \item(\cfcite{Song2025Entropy}{Lemma 6}) It holds that $\abs{\cD_{q,k}(n)}\ge q^{n-1}$.
    \item If $\bmx\in\cD_{q,k}(n)$, every segment of $\widehat{\bmx}$ has length at most $\ell_{n,k}$.
\end{enumerate}
\end{lemma}
\begin{IEEEproof}
     Claim~(i) follows from \cite[Lemma 6]{Song2025Entropy}, where Song \emph{et al.} presented an efficient algorithm that encodes any sequence in $\Sigma_q^{n-1}$ into $\cD_{q,k}(n)$ whenever $n\ge q^{\frac{(6k+3)\log q-\log\e}{0.4\log q}}$, where $\e$ is the base of the natural logarithm. Claim~(ii) follows immediately from the definition.
\end{IEEEproof}

\begin{lemma}\label{lem_segmentreplacement}
Let $\bmx\in\Sigma_q^n$ and $\bmx\to_k\bmy$. Write $\widehat{\bmx}$ and $\widehat{\bmy}$ as in \eqref{eq_extendededit}.
\begin{enumerate}[$(i)$]
    \item It holds that $\bmF_{k}(\bmy)_i\ge C_{k}\abs{\widehat{\bmy}^{(k)}_i}$ and $\abs{\widehat{\bmy}^{(k)}_i}\ge 2k$ for every $1\le i\le m_{k}(\widehat{\bmy})$. Thus, one may take $C^{(1)}=C_{k}$ and $C^{(2)}=2k$ in \Cref{prop_unequallength}.
    \item There are sequences of segments $\bmP$, $\bmS$, $\bmA$ and $\bmA^\prime$ such that
    \begin{equation}\label{eq_segmentreplacement}
    \mathsf{Part}_k\parenv{\widehat{\bmx}}=\bmP\bmA\bmS,
    \;
    \mathsf{Part}_k\parenv{\widehat{\bmy}}=\bmP\bmA^\prime\bmS.
    \end{equation}
    Let $r=\abs{\cD_{\rm es}\parenv{\widehat{\bmx}}}$ and $r^\prime=\abs{\cC_{\rm re}\parenv{\widehat{\bmy}}}$. Then $\bmA$ contains $r+1$ segments, and $\bmA^\prime$ contains $r^\prime+1$ segments. Moreover,
    \begin{itemize}
        \item $0\le r,r^\prime\le2$;
        \item the edit region in $\widehat{\bmx}$ is covered by the concatenation of segments in $\bmA$;
        \item the edit region in $\widehat{\bmy}$ is covered by the concatenation of segments in $\bmA^\prime$.
    \end{itemize}
    \item The sequence $\bmF_{k}(\bmy)$ is obtained from $\bmF_{k}(\bmx)$ by one $3$-substring edit, and $\abs{m_k(\bmx)-m_k(\bmy)}\le2$.
    \item  We have
          \begin{equation}\label{eq_Fsumdiff}
             \abs{\lwt{\bmF_{k}(\bmx)}-\lwt{\bmF_{k}(\bmy)}}
             \le C_{k}(k+1)-1.
         \end{equation}
         Thus one may take $C^{(3)}=C_k(k+1)-1$ in
         \Cref{prop_unequallength}.
    \item If $\bmx\in\cD_{q,k}(n)$, every segment of $\widehat{\bmy}$ has length at most $3\ell_{n,k}+k$.
    \item If $\bmx\in\cD_{q,k}(n)$, both $\bmF_{k}(\bmx)$ and $\bmF_{k}(\bmy)$ are $A_{n,k}$ bounded, where
    \begin{equation}\label{eq_Amu}
        A_{n,k}\triangleq Q-1+B\parenv{3\ell_{n,k}+1}+C_{k}(3\ell_{n,k}+k).
    \end{equation}
\end{enumerate}
\end{lemma}
\begin{IEEEproof}
Claim (i) follows from \eqref{eq_fmu}, because $D_q$ and $N_k$ are nonnegative and every segment of $\widehat{\bmy}$ begins with the length-$2k$ pattern. In the following, denote $\rho=2k+\abs{\bmu}$.

Next, we prove claim (ii). Let $a\in\cD_{\rm es}\parenv{\widehat{\bmx}}$. By the definition of $\cD_{\rm es}\parenv{\widehat{\bmx}}$, either
$[a,a+2k-1]\cap[\rho+1,\rho+\abs{\bms}]\ne\varnothing$ when $\abs{\bms}>0$, or $a\le\rho<a+2k-1$ when $\abs{\bms}=0$. Both cases imply $a\in[\rho-2k+2,\rho+\abs{\bms}]$. This interval has length
$2k+\abs{\bms}-1\le3k-1$, whereas distinct pattern starts differ by at least $2k$. Therefore, the interval $[\rho-2k+2,\rho+\abs{\bms}]$ contains at most two starts, which implies $\abs{\cD_{\rm es}\parenv{\widehat{\bmx}}}\le2$. Similarly, replacing
$(\widehat{\bmx},\bms,\cD_{\rm es}\parenv{\widehat{\bmx}})$ by
$(\widehat{\bmy},\bmt,\cC_{\rm re}\parenv{\widehat{\bmy}})$ gives $\abs{\cC_{\rm re}\parenv{\widehat{\bmy}}}\le2$.

Let $\alpha$ be the largest start
in $\cU\parenv{\widehat{\bmx}}$ satisfying $\alpha\le\rho-2k+1$, and let $\beta$ be the
smallest start in $\cU\parenv{\widehat{\bmx}}$ satisfying
$\beta\ge\rho+\abs{\bms}+1$. Their existence follows from \eqref{eq_elementUxUy}. Write $\cU\parenv{\widehat{\bmx}}\cap[\rho-2k+1]=\mathset{i_1,\ldots,i_L}$ and $\cU\parenv{\widehat{\bmx}}\cap[\rho+\abs{\bms}+1,\abs{\widehat{\bmx}}]=\mathset{j_1,\ldots,j_R}$, where $1=i_1<\cdots<i_L=\alpha$ and $\beta=j_1<\cdots<j_R=n+2k+1$. Then $\widehat{\bmx}_{\sparenv{i_1,i_2-1}},\ldots,\widehat{\bmx}_{\sparenv{i_{L-1},i_L-1}}$ are all the segments of $\widehat{\bmx}$ whose starts are to the left of $\alpha$, and $\widehat{\bmx}_{\sparenv{j_1,j_2-1}},\ldots,\widehat{\bmx}_{\sparenv{j_{R-1},j_R-1}},\widehat{\bmx}_{\sparenv{j_{R},\abs{\widehat{\bmx}}}}$ are all the segments of $\widehat{\bmx}$ whose starts are no smaller than $\beta$. Let $\bmP=\parenv{\widehat{\bmx}_{\sparenv{i_1,i_2-1}},\ldots,\widehat{\bmx}_{\sparenv{i_{L-1},i_L-1}}}$ and $\bmS=\parenv{\widehat{\bmx}_{\sparenv{j_1,j_2-1}},\ldots,\widehat{\bmx}_{\sparenv{j_{R-1},j_R-1}},\widehat{\bmx}_{\sparenv{j_{R},\abs{\widehat{\bmx}}}}}$. Notice that $\bmP$ might be empty.

By the definition of $\cD_{\rm es}\parenv{\widehat{\bmx}}$, the set $\mathset{\alpha}\cup\cD_{\rm es}\parenv{\widehat{\bmx}}$ consists of the pattern starts that determine all the segments of $\widehat{\bmx}$ which do not appear in $\bmP$ or $\bmS$. Let $\bmA$ be the sequence of these $r+1$ segments. Then we have $\Part_{k}\parenv{\widehat{\bmx}}=\bmP\bmA\bmS$.

By \eqref{eq_correspondxy}, $\alpha$ is also the largest start in $\cU\parenv{\widehat{\bmy}}$ satisfying $\alpha\le\rho-2k+1$, and $\phi(\beta)$ is the
smallest start in $\cU\parenv{\widehat{\bmy}}$ satisfying
$\phi(\beta)\ge\rho+\abs{\bmt}+1$. Then it follows from \eqref{eq_correspondxy} and \eqref{eq_Uleft} that $\widehat{\bmx}_{\sparenv{i_1,i_2-1}},\ldots,\widehat{\bmx}_{\sparenv{i_{L-1},i_L-1}}$ are all the segments of $\widehat{\bmy}$ whose starts are to the left of $\alpha$, and that $\widehat{\bmx}_{\sparenv{j_1,j_2-1}},\ldots,\widehat{\bmx}_{\sparenv{j_{R-1},j_R-1}},\widehat{\bmx}_{\sparenv{j_{R},\abs{\widehat{\bmx}}}}$ are all the segments of $\widehat{\bmy}$ whose starts are no smaller than $\phi(\beta)$. By the definition of $\cC_{\rm re}\parenv{\widehat{\bmy}}$, the set $\mathset{\alpha}\cup\cC_{\rm re}\parenv{\widehat{\bmy}}$ consists of the pattern starts that determine all the segments of $\widehat{\bmy}$ which do not appear in $\bmP$ and $\bmS$. Let $\bmA^\prime$ be the sequence of these $r^\prime+1$ segments. Thus we have $\Part_{k}\parenv{\widehat{\bmy}}=\bmP\bmA^\prime\bmS$.
This proves (ii).

Claim (ii) and \eqref{eq_Fmu} imply that $\bmF_{k}(\bmy)$ is obtained from $\bmF_{k}(\bmx)$ by one $3$-substring edit and $\abs{m_k\parenv{\widehat{\bmx}}-m_{k}\parenv{\widehat{\bmy}}}\le2$. Since $m_k\parenv{\widehat{\bmx}}=m_k\parenv{\bmx}+2$ and $m_k\parenv{\widehat{\bmy}}=m_k\parenv{\bmy}+2$, we have $\abs{m_k(\bmx)-m_k(\bmy)}\le2$.

To prove (iv), define
\begin{equation}\label{eq_DeltaDN}
 \Delta_D=\sum_{i=1}^{m_{k}(\widehat{\bmx})}D_q\parenv{\widehat{\bmx}^{(k)}_i}-
          \sum_{j=1}^{m_{k}(\widehat{\bmy})}D_q\parenv{\widehat{\bmy}^{(k)}_j},
 \;
 \Delta_N=\sum_{i=1}^{m_{k}(\widehat{\bmx})}N_{k}\parenv{\widehat{\bmx}^{(k)}_i}-
          \sum_{j=1}^{m_{k}(\widehat{\bmy})}N_{k}\parenv{\widehat{\bmy}^{(k)}_j}.
\end{equation}
Let $\cS$ and $\cS^\prime$ be the multisets of segments contained in $\bmA$ and $\bmA^\prime$, respectively. Then by claim (ii), we have $1\le\abs{\cS},\abs{\cS^\prime}\le 3$ and
\begin{equation}\label{eq_absDeltaD}
    \begin{aligned}
    \abs{\Delta_D}&=\abs{\sum_{\widehat{\bmx}^{(k)}_i\in\cS}D_q\parenv{\widehat{\bmx}^{(k)}_i}-
          \sum_{\widehat{\bmy}^{(k)}_j\in\cS^\prime}D_q\parenv{\widehat{\bmy}^{(k)}_j}}\\
          &\le\max\mathset{\sum_{\widehat{\bmx}^{(k)}_i\in\cS}D_q\parenv{\widehat{\bmx}^{(k)}_i},\sum_{\widehat{\bmy}^{(k)}_j\in\cS^\prime}D_q\parenv{\widehat{\bmy}^{(k)}_j}}\\
          &\le 3(Q-1)=D_{\max}.
    \end{aligned}
\end{equation}
Notice that no occurrence of $1^k$ overlaps with two segments, because every segment starts with $0^k$. Consequently,
\begin{equation}\label{eq_Nmuwhole}
    \sum_{i=1}^{m_{k}\parenv{\widehat{\bmx}}}N_k\parenv{\widehat{\bmx}^{(k)}_i}
 =N_k\parenv{\widehat{\bmx}},\; \sum_{j=1}^{m_{k}\parenv{\widehat{\bmy}}}N_{k}\parenv{\widehat{\bmy}^{(k)}_j}=N_{k}\parenv{\widehat{\bmy}}.
\end{equation}
Therefore,
 \begin{equation}\label{eq_DeltaNequiv}
     \Delta_N=N_{k}(\widehat{\bmx})-N_{k}(\widehat{\bmy}).
 \end{equation}
 Suppose that $\widehat{\bmx}_{[a,a+k-1]}=1^{k}$. This occurrence of $1^{k}$ can be destroyed by the $k$-substring edit only if $[a,a+k-1]$ intersects or crosses the edit region in $\widehat{\bmx}$, implying that $\rho-k+2\le a\le \rho+\abs{\bms}$. Therefore, there are at most $\abs{\bms}+k-1\le 2k-1$ copies of $1^{k}$ in $\widehat{\bmx}$ that can be destroyed. By the same argument, there are at most $2k-1$ copies of $1^{k}$ in $\widehat{\bmy}$ that can be created. Consequently,
 \begin{equation}\label{eq_absDeltaN}
     \abs{\Delta_N}=\abs{N_{k}(\widehat{\bmx})-N_{k}(\widehat{\bmy})}\le 2k-1.
 \end{equation}

Now it follows that
\begin{equation}
    \begin{aligned}
        &\abs{\lwt{\bmF_{k}(\bmx)}-\lwt{\bmF_{k}(\bmy)}}\\
        \le&\abs{\sum_{i}D_q\parenv{\widehat{\bmx}^{(k)}_i}-
          \sum_{j}D_q\parenv{\widehat{\bmy}^{(k)}_j}}+B\abs{\sum_{i}N_k\parenv{\widehat{\bmx}^{(k)}_i}-\sum_{j}N_k\parenv{\widehat{\bmy}^{(k)}_j}}+C_{k}\abs{\sum_{i}\abs{\widehat{\bmx}^{(k)}_i}-\sum_{j}\abs{\widehat{\bmy}^{(k)}_j}}\\
        =&\abs{\Delta_D}+B\abs{\Delta_N}+C_{k}\abs{\abs{\bmx}-\abs{\bmy}}\\
        \le& D_{\max}+B(2k-1)+C_{k}k=C_{k}(k+1)-1.
    \end{aligned}
\end{equation}
This completes the proof of (iv).

Since $C_{k}(k+1)-1<2k C_{k}$, the choice of $C^{(1)}$, $C^{(2)}$ and $C^{(3)}$ in claims (i) and (iv) is correct.

When $\bmx\in\cD_{q,k}(n)$, each segment of $\widehat{\bmx}$ has length at most $\ell_{n,k}$ by \Cref{lem_pdense}. A segment of $\widehat{\bmy}$ contains symbols from at most three segments of $\widehat{\bmx}$, together with at most $k$
symbols of $\bmt$. Its length is therefore at most
$3\ell_{n,k}+k$. This proves~(v).

Finally, claim (vi) follows from \eqref{eq_fmu}, \eqref{eq_Fmu} and (v).
\end{IEEEproof}

\subsection{Strong Locatability of the Induced \texorpdfstring{$3$}{3}-Substring Edit}\label{sec_locatableverify}
Throughout this subsection, let $\bmy$ be obtained from $\bmx$ by a \emph{nontrivial} $k$-substring edit, and use the representation in \eqref{eq_editbasic}--\eqref{eq_minimalends}. Then $\widehat{\bmx}$ and $\widehat{\bmy}$ have the forms in \eqref{eq_extendededit}. We will show that the induced $3$-substring edit $\bmF_{k}(\bmx)\to_{3}\bmF_{k}(\bmy)$ is strongly locatable. Conditions (i)--(iii) in \Cref{def_locatable} are established in \Cref{lem_condition1,lem_condition2,lem_condition3}, respectively.

Notice that $\abs{\bmF_{k}(\bmx)}=\abs{\bmF_{k}(\bmy)}$ (i.e., $m_k(\widehat{\bmx})=m_k(\widehat{\bmy})$) is equivalent to $m_{k}(\bmx)=m_{k}(\bmy)$. Recall the definition of $\cD_{\rm es}\parenv{\widehat{\bmx}}$ and $\cC_{\rm re}\parenv{\widehat{\bmy}}$ from \Cref{def_descrepattern}. In the following proofs, put $\rho=\abs{\bmp_{k}\bmu}=2k+\abs{\bmu}$.
\begin{lemma}\label{lem_patternchange}
    If $m_{k}(\bmx)=m_{k}(\bmy)$, then exactly one of the following two situations occurs.
    \begin{enumerate}[$(i)$]
        \item $\cD_{\rm es}\parenv{\widehat{\bmx}}=\cC_{\rm re}\parenv{\widehat{\bmy}}=\varnothing$. In this case, for some
        $i\in[m_k(\widehat{\bmx})]$, the edit region in $\widehat{\bmx}$ is covered by the
        $i$-th segment of $\widehat{\bmx}$ and the edit region in $\widehat{\bmy}$ is covered by
        the $i$-th segment of $\widehat{\bmy}$. Moreover,
        $\widehat{\bmx}^{(k)}_j=\widehat{\bmy}^{(k)}_j$ for every
        $j\ne i$.
        \item $\abs{\cD_{\rm es}\parenv{\widehat{\bmx}}}=\abs{\cC_{\rm re}\parenv{\widehat{\bmy}}}=1$.
    \end{enumerate}
\end{lemma}
\begin{IEEEproof}
Set $r=\abs{\cD_{\rm es}\parenv{\widehat{\bmx}}}$ and $r^\prime=\abs{\cC_{\rm re}\parenv{\widehat{\bmy}}}$.
By \Cref{lem_segmentreplacement} (ii), we have $0\le r,r^\prime\le2$ and $m_k(\bmx)-m_k(\bmy)=m_{k}\parenv{\widehat{\bmx}}-m_{k}\parenv{\widehat{\bmy}}=r-r^\prime$. The assumption $m_k(\bmx)=m_k(\bmy)$ therefore gives $r=r^\prime$.

We exclude the case $r=r^\prime=2$. Suppose first that $r=2$. Write $\cD_{\rm es}\parenv{\widehat{\bmx}}=\{a_1,a_2\}$ with $a_1<a_2$. It holds that $a_2-a_1\ge 2k$ because both $a_1$ and $a_2$ are pattern starts in $\widehat{\bmx}$. By the definition of $\cD_{\rm es}\parenv{\widehat{\bmx}}$, both $[a_1,a_1+2k-1]$ and $[a_2,a_2+2k-1]$ intersect or cross the edit region in $\widehat{\bmx}$. If $\bms=\epsilon$, this implies
$a_j\le\rho<a_j+2k-1$ for $j=1,2$, contradicting the fact $a_2-a_1\ge2k$.

Hence $\abs{\bms}\ge1$. Put
$L=\rho+1$ and $R=\rho+\abs{\bms}$. Since $\abs{\bms}\le k$ and both $[a_1,a_1+2k-1]$ and $[a_2,a_2+2k-1]$ intersect $[L,R]$, we have $L\in[a_1+k,a_1+2k-1]$ and $R\in[a_2,a_2+k-1]$. Thus
$s_1=\widehat{\bmx}_L=1$ and $s_{\abs{\bms}}=\widehat{\bmx}_R=0$.

The same argument applied to $\widehat{\bmy}$, $\bmt$ and $\cC_{\rm re}\parenv{\widehat{\bmy}}$ shows that
$r^\prime=2$ implies $\abs{\bmt}\ge1$, $t_1=1$, and
$t_{\abs{\bmt}}=0$. This contradicts the assumption
$s_1\ne t_1$ and $s_{\abs{\bms}}\ne t_{\abs{\bmt}}$ in
\eqref{eq_minimalends}.  Hence $r=r^\prime\in\{0,1\}$.

If $r=r^\prime=1$, claim (ii) holds.  If $r=r^\prime=0$, it follows from \Cref{lem_segmentreplacement} (ii) that the edit region in $\widehat{\bmx}$ is covered by the single segment in $\bmA$, and the edit region in $\widehat{\bmy}$ is covered by the single segment in $\bmA^\prime$. \Cref{eq_segmentreplacement} implies that if the segment in $\bmA$ is the $i$-th segment of $\widehat{\bmx}$, then the segment in $\bmA^\prime$ is the $i$-th segment of $\widehat{\bmy}$. Finally, the assertion that $\widehat{\bmx}^{(k)}_j=\widehat{\bmy}^{(k)}_j$ for every $j\ne i$ also follows from \eqref{eq_segmentreplacement}.
\end{IEEEproof}

The next statement is the $q$-ary counterpart of
\cite[Lemma 10]{YutingLi2025IT}. We give a new proof because the argument there is stated for binary sequences.
\begin{lemma}\label{lem_patternshift}
Suppose $\abs{\bmx}=\abs{\bmy}$ and $m_{k}(\bmx)=m_{k}(\bmy)$. If situation (ii) in \Cref{lem_patternchange} holds, then $N_{k}\parenv{\widehat{\bmx}}\ne N_{k}\parenv{\widehat{\bmy}}$.
\end{lemma}
\begin{IEEEproof}
Since $\abs{\bmx}=\abs{\bmy}$, we have $\abs{\bms}=\abs{\bmt}$. Since $\bmx\ne\bmy$, $1\le \abs{\bms}\le k$. Put $[L,R]=[\rho+1,\rho+\abs{\bms}]$. Equation~\eqref{eq_minimalends}
gives $\widehat{x}_L\ne\widehat{y}_L$ and
$\widehat{x}_R\ne\widehat{y}_R$.
Write $\cD_{\rm es}\parenv{\widehat{\bmx}}=\{a\}$ and $\cC_{\rm re}\parenv{\widehat{\bmy}}=\{b\}$. By definition, it holds that
$\widehat{\bmx}_{[a,a+2k-1]}=\widehat{\bmy}_{[b,b+2k-1]}=\bmp_k$, and both intervals $[a,a+2k-1]$ and $[b,b+2k-1]$ intersect the edit region $[L,R]$.

We first show that $a\ne b$. Suppose $a=b$. Then $\widehat{\bmx}_{[a,a+2k-1]}=\widehat{\bmy}_{[a,a+2k-1]}$. Since $R-L<2k$, the interval $[a,a+2k-1]$ cannot be contained in $[L,R]$. On the other hand, we have $[L,R]\cap[a,a+2k-1]\ne\varnothing$. Therefore, at least one of $L$ and $R$ belongs to $[a,a+2k-1]$, implying $\widehat{x}_L=\widehat{y}_L$ or $\widehat{x}_R=\widehat{y}_R$. This contradicts the fact $\widehat{x}_L\ne\widehat{y}_L$ and $\widehat{x}_R\ne\widehat{y}_R$. Thus $a\ne b$. 

Interchanging $\bmx$ and $\bmy$ interchanges $a$ and $b$. Thus, it is sufficient to assume $a<b$. This arrangement is illustrated in \Cref{fig_patternshift}. We next prove
\begin{equation}\label{eq_twoinclusion}
    L\in[a+k,a+2k-1],\; R\in[b,b+k-1].
\end{equation}
There are three possible relative positions of $a$ and $b$.

If $b\ge a+2k$, \eqref{eq_twoinclusion} follows immediately from the fact that $R-L<k$ and the interval $[L,R]$ intersects both $[a,a+2k-1]$ and $[b,b+2k-1]$.

If $b<a+k$, then $a+k\le b+k-1< a+2k-1$. This gives $\widehat{\bmx}_{[a+k,b+k-1]}=1^{b-a}$ and $\widehat{\bmy}_{[a+k,b+k-1]}=0^{b-a}$, which implies $L\le a+k$ and $R\ge b+k-1$. Since $R-L<k$, $\widehat{\bmx}_{[b,a+k-1]}=\widehat{\bmy}_{[b,a+k-1]}=0^{a+k-b}$, $\widehat{\bmx}_{[b+k,a+2k-1]}=\widehat{\bmy}_{[b+k,a+2k-1]}=1^{a+k-b}$, we conclude that $L=a+k$ and $R=b+k-1$.

Finally, suppose $a+k\le b<a+2k$. Then $a+2k-1\le b+k-1$. Therefore, we have $\widehat{\bmx}_{[b,a+2k-1]}=1^{a+2k-b}$ and $\widehat{\bmy}_{[b,a+2k-1]}=0^{a+2k-b}$, which implies $L\le b$ and $R\ge a+2k-1$. Combining these inequalities with $R-L<k$ yields $a+k\le L\le b$ and $a+2k-1\le R\le b+k-1$. This proves \eqref{eq_twoinclusion} in all three cases.

Combining $L\in[a+k,a+2k-1]$ with $\widehat{x}_L\ne\widehat{y}_L$, we get $\widehat{x}_L=1$ and $\widehat{y}_L\ne1$. Combining $R\in[b,b+k-1]$ with $\widehat{x}_R\ne\widehat{y}_R$, we get $\widehat{x}_R\ne0$ and $\widehat{y}_R=0$. In particular, $\widehat{\bmx}_{[a+k,a+2k-1]}=1^{k}$ but $\widehat{\bmy}_{[a+k,a+2k-1]}\ne 1^{k}$.

It remains to show that $\widehat{\bmy}$ has no occurrence of $1^k$ that is absent from $\widehat{\bmx}$. Consider a length-$k$ substring $\widehat{\bmy}_{[i,i+k-1]}$ of $\widehat{\bmy}$. If $L-k+1\le i\le R$, the interval $[i,i+k-1]$ contains $L$ or $R$. This implies that the substring $\widehat{\bmy}_{[i,i+k-1]}$ must contain a symbol that is not $1$ and thus, $\widehat{\bmy}_{[i,i+k-1]}\ne 1^{k}$. If $i\le L-k$ or $i>R$, we have $\widehat{\bmy}_{[i,i+k-1]}=\widehat{\bmx}_{[i,i+k-1]}$. Therefore, we can conclude that $N_{k}\parenv{\widehat{\bmx}}- N_{k}\parenv{\widehat{\bmy}}\ge1$. This completes the proof.
\end{IEEEproof}

Regarding the function $D_q(\cdot)$ defined in \eqref{eq_functionD}, we have the following lemma.
\begin{lemma}\label{lem_Ddetector}
Let $\bmz,\bmz^\prime\in\Sigma_q^r$ be two distinct sequences. Suppose that there is an interval $[L,R]\subseteq[r]$ with $R-L< k$ such that $z_i\ne z_i^\prime$ only if $i\in [L,R]$. Then $D_q(\bmz)\ne D_q(\bmz^\prime)$.
\end{lemma}
\begin{IEEEproof}
Because $R-L<k$, the interval $[L,R]$ contains at most one coordinate from each residue class modulo $k$. Then by \eqref{eq_functiond} and the assumption that $\bmz\ne\bmz^\prime$ and $\bmz_{[r]\setminus[L,R]}= \bmz^\prime_{[r]\setminus[L,R]}$, we have $\parenv{d_1(\bmz),\ldots,d_k(\bmz)}\ne \parenv{d_1(\bmz^\prime),\ldots,d_k(\bmz^\prime)}$. Now the conclusion follows from \eqref{eq_functionD}.
\end{IEEEproof}

With \Cref{lem_patternchange,lem_patternshift,lem_Ddetector} in hand, we are ready to prove that $\bmF_{k}(\bmx)$ and $\bmF_{k}(\bmy)$ satisfy condition (i) in \Cref{def_locatable}.
\begin{lemma}\label{lem_condition1}
    If $m_{k}(\bmx)=m_{k}(\bmy)$, then $\lwt{\bmF_{k}(\bmx)}\ne\lwt{\bmF_{k}(\bmy)}$.
\end{lemma}
\begin{IEEEproof}
Let $E=\lwt{\bmF_{k}(\bmx)}-\lwt{\bmF_{k}(\bmy)}$. From \eqref{eq_fmu}, \eqref{eq_Fmu} and \eqref{eq_DeltaDN}, we obtain $E=\Delta_D+B\Delta_N+C_{k}\parenv{\abs{\widehat{\bmx}}-\abs{\widehat{\bmy}}}$. It is sufficient to show $\abs{E}>0$.

First suppose $\abs{\bmx}\ne\abs{\bmy}$. By \eqref{eq_extendededit}, we have $\abs{\abs{\widehat{\bmx}}-\abs{\widehat{\bmy}}}\ge1$. Then $\abs{E}\ge C_{k}\abs{\abs{\widehat{\bmx}}-\abs{\widehat{\bmy}}}-\parenv{\abs{\Delta_D}+B\abs{\Delta_N}}\ge C_{k}-\parenv{\abs{\Delta_D}+B\abs{\Delta_N}}\overset{(a)}{\ge} C_{k}-\parenv{D_{\max}+B(2k-1)}>0$, where (a) follows from \eqref{eq_absDeltaD} and \eqref{eq_absDeltaN}, and the last inequality follows from \eqref{eq_TandC}.

Now suppose $\abs{\bmx}=\abs{\bmy}$. If $N_{k}\parenv{\widehat{\bmx}}\ne N_{k}\parenv{\widehat{\bmy}}$, it holds that $\abs{E}\ge B\abs{\Delta_N}-\abs{\Delta_D}\overset{(b)}{=}B\abs{N_{k}\parenv{\widehat{\bmx}}- N_{k}\parenv{\widehat{\bmy}}}-\abs{\Delta_D}\ge B-D_{\max}>0$, where (b) follows from \eqref{eq_DeltaNequiv} and the last inequality follows from \eqref{eq_DandB}.

It remains to consider the case $\abs{\bmx}=\abs{\bmy}$ and $N_{k}\parenv{\widehat{\bmx}}= N_{k}\parenv{\widehat{\bmy}}$. \Cref{lem_patternshift} excludes situation~(ii) of \Cref{lem_patternchange}. Let $i\in\sparenv{m_{k}\parenv{\widehat{\bmx}}}$ be the index in \Cref{lem_patternchange}~(i). Since $\abs{\widehat{\bmx}}=\abs{\widehat{\bmy}}$ and $\widehat{\bmx}_j^{(k)}=\widehat{\bmy}_j^{(k)}$ for every $j\ne i$, we conclude $\abs{\widehat{\bmx}_i^{(k)}}=\abs{\widehat{\bmy}_i^{(k)}}$. Moreover,  by \Cref{lem_patternchange}~(i), $\widehat{\bmx}_i^{(k)}$ and $\widehat{\bmy}_i^{(k)}$ differ only in the edit regions, whose length is at most $k$. Thus, the hypotheses of \Cref{lem_Ddetector} hold and $\abs{E}=\abs{\Delta_D}=\abs{D_q\parenv{\widehat{\bmx}_i^{(k)}}-D_q\parenv{\widehat{\bmy}_i^{(k)}}}>0$. This completes the proof.
\end{IEEEproof}
\begin{figure}[!t]
\centering
\begin{tikzpicture}[x=0.55cm,y=0.7cm,>=Latex,font=\small]
  \node[left] at (0,0) {$\widehat{\bmx}$};
  \draw[thick] (0,0)--(19,0);
  \draw[fill=gray!25] (2,-0.28) rectangle (6,0.28);
  \draw[fill=gray!25] (6,-0.28) rectangle (10,0.28);
  \node at (4,0) {$0^k$}; \node at (8,0) {$1^k$};
  \draw[decorate,decoration={brace,amplitude=3pt}] (2,0.5)--(6,0.5);
  \draw[decorate,decoration={brace,amplitude=3pt}] (6,0.5)--(10,0.5);
  \node[above] at (4,0.64) {$k$ symbols};
  \node[above] at (8,0.64) {$k$ symbols};
  \node[below] at (2,-0.4) {$a$};
  \node[left] at (0,-2.8) {$\widehat{\bmy}$};
  \draw[thick] (0,-2.8)--(19,-2.8);
  \draw[fill=gray!25] (8,-3.08) rectangle (12,-2.52);
  \draw[fill=gray!25] (12,-3.08) rectangle (16,-2.52);
  \node at (10,-2.8) {$0^k$}; \node at (14,-2.8) {$1^k$};
  \draw[decorate,decoration={brace,mirror,amplitude=3pt}] (8,-3.5)--(12,-3.5);
  \draw[decorate,decoration={brace,mirror,amplitude=3pt}] (12,-3.5)--(16,-3.5);
  \node[below] at (10,-3.64) {$k$ symbols};
  \node[below] at (14,-3.64) {$k$ symbols};
  \node[below left] at (8,-3.12) {$b$};
  \draw[red,thick] (7.5,-0.38) rectangle (10.5,0.38);
  \draw[red,thick] (7.5,-3.18) rectangle (10.5,-2.42);
  \draw[red,dashed] (7.5,-0.4)--(7.5,-2.4);
  \draw[red,dashed] (10.5,-0.4)--(10.5,-2.4);
  \node[red,align=center] at (9,-1.4) {$[L,R]$\\$R-L< k$};
\end{tikzpicture}
\caption{The case $a<b$ in the proof of \Cref{lem_patternshift}. The edit region $[L,R]$ meets the $1^k$ half of the pattern starting at $a$ and the $0^k$ half of the pattern starting at $b$.}
\label{fig_patternshift}
\end{figure}
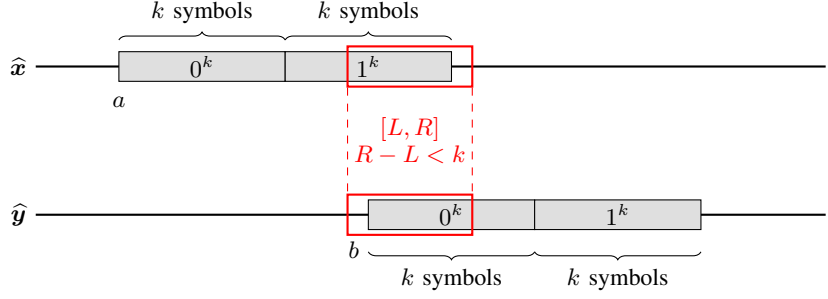

The next lemma states that $\bmF_{k}(\bmx)$ and $\bmF_{k}(\bmy)$ satisfy condition (ii) in \Cref{def_locatable}.
\begin{lemma}\label{lem_condition2}
    If $m_{k}(\bmx)=m_{k}(\bmy)$, then $\bmF_{k}(\bmx)$ and $\bmF_{k}(\bmy)$ differ in either one coordinate or two consecutive coordinates.
\end{lemma}
\begin{IEEEproof}
First, it follows from \Cref{lem_condition1} that $\bmF_{k}(\bmx)\ne\bmF_{k}(\bmy)$.

If situation (i) in \Cref{lem_patternchange} holds, then by \eqref{eq_Fmu}, there is some $i\in\sparenv{m_{k}\parenv{\widehat{\bmx}}}$ such that $\bmF_{k}\parenv{\bmx}_j=\bmF_{k}\parenv{\bmy}_j$ for all $j\ne i$. Since $\bmF_{k}(\bmx)\ne\bmF_{k}(\bmy)$, we conclude that $\bmF_{k}\parenv{\bmx}_{i}\ne\bmF_{k}\parenv{\bmy}_{i}$.

Now suppose that situation (ii) in \Cref{lem_patternchange} holds. By \Cref{lem_segmentreplacement}~(ii), we can write $\bmA=\parenv{\widehat{\bmx}^{(k)}_i,\widehat{\bmx}^{(k)}_{i+1}}$ and $\bmA^\prime=\parenv{\widehat{\bmy}^{(k)}_i,\widehat{\bmy}^{(k)}_{i+1}}$ for some $1\le i<m_{k}\parenv{\widehat{\bmx}}$. \Cref{eq_Fmu,eq_segmentreplacement} give $\bmF_{k}\parenv{\bmx}_j=\bmF_{k}\parenv{\bmy}_j$ for every $j\notin\{i,i+1\}$. Since $\bmF_{k}(\bmx)\ne\bmF_{k}(\bmy)$, at least one of the two equalities $\bmF_{k}\parenv{\bmx}_{i}=\bmF_{k}\parenv{\bmy}_{i}$ and $\bmF_{k}\parenv{\bmx}_{i+1}=\bmF_{k}\parenv{\bmy}_{i+1}$ must fail. Now the proof is completed.
\end{IEEEproof}

Finally, we verify condition (iii) in \Cref{def_locatable}.
\begin{lemma}\label{lem_condition3}
  It holds that $\abs{\lwt{\bmF_{k}(\bmx)}-\lwt{\bmF_{k}(\bmy)}}<\bmF_{k}(\bmy)_i$ for all $1\le i\le m_{k}\parenv{\widehat{\bmy}}$.
\end{lemma}
\begin{IEEEproof}
    It follows from \eqref{eq_fmu} and \eqref{eq_Fmu} that $\bmF_{k}(\bmy)_i=f_{k}\parenv{\widehat{\bmy}^{(k)}_i}\ge C_{k}\abs{\widehat{\bmy}^{(k)}_i}\ge 2k C_{k}$ for every $i$. On the other hand, \eqref{eq_Fsumdiff} gives $\abs{\lwt{\bmF_k(\bmx)}-\lwt{\bmF_k(\bmy)}}\le C_k(k+1)-1<2kC_k\le \bmF_{k}(\bmy)_i$, completing the proof.
\end{IEEEproof}

Combining \Cref{lem_condition1,lem_condition2,lem_condition3} yields the following main result of this subsection.
\begin{corollary}\label{cor_locatableverify}
Let $\bmx\in\Sigma_q^n$ and $\bmy$ be obtained from $\bmx$ by one nontrivial $k$-substring edit. Then $\bmF_{k}(\bmy)$ is obtained from $\bmF_{k}(\bmx)$ by a strongly locatable $3$-substring edit.
\end{corollary}

\subsection{The Localization Syndrome}\label{sec_finallocator}
Throughout this subsection, it is assumed that $\bmx\in\Sigma_q^n$ and $\bmy$ is obtained from $\bmx$ by a (not necessarily nontrivial) $k$-substring edit. Write $\bmx$ and $\bmy$ as in \eqref{eq_editbasic}. Then $\widehat{\bmx}$ and $\widehat{\bmy}$ have the forms in \eqref{eq_extendededit}. Define
\begin{equation}\label{eq_fivediff}
    \begin{array}{c}
         d:=\abs{\bmF_{k}\parenv{\bmy}}-\abs{\bmF_{k}\parenv{\bmx}}=m_{k}(\bmy)-m_{k}(\bmx),\\
         E:=\lwt{\bmF_{k}\parenv{\bmx}}-\lwt{\bmF_{k}\parenv{\bmy}},\\
    \Delta_{\rm VT}:=\vt\parenv{\bmF_{k}\parenv{\bmx}}-\vt\parenv{\bmF_{k}\parenv{\bmy}},\\
    \Delta_{G_r}:=G_r\parenv{\bmF_{k}(\bmx)}-G_r\parenv{\bmF_{k}(\bmy)}, \; r\in\{0,1\},
    \end{array}
\end{equation}
where the functions $G_0$ and $G_1$ are defined in \eqref{eq_functionG}.

With the results in \Cref{lem_segmentreplacement} and \Cref{cor_locatableverify}, to apply \Cref{thm_framework}, it remains to determine the values of $d$, $E$, $\Delta_{\rm VT}$, $\Delta_{G_0}$ and $\Delta_{G_1}$. The information needed to determine these values is specified in \Cref{lem_dEDeltaVT,lem_Gvalue} below. It is important to notice that the values of $m_{k}(\bmy)$, $\lwt{\bmF_{k}\parenv{\bmy}}$, $\vt\parenv{\bmF_{k}\parenv{\bmy}}$, $G_0\parenv{\bmF_{k}(\bmy)}$, $G_1\parenv{\bmF_{k}(\bmy)}$ can be immediately computed once $\bmy$ is given.

Recall $Q$, $B$, $C_{k}$ from \eqref{eq_DandB}-\eqref{eq_TandC}, and $A_{n,k}$ from \eqref{eq_Amu}. Set
\begin{equation}\label{eq_MRmu}
    \begin{aligned}
        R_{k}=2C_{k}(k+1)-1,\; M_{n,k}\triangleq 18(S_{n,k}+A_{n,k}),
    \end{aligned}
\end{equation}
where $S_{n,k}=(Q-1)\parenv{\floorenv{n/(2k)}+2}
+B(n+2)+C_{k}(n+4k)$.

Recall that $\cD_{q,k}(n)$ is the set of all $\parenv{\bmp_{k},\ell_{n,k}}$-dense sequences in $\Sigma_q^n$.
\begin{lemma}\label{lem_dEDeltaVT}
    Suppose $\bmx\in\cD_{q,k}(n)$.  Given $\bmy$ and the three residues
   $$
       m_k(\bmx)\bmod{5},\quad
       \lwt{\bmF_k(\bmx)}\bmod{R_k},\quad
       \vt(\bmF_k(\bmx))\bmod{M_{n,k}},
   $$
one can determine $d$, $E$, and $\Delta_{\rm VT}$ uniquely.
\end{lemma}
\begin{IEEEproof}
By \Cref{lem_segmentreplacement}~(iii), $\abs{d}\le 2$.  Now the congruence $m_{k}(\bmx)\pmod{5}$ determines $d$ uniquely because $m_{k}(\bmy)$ is known. \Cref{eq_Fsumdiff} gives
$$
\abs{E}\le C_k(k+1)-1=\frac{R_k-1}{2}.
$$
Thus knowing $\lwt{\bmF_{k}\parenv{\bmx}}\pmod{R_{k}}$ is sufficient to obtain the value of $E$.

By \Cref{lem_segmentreplacement}~(iii), the sequence $\bmF_{k}(\bmy)$ is obtained from $\bmF_{k}(\bmx)$ by a $3$-substring edit. Moreover, \Cref{lem_segmentreplacement}~(vi) says that both $\bmF_{k}(\bmx)$ and $\bmF_{k}(\bmy)$ are $A_{n,k}$ bounded. Then it follows from \Cref{lem_absVT} that $\abs{\Delta_{\rm VT}}<9\parenv{\lwt{\bmF_{k}(\bmx)}+A_{n,k}}$. This implies that knowing $\vt\parenv{\bmF_{k}\parenv{\bmx}}\pmod{18\parenv{\lwt{\bmF_{k}(\bmx)}+A_{n,k}}}$ is sufficient to obtain the value of $\Delta_{\rm VT}$.

It remains to prove $\lwt{\bmF_{k}(\bmx)}\le S_{n,k}$. Since every segment has length at least $2k$ and $\abs{\widehat{\bmx}}=\abs{\bmx}+4k=n+4k$, we have $m_{k}\parenv{\widehat{\bmx}}\le\floorenv{(n+4k)/(2k)}$. Moreover, it holds that $N_{k}\parenv{\widehat{\bmx}}=1+N_{k}\parenv{\bmp_{k}\bmx}\le n+2$. Then it follows from \eqref{eq_fmu}, \eqref{eq_Fmu} and \eqref{eq_Nmuwhole} that
\begin{align*}
    \lwt{\bmF_{k}(\bmx)}&=\sum_{i=1}^{m_{k}\parenv{\widehat{\bmx}}}D_q\parenv{\widehat{\bmx}^{(k)}_i}+B\sum_{i=1}^{m_{k}\parenv{\widehat{\bmx}}}N_{k}\parenv{\widehat{\bmx}^{(k)}_i}+C_{k}\abs{\widehat{\bmx}}\\
    &\le m_{k}\parenv{\widehat{\bmx}}(Q-1)+B\sum_{i=1}^{m_{k}\parenv{\widehat{\bmx}}}N_{k}\parenv{\widehat{\bmx}^{(k)}_i}+C_{k}(n+4k)\\
    &=m_{k}\parenv{\widehat{\bmx}}(Q-1)+BN_{k}\parenv{\widehat{\bmx}}+C_{k}(n+4k)\\
    &\le\parenv{\floorenv{\frac{n}{2k}}+2}(Q-1)+B(n+2)+C_{k}(n+4k)\\
    &=S_{n,k}.
\end{align*}
Now the proof is completed.
\end{IEEEproof}

\begin{lemma}\label{lem_Gvalue}
If $m_k(\bmx)=m_k(\bmy)$, then $\abs{\bmF_k(\bmx)_i-\bmF_k(\bmy)_i}\le \Lambda_{q,k}$ for every
$i\in[m_k(\widehat{\bmx})]$, where
\begin{equation}\label{eq_Lambdaqk}
 \Lambda_{q,k}=(Q-1)+B(4k-2)+C_k(3k-2).
\end{equation}
Consequently, for $r\in\{0,1\}$, the sequence $\bmy$ and the residue $G_r(\bmF_k(\bmx))\bmod(2\Lambda_{q,k}+1)$ determine
$\Delta_{G_r}$ uniquely.
\end{lemma}
\begin{IEEEproof}
If $\bmx=\bmy$, every claimed difference is zero. Assume henceforth that the edit is nontrivial. Recall the definition of $\cD_{\rm es}\parenv{\widehat{\bmx}}$ and $\cC_{\rm re}\parenv{\widehat{\bmy}}$ from \Cref{def_descrepattern}. Denote $r=\abs{\cD_{\rm es}\parenv{\widehat{\bmx}}}$ and $r^\prime=\abs{\cC_{\rm re}\parenv{\widehat{\bmy}}}$. The equality $m_k(\bmx)=m_k(\bmy)$ and
\Cref{lem_patternchange} leave two cases: $r=r^\prime=0$ or $r=r^\prime=1$.

Suppose $r=r^\prime=0$. In this case, there is some
$j\in[m_k(\widehat{\bmx})]$, such that
$\widehat{\bmx}^{(k)}_i=\widehat{\bmy}^{(k)}_i$ if and only if $i\ne j$, implying $\abs{\bmF_k(\bmx)_i-\bmF_k(\bmy)_i}=0$ for every $i\ne j$. Let $\widehat{\bmx}_{[i_1,i_2]}$ be an occurrence of $1^{k}$ in $\widehat{\bmx}^{(k)}_j$. This $1^{k}$ can be destroyed by the $k$-substring edit only if $[i_1,i_2]$ intersects (in the case $\abs{\bms}>0$) or crosses (in the case $\bms=\epsilon$) the edit region in $\widehat{\bmx}$. Therefore, there are at most $\abs{\bms}+k-1$ occurrences of $1^{k}$ in $\widehat{\bmx}$ that can be destroyed. Similarly, there are at most $\abs{\bmt}+k-1$ occurrences of $1^{k}$ in $\widehat{\bmy}$ that can be newly created. Then it follows that $\abs{N_{k}\parenv{\widehat{\bmx}^{(k)}_j}-N_{k}\parenv{\widehat{\bmy}^{(k)}_j}}\le\max\mathset{\abs{\bms},\abs{\bmt}}+k-1\le 2k-1$. Now combining \eqref{eq_fmu} and \eqref{eq_Fmu} gives
\begin{align*}
    \abs{\bmF_k(\bmx)_j-\bmF_k(\bmy)_j}&\le\abs{D_q\parenv{\widehat{\bmx}^{(k)}_j}-D_q\parenv{\widehat{\bmy}^{(k)}_j}}+B\abs{N_{k}\parenv{\widehat{\bmx}^{(k)}_j}-N_{k}\parenv{\widehat{\bmy}^{(k)}_j}}+C_{k}\abs{\abs{\widehat{\bmx}^{(k)}_j}-\abs{\widehat{\bmy}^{(k)}_j}}\\
    &\le Q-1+B(2k-1)+C_{k}k<\Lambda_{q,k}.
\end{align*}

Now suppose $r=r^\prime=1$. Write $\cD_{\rm es}(\widehat{\bmx})=\{a\}$ and
$\cC_{\rm re}(\widehat{\bmy})=\{b\}$. Let $\alpha$ and $\beta$ be defined as in the proof of \Cref{lem_segmentreplacement}~(ii). Set $d=\abs{\bmt}-\abs{\bms}$. The proof of \Cref{lem_segmentreplacement}~(ii) shows that $\bms$ is contained in the substring $\widehat{\bmx}_{\sparenv{\alpha,a-1}}\widehat{\bmx}_{\sparenv{a,\beta-1}}$, and $\bmt$ is contained in the substring $\widehat{\bmy}_{\sparenv{\alpha,b-1}}\widehat{\bmy}_{\sparenv{b,\beta+d-1}}$. Suppose that $\widehat{\bmx}_{\sparenv{\alpha,a-1}}$ is the $j$-th segment of $\widehat{\bmx}$.  \Cref{lem_segmentreplacement}~(ii) claims that $\widehat{\bmy}_{\sparenv{\alpha,b-1}}$ is the $j$-th segment of $\widehat{\bmy}$. Then it holds that $\widehat{\bmx}^{(k)}_i=\widehat{\bmy}^{(k)}_i$ whenever $i\notin\{j,j+1\}$, yielding $\abs{\bmF_k(\bmx)_i-\bmF_k(\bmy)_i}=0$ for every $i\notin\{j,j+1\}$.

To bound $\abs{\bmF_k(\bmx)_j-\bmF_k(\bmy)_j}$ and $\abs{\bmF_k(\bmx)_{j+1}-\bmF_k(\bmy)_{j+1}}$, we need to bound the four quantities: $\abs{\abs{\widehat{\bmx}_{\sparenv{\alpha,a-1}}}-\abs{\widehat{\bmy}_{\sparenv{\alpha,b-1}}}}$, $\abs{\abs{\widehat{\bmx}_{\sparenv{a,\beta-1}}}-\abs{\widehat{\bmy}_{\sparenv{b,\beta+d-1}}}}$, $\abs{N_{k}\parenv{\widehat{\bmx}_{\sparenv{\alpha,a-1}}}-N_{k}\parenv{\widehat{\bmy}_{\sparenv{\alpha,b-1}}}}$ and $\abs{N_{k}\parenv{\widehat{\bmx}_{\sparenv{a,\beta-1}}}-N_{k}\parenv{\widehat{\bmy}_{\sparenv{b,\beta+d-1}}}}$. 

We first bound $\abs{\abs{\widehat{\bmx}_{\sparenv{\alpha,a-1}}}-\abs{\widehat{\bmy}_{\sparenv{\alpha,b-1}}}}$ and $\abs{\abs{\widehat{\bmx}_{\sparenv{a,\beta-1}}}-\abs{\widehat{\bmy}_{\sparenv{b,\beta+d-1}}}}$.
To that end, let $\rho=\abs{\bmp_{k}\bmu}$ (see \eqref{eq_extendededit}). The definition of $\cD_{\rm es}(\widehat{\bmx})$ and $\cC_{\rm re}(\widehat{\bmy})$ gives $\rho-2k+2\le a\le\rho+\abs{\bms}$ and $\rho-2k+2\le b\le\rho+\abs{\bmt}$, respectively. In other words,
\begin{equation}\label{eq_abinterval}
 a-\abs{\bms}, b-\abs{\bmt}\in\sparenv{\rho-3k+2,\rho}.
\end{equation}
It is easy to see that
\begin{equation}\label{eq_seglengthdiff1}
    \abs{\abs{\widehat{\bmx}_{\sparenv{\alpha,a-1}}}-\abs{\widehat{\bmy}_{\sparenv{\alpha,b-1}}}}=\abs{a-b}\le 3k-2.
\end{equation}
For the second quantity, notice that
\begin{equation}\label{eq_seglengthdiff2}
\begin{aligned}
    \abs{\abs{\widehat{\bmx}_{\sparenv{a,\beta-1}}}-\abs{\widehat{\bmy}_{\sparenv{b,\beta+d-1}}}}=\abs{(\beta+d-b)-(\beta-a)}=\abs{a+d-b}=\abs{a-\abs{\bms}-(b-\abs{\bmt})}\le 3k-2,
\end{aligned}
\end{equation}
where the last inequality follows from \eqref{eq_abinterval}.

Next, we bound $\abs{N_{k}\parenv{\widehat{\bmx}_{\sparenv{\alpha,a-1}}}-N_{k}\parenv{\widehat{\bmy}_{\sparenv{\alpha,b-1}}}}$ and $\abs{N_{k}\parenv{\widehat{\bmx}_{\sparenv{a,\beta-1}}}-N_{k}\parenv{\widehat{\bmy}_{\sparenv{b,\beta+d-1}}}}$. Set $\lambda=\min\{a,b,\rho+1\}$. By \eqref{eq_extendededit}, the two segments $\widehat{\bmx}_{\sparenv{\alpha,a-1}}$ and $\widehat{\bmy}_{\sparenv{\alpha,b-1}}$
have the common prefix $\widehat{\bmx}_{[\alpha,\lambda-1]}=\widehat{\bmy}_{[\alpha,\lambda-1]}$. After this prefix is removed, their respective lengths are $a-\lambda$ and $b-\lambda$. By the choice of $\lambda$, we have
\begin{align*}
 0\le a-\lambda
 &\le\max\{|a-b|,a-(\rho+1),0\}\le 3k-2,\\
 0\le b-\lambda
 &\le\max\{|a-b|,b-(\rho+1),0\}\le 3k-2,
\end{align*}
where the last inequalities follow from \eqref{eq_abinterval} and the fact $\abs{\bms},\abs{\bmt}\le k$. An occurrence of $1^{k}$ in $\widehat{\bmx}_{[\alpha,a-1]}$ can be destroyed by the $k$-substring edit only if it overlaps with $\widehat{\bmx}_{[\lambda,a-1]}$, and an occurrence of $1^{k}$ in $\widehat{\bmy}_{[\alpha,b-1]}$ is created by the $k$-substring edit only if it overlaps with $\widehat{\bmy}_{[\lambda,b-1]}$. Thus, it holds that
\begin{equation}\label{eq_Nmudiff1}
    \abs{N_{k}\parenv{\widehat{\bmx}_{\sparenv{\alpha,a-1}}}-N_{k}\parenv{\widehat{\bmy}_{\sparenv{\alpha,b-1}}}}\le\max\mathset{a-\lambda,b-\lambda}+k-1<4k-2.
\end{equation}

The two segments $\widehat{\bmx}_{[a,\beta-1]}$ and $\widehat{\bmy}_{[b,\beta+d-1]}$ have the common suffix $\widehat{\bmx}_{[\rho+\abs{\bms}+1,\beta-1]}=\widehat{\bmy}_{[\rho+\abs{\bmt}+1,\beta+d-1]}$. After this suffix is removed, their respective lengths are $\rho+\abs{\bms}-a+1$ and $\rho+\abs{\bmt}-b+1$. It follows from \eqref{eq_abinterval} that $\rho+\abs{\bms}-a+1,\;\rho+\abs{\bmt}-b+1\le 3k-1$. Then by the same argument in the previous paragraph, we have
\begin{equation}\label{eq_Nmudiff2}
    \abs{N_{k}\parenv{\widehat{\bmx}_{\sparenv{a,\beta-1}}}-N_{k}\parenv{\widehat{\bmy}_{\sparenv{b,\beta+d-1}}}}\le 4k-2.
\end{equation}

Now combining \eqref{eq_fmu}, \eqref{eq_Fmu}, \eqref{eq_seglengthdiff1} and \eqref{eq_Nmudiff1} gives
\begin{align*}
    \abs{\bmF_k(\bmx)_j-\bmF_k(\bmy)_j}&\le\abs{D_q\parenv{\widehat{\bmx}^{(k)}_j}-D_q\parenv{\widehat{\bmy}^{(k)}_j}}+B\abs{N_{k}\parenv{\widehat{\bmx}^{(k)}_j}-N_{k}\parenv{\widehat{\bmy}^{(k)}_j}}+C_{k}\abs{\abs{\widehat{\bmx}^{(k)}_j}-\abs{\widehat{\bmy}^{(k)}_j}}\\
    &< Q-1+B(4k-2)+C_{k}(3k-2)=\Lambda_{q,k}.
\end{align*}
Similarly, the bound $\abs{\bmF_k(\bmx)_{j+1}-\bmF_k(\bmy)_{j+1}}\le\Lambda_{q,k}$ follows from \eqref{eq_fmu}, \eqref{eq_Fmu}, \eqref{eq_seglengthdiff2} and \eqref{eq_Nmudiff2}.

The preceding argument shows that $\abs{G_r\parenv{\bmF_{k}(\bmx)}-G_r\parenv{\bmF_{k}(\bmy)}}\le\Lambda_{q,k}$ for $r\in\{0,1\}$. Then given $\bmy$ and the residue $G_r\parenv{\bmF_{k}(\bmx)}\bmod{2\Lambda_{q,k}+1}$, one can compute the residue $\parenv{G_r\parenv{\bmF_{k}(\bmx)}-G_r\parenv{\bmF_{k}(\bmy)}}\bmod{2\Lambda_{q,k}+1}$ and then obtain the exact value of $G_r\parenv{\bmF_{k}(\bmx)}-G_r\parenv{\bmF_{k}(\bmy)}$. This completes the proof.
\end{IEEEproof}

\bigskip
Recall $R_k$, $M_{n,k}$ from \eqref{eq_MRmu} and $\Lambda_{q,k}$ from \eqref{eq_Lambdaqk}. Define the syndrome
\begin{equation}\label{eq_LocSyn}
\begin{aligned}
 \mathsf{Loc}_{\rm edit}(\bmx)=\big(&m_k(\bmx)\bmod 5,
 \vt(\bmF_k(\bmx))\bmod{M_{n,{k}}},
 \lwt{\bmF_k(\bmx)}\bmod{R_{k}},\\
 &G_0(\bmF_k(\bmx))\bmod{2\Lambda_{q,k}+1},
 G_1(\bmF_k(\bmx))\bmod{2\Lambda_{q,k}+1}\big).
\end{aligned}
\end{equation}

The main result of this subsection is presented below.
\begin{theorem}\label{thm_specificlocator}
Let $n\ge\ell_{n,k}$ be a sufficiently large integer, where $\ell_{n,k}$ is given in \eqref{eq_ellnk}.
Suppose $\bmx\in\cD_{q,k}(n)$ and $\bmx\to_{k}\bmy$. From $\bmy$ and $\mathsf{Loc}_{\rm edit}(\bmx)$, one can first determine whether $\bmy=\bmx$. If $\bmy\ne\bmx$, one can compute an interval in $[n]$ covering the edit region in $\bmx$. The interval has length at most
\begin{equation}\label{eq_specificinterval}
\begin{aligned}
    W_{n,k}&=\max\mathset{15\ell_{n,k}+5k,
 \left\lceil\frac{36A_{n,k}}{C_k-\frac{C_k(k+1)-1}{2k}}\right\rceil
 +3(3\ell_{n,k}+k)}+2k\\
 &=O_{q,k}\parenv{\log n},
\end{aligned}
\end{equation}
where $C_{k}$ and $A_{n,k}$ are defined in \eqref{eq_TandC} and \eqref{eq_Amu}, respectively. Computing the interval $I$ from $\bmy$ and $\mathsf{Loc}_{\rm edit}(\bmx)$ requires $O_{q,k}\parenv{n(\log n)^2}$ bit operations.
\end{theorem}
\begin{IEEEproof}
Let $d,E,\Delta_{\rm VT},\Delta_{G_0},\Delta_{G_1}$ be the differences in \eqref{eq_fivediff}. Write $(s_m,s_{\rm VT},s_{\rm wt},s_0,s_1)
 =\mathsf{Loc}_{\rm edit}(\bmx)$. We divide the process of deciding whether $\bmy=\bmx$ and finding the claimed interval into the four steps below. The complexity of each step is analyzed.
\begin{enumerate}[$(i)$]
\item \textbf{Compute $m_k(\bmy)$, $\bmF_k(\bmy)$, $d$, and $E$.}

Form $\widehat{\bmy}=\bmp_k\bmy\bmp_k$, whose length is at most $n+5k$. Scan $\widehat{\bmy}$ to find all occurrences of $\bmp_k=0^k1^k$ and record the pattern starts. These starts, together with the endpoints $1$ and $\abs{\widehat{\bmy}}+1$, determine the segment boundaries needed to compute $\bmF_k(\bmy)$. During the same scan, initialize a counter at zero
and increase it by one whenever an occurrence of
$\bmp_k$ is entirely contained in $\bmy$.
At the end of the scan, this counter equals
$m_k(\bmy)$. 

At each position, testing for an occurrence of $\bmp_k$ compares at most $2k$ symbols and hence costs $O_{q,k}(1)$ bit operations. Recording a pattern start or updating an occurrence counter of $\bmp_k$ costs $O_k(\log n)$ bit operations because every pattern start and every value of the counter is an integer in $[0,n+5k]$. Since there are $O_k(n)$ positions, this scan requires $O_{q,k}(n\log n)$ bit operations.

Let $\bmw=\bmF_k(\bmy)$ and $m=\abs{\bmw}$.
For each segment $\widehat{\bmy}^{(k)}_j$ used in
\eqref{eq_Fmu}, write
$\widehat{\bmy}^{(k)}_j=v_1\cdots v_{\ell_j}$.
We compute the quantities defining $w_j$ as follows.
First, obtain the length $\ell_j$ from the recorded
segment boundaries by one subtraction. This requires $O_k(\log n)$ bit operations, since these boundaries have $O_k(\log n)$ bits.

Next, initialize a count at zero. Scan $\widehat{\bmy}^{(k)}_j$ to find occurrences of $1^k$. Whenever an occurrence is found, increase the count
by one. At the end of the scan, the count equals
$N_k\parenv{\widehat{\bmy}^{(k)}_j}$. At each position, testing for an occurrence of $1^k$ compares at most $k$ symbols and costs $O_{q,k}(1)$ bit operations.
The count is at most $\ell_j\le n+5k$, so increasing
it by one costs $O_k(\log n)$ bit operations.
Thus computing $N_k\parenv{\widehat{\bmy}^{(k)}_j}$
requires $O_{q,k}(\ell_j\log n)$ bit operations.

To compute the residues in \eqref{eq_functiond},
initialize $k$ integers $e_1,\ldots,e_k$ at zero.
Read $v_1,\ldots,v_{\ell_j}$ in order. When reading $v_i$, replace $e_a$ by $(e_a+v_i)\bmod q$, where
$a=1+((i-1)\bmod k)$. The index $a$ cycles through $1,\ldots,k$. At the end of this scan, $e_a=d_a\parenv{\widehat{\bmy}^{(k)}_j}$ for every $a\in[k]$. Each $e_a$ remains in $[0,q-1]$. Since $q$ and $k$ are fixed, each update costs $O_{q,k}(1)$ bit operations. Consequently, all $k$ residues are obtained using $O_{q,k}(\ell_j)$ bit operations.

There are $m=O_k(n)$ segments used in \eqref{eq_Fmu},
and their total length is at most $n+5k$.
Consequently, computing all the segment lengths,
occurrence counts of $1^k$, and residues requires
$O_{q,k}\parenv{\parenv{m+\sum_{j=1}^{m}\ell_j}\log n} =O_{q,k}(n\log n)$ bit operations.

For each $j$, the computed residues determine $D_q\parenv{\widehat{\bmy}^{(k)}_j}=\sum_{a=1}^{k}e_aq^{a-1}$.
Since $q$ and $k$ are fixed, this calculation costs $O_{q,k}(1)$ bit operations. We then compute $w_j=D_q\parenv{\widehat{\bmy}^{(k)}_j}+B N_k\parenv{\widehat{\bmy}^{(k)}_j}+C_k\ell_j$.
The integers $N_k\parenv{\widehat{\bmy}^{(k)}_j}$ and
$\ell_j$ have $O_k(\log n)$ bits.
Then computing $w_j$ from these quantities
costs $O_{q,k}(\log n)$ bit operations because $B$ and $C_k$ are fixed. Therefore, computing all $m$ entries of $\bmw$ costs $O_{q,k}(n\log n)$ bit operations.

Since $D_q(\widehat{\bmy}^{(k)}_j)\le Q-1$ and
$N_k(\widehat{\bmy}^{(k)}_j)\le\ell_j$, we have
$0\le\lwt{\bmw}\le m(Q-1)+(B+C_k)\sum_{j=1}^{m}\ell_j\le m(Q-1)+(B+C_k)(n+5k)=O_{q,k}(n)$. Hence every partial sum of $w_1,\ldots,w_m$ has $O_{q,k}(\log n)$ bits.
Computing $\lwt{\bmw}$ therefore requires $O_{q,k}(m\log n)=O_{q,k}(n\log n)$ further bit operations.

It remains to recover $d$ and $E$. Compute the least nonnegative residues $u_d=(m_k(\bmy)-s_m)\bmod 5$ and $u_E=(s_{\rm wt}-\lwt{\bmw})\bmod R_k$. By \Cref{lem_segmentreplacement}~(iii), \eqref{eq_Fsumdiff}, and \eqref{eq_fivediff}, the integer $d$ is the unique integer in $[-2,2]$ satisfying $d\equiv u_d\pmod 5$, and $E$ is the unique integer in $[-(R_k-1)/2,(R_k-1)/2]$ satisfying $E\equiv u_E\pmod{R_k}$. Therefore, set $d=u_d$ if $u_d\le2$, and $d=u_d-5$ otherwise. Similarly, set $E=u_E$ if $u_E\le(R_k-1)/2$, and $E=u_E-R_k$ otherwise.

The recovery of $d$ and $E$ uses a fixed number of subtractions, divisions with remainder, and comparisons involving $O_{q,k}(\log n)$-bit integers. It therefore costs $O_{q,k}((\log n)^2)$ bit operations. Together with the preceding scan for pattern starts and the computation of $\bmw$ and $\lwt{\bmw}$, step~(i) requires $O_{q,k}(n\log n)$ bit operations in total.

\item \textbf{Determine whether $\bmy=\bmx$.}

If $\bmy=\bmx$, then $d=E=0$. Conversely, if $\bmy\ne\bmx$, then \Cref{cor_locatableverify} and condition~(i) of \Cref{def_locatable} imply $(d,E)\ne(0,0)$. Thus $\bmy=\bmx$ if and only if $d=E=0$. Since $\abs{d}\le2$ and $\abs{E}\le C_k(k+1)-1$, this test costs $O_{q,k}(1)$ bit operations. If $d=E=0$, the procedure stops. For the remaining steps, assume $(d,E)\ne(0,0)$.

\item \textbf{Recover $\Delta_{\rm VT}$ and, when $d=0$, $\Delta_{G_0}$ and $\Delta_{G_1}$.}

We first compute $\vt(\bmw)=\sum_{j=1}^{m}jw_j$. Step~(i) gives $m=O_k(n)$ and $\lwt{\bmw}=O_{q,k}(n)$. Therefore, $0\le\vt(\bmw)\le m\lwt{\bmw}=O_{q,k}(n^2)$.
Thus every partial sum of $w_1,2w_2,\ldots,mw_m$ has $O_{q,k}(\log n)$ bits. Each multiplication $jw_j$ costs $O_{q,k}((\log n)^2)$ bit operations, and each addition costs $O_{q,k}(\log n)$ bit operations. There are $m=O_k(n)$ terms, resulting in a total of $O_{q,k}(n(\log n)^2)$ bit operations.

Next, use the given residue
$s_{\rm VT}=\vt(\bmF_k(\bmx))\bmod M_{n,k}$
to compute
$u_{\rm VT}=(s_{\rm VT}-\vt(\bmw))\bmod M_{n,k}$,
where $0\le u_{\rm VT}<M_{n,k}$.
By \eqref{eq_fivediff}, we have
$\Delta_{\rm VT}\equiv u_{\rm VT}\pmod{M_{n,k}}$.
The proof of \Cref{lem_dEDeltaVT} gives
$\abs{\Delta_{\rm VT}}
<9(S_{n,k}+A_{n,k})=M_{n,k}/2$.
Since $M_{n,k}$ is even and $\Delta_{\rm VT}$ is
an integer, it follows that
$\Delta_{\rm VT}\in
[-M_{n,k}/2+1,M_{n,k}/2-1]$.
Thus $u_{\rm VT}\ne M_{n,k}/2$, and
\[
 \Delta_{\rm VT}=
 \begin{cases}
 u_{\rm VT},&0\le u_{\rm VT}<M_{n,k}/2,\\
 u_{\rm VT}-M_{n,k},&M_{n,k}/2<u_{\rm VT}<M_{n,k}.
 \end{cases}
\]

Since $0\le s_{\rm VT}<M_{n,k}=O_{q,k}(n)$
and $\vt(\bmw)=O_{q,k}(n^2)$,
all integers used in this computation have
$O_{q,k}(\log n)$ bits.
Computing $u_{\rm VT}$ requires one subtraction
and one reduction modulo $M_{n,k}$, costing
$O_{q,k}((\log n)^2)$ bit operations.
Selecting the corresponding value of
$\Delta_{\rm VT}$ requires one comparison and
at most one further subtraction, costing
$O_{q,k}(\log n)$ bit operations.
Therefore, recovering $\Delta_{\rm VT}$ from
$s_{\rm VT}$ and the computed value $\vt(\bmw)$
requires $O_{q,k}((\log n)^2)$ bit operations.

If $d=0$, compute $G_0(\bmw)$ and $G_1(\bmw)$ by adding $w_j$ to the sum indexed by $j\bmod3$ whenever that residue is $0$ or $1$. Each sum is at most $\lwt{\bmw}=O_{q,k}(n)$, implying that these additions require $O_{q,k}(n\log n)$ bit operations. Since $d=0$ implies $m_k(\bmx)=m_k(\bmy)$, \Cref{lem_Gvalue} applies. For $r\in\{0,1\}$, compute $ u_r=(s_r-G_r(\bmw))\bmod(2\Lambda_{q,k}+1)$.
The bound $\abs{\Delta_{G_r}}\le\Lambda_{q,k}$ gives $\Delta_{G_r}=u_r$ when $u_r\le\Lambda_{q,k}$, and $\Delta_{G_r}=u_r-(2\Lambda_{q,k}+1)$ otherwise. Since $q$ and $k$ are fixed, the integer $2\Lambda_{q,k}+1$ is independent of $n$. Each integer $s_r-G_r(\bmw)$ has
$O_{q,k}(\log n)$ bits. Reducing such an integer modulo $2\Lambda_{q,k}+1$ requires $O_{q,k}(\log n)$ bit operations. The subsequent comparisons and subtractions involve only integers with $O_{q,k}(1)$ bits and therefore cost $O_{q,k}(1)$ bit operations. Hence recovering both $\Delta_{G_0}$ and $\Delta_{G_1}$ requires $O_{q,k}(\log n)$ bit operations. If $d\ne0$, these differences are not needed in \Cref{thm_framework}. In either case, step~(iii) requires $O_{q,k}(n(\log n)^2)$ bit operations.

\item \textbf{Compute the claimed interval.}

Apply \Cref{thm_framework} to $\widehat{\bmx}$ and $\widehat{\bmy}$ with the partition and segment function from \Cref{subsec_partitionandf}, and with
$$
 \begin{gathered}
 K=3,\quad A=A_{n,k},\quad L=3\ell_{n,k}+k,\\
 C^{(1)}=C_k,\quad C^{(2)}=2k,\quad
 C^{(3)}=C_k(k+1)-1.
 \end{gathered}
$$
The common segment decompositions in \Cref{lem_segmentreplacement}~(ii) give \Cref{assump1}. Parts~(v) and~(vi) of that lemma give the segment-length bound and the $A$-bounded property, respectively. Strong locatability follows from \Cref{cor_locatableverify}, since step~(ii) has established $\bmy\ne\bmx$. Finally, parts~(i) and~(iv) of \Cref{lem_segmentreplacement} give $w_j\ge C_k\ell_j$, $\ell_j\ge2k$, and $\abs{E}\le C_k(k+1)-1$, while
\[
 C^{(1)}-\frac{C^{(3)}}{C^{(2)}}
 =\frac{(k-1)C_k+1}{2k}>0.
\]
Thus the hypotheses of \Cref{thm_framework} are satisfied. The partition of $\widehat{\bmy}$ and $\bmF_k(\bmy)$ were computed in step~(i). The required differences $d$, $E$, $\Delta_{\rm VT}$, $\Delta_{G_0}$ and $\Delta_{G_1}$ were obtained in steps~(i) and~(iii).

The original sequence in this application has length $\abs{\widehat{\bmx}}=n+4k$. Therefore, the parameter $\lambda$ in \Cref{thm_framework} is
\[
 \lambda=\ceilenv{\log\parenv{9(n+5k+1)^2A}}
 =O_{q,k}(\log n),
\]
where $A=A_{n,k}=O_{q,k}(\log n)$ follows from \eqref{eq_ellnk} and \eqref{eq_Amu}. \Cref{thm_framework} consequently computes an interval $\widehat I=[a,b]\subseteq[n+4k]$ covering the edit region in $\widehat{\bmx}$ using $O_{q,k}((n+4k)\lambda^2)=O_{q,k}(n(\log n)^2)$ bit operations. Substituting the above parameters into \eqref{eq_geninterval}, and using $5L=15\ell_{n,k}+5k$, gives $\abs{\widehat I}\le W_{n,k}$.

To obtain an interval in the coordinates of $\bmx$, define $I=[\pi(a),\pi(b)]$ where $\pi(t)=\min\mathset{n,\max\mathset{1,t-2k}}$. It is easy to verify that $I$ covers the edit region in $\bmx$. Moreover, $\pi$ is non-decreasing and satisfies $\pi(b)-\pi(a)\le b-a$. Therefore, $I\subseteq[n]$ and $\abs{I}\le\abs{\widehat I}\le W_{n,k}$. Computing these endpoints requires $O_k(\log n)$ further bit operations. Since $\ell_{n,k}$ and $A_{n,k}$ are $O_{q,k}(\log n)$ and $C_k-(C_k(k+1)-1)/(2k)$ is a positive constant, \eqref{eq_specificinterval} also gives $W_{n,k}=O_{q,k}(\log n)$.
\end{enumerate}

Steps~(i) and~(ii) cost $O_{q,k}(n\log n)$ and $O_{q,k}(1)$ bit operations, respectively. Each of steps~(iii) and~(iv) costs $O_{q,k}(n(\log n)^2)$ bit operations. Their sum is $O_{q,k}(n(\log n)^2)$. Now the proof is completed.
\end{IEEEproof}

\section{Completing the Code Construction}\label{sec_substringeditcode}
In the previous section, the edit region in $\bmx$ has been determined to within a short block. To complete the construction of our $k$-substring edit correcting code, it remains to design a code $\cC^\prime$ that can correct a $k$-substring edit error which occurs in a known short block.

In \cite[Section III]{YutingLi2025IT}, Li \emph{et al.} used a code with redundancy $4k\log m+o(\log m)$ as a component of $\cC^\prime$ (see \cite[Lemmas 15 and 16]{YutingLi2025IT}), where $m$ is the length of that code. In \Cref{sec_generalcode} below, we apply the distributed graph coloring framework in \cite{LiGabrysFarnoud2026IT} to construct a code with redundancy $4\log m+O_{q,k}\parenv{\log\log m}$ to serve as the component. Our code $\cC^\prime$ is presented in \Cref{sec_blockpartitions}. The final $k$-substring edit correcting code is given in \Cref{sec_finalsubstingcode}.

\subsection{The Component Code for Local Recovery}\label{sec_generalcode}
Let $m\ge\max\mathset{k,2}$ be an integer. For $\bmb\in\Sigma_q^m$, recall that $\cE_{1,k}(\bmb)$ is the set of sequences obtained from $\bmb$ by one $k$-substring edit. Define
\begin{equation*}
    \cB_{k}(\bmb)\triangleq\mathset{\bmc\in\Sigma_q^m:\cE_{1,k}(\bmb)\cap\cE_{1,k}(\bmc)\ne\varnothing}.
\end{equation*}
Clearly, $\bmb\in\cB_{k}(\bmb)$.
Let $S_{q,k}\triangleq\sum_{j=0}^{k}q^j$, $U_m\triangleq S_{q,k}((k+1)(m+1)-\binom{k+1}{2})$, and $D_m\triangleq U_mU_{m+k}$.

\begin{lemma}\label{lem_confusionball}
For every $\bmb\in\Sigma_q^m$, it holds that $\abs{\cB_{k}(\bmb)}\le D_m$. All members of $\cB_k(\bmb)$ can be generated explicitly in $O_{q,k}(m^3\log m)$ bit operations.
\end{lemma}
\begin{IEEEproof}
Since every $k$-substring edit is reversible, a sequence
$\bmc\in\Sigma_q^m$ belongs to $\cB_k(\bmb)$ if and only if there exists $\bmc_1\in\cE_{1,k}(\bmb)$ such that
$\bmc\in\cE_{1,k}(\bmc_1)$. Thus, we can generate $\cB_k(\bmb)$ by the following three steps: first applying every possible $k$-substring edit to $\bmb$, then applying every possible $k$-substring edit to each resulting sequence $\bmc_1$, and finally retaining only the resulting sequences of length $m$ and removing repeated sequences.

Every $\bmc_1\in\cE_{1,k}(\bmb)$ is obtained by first deleting a substring $\bms$ of length at most $k$, and then inserting a sequence $\bmt\in\Sigma_q^{*}$ of length at most $k$ in the same position. There are at most $S_{q,k}$ choices for $\bmt$ and $k+1$ choices for the length of $\bms$. For each $\abs{\bms}$, there are $m+1-\abs{\bms}$ choices for the positions of $\bms$. Therefore, we have $\abs{\cE_{1,k}(\bmb)}\le S_{q,k}\sum_{\abs{\bms}=0}^{k}(m-\abs{\bms}+1)=U_m=O_{q,k}(m)$. Since the length of $\bmc_1$ is at most $m+k$, there are $\abs{\cE_{1,k}(\bmc_1)}\le U_{m+k}=O_{q,k}(m)$ possibilities for $\bmc$ once $\bmc_1$ is fixed. Therefore, it holds that $\abs{\cB_{k}(\bmb)}\le U_mU_{m+k}=D_m=O_{q,k}(m^2)$.

Next, we analyze the bit operations needed to generate $\cB_k(\bmb)$. Constructing each $\bmc_1$ requires copying its unchanged prefix and suffix and writing the inserted symbols, which takes $O_{q,k}(m)$ bit operations. Similarly, constructing a candidate $\bmc$ from a given $\bmc_1$ takes $O_{q,k}(m+k)=O_{q,k}(m)$ bit operations. Thus, generating all the candidates $\bmc$ costs $D_m\parenv{O(m)+O(m)}=O_{q,k}(m^3)$ bit operations. 

To remove repetitions, we use merge sort to arrange the
retained length-$m$ sequences in lexicographic order. This requires at most $O(D_m\log D_m)=O_{q,k}(m^2\log m)$ string comparisons. Each comparison takes $O(m)$ bit operations because the strings have length $m$. Hence the sorting procedure takes $O_{q,k}(m^3\log m)$ bit operations. We then scan the sorted list and retain one copy of each distinct sequence. Notice that equal sequences are consecutive in the sorted list, implying that removing repeated sequences costs additional $O_{q,k}(m^3)$ bit operations. Therefore, all members of $\cB_k(\bmb)$ can be generated without repetitions in $O_{q,k}(m^3\log m)$ bit operations.
\end{IEEEproof}

Put $N=\max\{q-1,(D_m-1)(m-1)\}$. \Cref{lem_confusionball} gives
$D_m=O_{q,k}(m^2)$, implying $N=O_{q,k}(m^3)$. Let $P_m$ be the smallest prime greater than $N$. By Bertrand's postulate \cite[Section 5.2]{Shoup2008}, there is a prime strictly between $N$ and $2N$ for every integer $N>1$. Consequently, we have $N<P_m<2N$, which gives $P_m=O_{q,k}(m^3)$. The next lemma helps to find $P_m$.

\begin{lemma}(\cfcite{Agrawal2004Annals}{Theorems 4.1 and 5.1})\label{lem_primetest}
    There is an algorithm which can decide, for any given integer $u\ge2$, whether $u$ is a prime in $(\log u)^{O(1)}$ bit operations.
\end{lemma}
See also \cite{Agrawal2019Annals} for the corrected
proof of Lemma~4.3 in \cite{Agrawal2004Annals}. 

To find $P_m$, we test the integers $N+1,N+2,\ldots$ in order and stop at the first prime. Since $P_m<2N$, fewer than $N$ integers are tested, and each test requires $(\log N)^{O(1)}$ bit operations by \Cref{lem_primetest}. Therefore, $P_m$ can be found in $O(N(\log N)^{O(1)})=O_{q,k}(m^3(\log m)^{O(1)})$ bit operations.

For a prime $p$, let $\mathbb{F}_p$ denote the finite field with $p$ elements. We identify elements in $\mathbb{F}_p$ with their representatives in $[0,p-1]$. For $\alpha\in\mathbb{F}_{P_m}$ and
$\bmb=b_1\cdots b_m\in\Sigma_q^m$, define
$$
 h_{\alpha}(\bmb)=\sum_{i=1}^{m}b_i\alpha^{i-1},
$$
which is an element in $\mathbb{F}_{P_m}$. The next lemma specializes the polynomial cover-free family in \cite[Lemma~13]{LiGabrysFarnoud2026IT} and the recoloring procedure in \cite[Lemma~15]{LiGabrysFarnoud2026IT} to the confusion graph whose neighbor set at $\bmb$ is $\cB_k(\bmb)\setminus\{\bmb\}$.

\begin{lemma}\label{lem_gm}
For every $\bmb\in\Sigma_q^m$, there exists
$\alpha\in\mathbb{F}_{P_m}$ such that
\begin{equation}\label{eq_hbnehc}
 h_{\alpha}(\bmb)\ne h_{\alpha}(\bmc)
\end{equation}
for all $\bmc\in\cB_{k}(\bmb)\setminus\{\bmb\}$. Let $\alpha(\bmb)$ be the smallest $\alpha$ in $\{0,1,\ldots,P_m-1\}$ satisfying this property. Define
$$
 g_m(\bmb)=\alpha(\bmb) P_m+
 h_{\alpha(\bmb)}(\bmb)\in[0,P_m^2-1].
$$ 
Then $g_m(\bmc)\ne g_m(\bmb)$ for every
$\bmc\in\cB_k(\bmb)\setminus\{\bmb\}$. The value $g_m(\bmb)$ is computable in $O_{q,k}(m^6(\log m)^{O(1)})$ bit operations.
\end{lemma}
\begin{IEEEproof}
For $\bmc=c_1\cdots c_m\in\cB_{k}(\bmb)\setminus\{\bmb\}$, consider the polynomial
$$
 P_{\bmb,\bmc}(X)=\sum_{i=1}^{m}(b_i-c_i)X^{i-1}
 \in\mathbb{F}_{P_m}[X].
$$
Since $\bmb\ne\bmc$, there is at least one $i\in[m]$ such that $b_i\ne c_i$. It follows from the fact $P_m>q-1$ and $0<\abs{b_i-c_i}\le q-1$ that $P_{\bmb,\bmc}(X)$ is a nonzero polynomial over $\mathbb{F}_{P_m}$. The degree of $P_{\bmb,\bmc}(X)$ is at most $m-1$, implying that it has at most $m-1$ roots. Then by \Cref{lem_confusionball} and the fact $P_m>(D_m-1)(m-1)$, there is at least one element $\alpha\in\mathbb{F}_{P_m}$ such that $P_{\bmb,\bmc}(\alpha)\ne0$ for every $\bmc\in\cB_{k}(\bmb)\setminus\{\bmb\}$.

Suppose that there is some $\bmc\in\cB_{k}(\bmb)\setminus\{\bmb\}$ satisfying $g_m(\bmc)=g_m(\bmb)$. Then it must hold that $\alpha(\bmc)=\alpha(\bmb)$ and $h_{\alpha(\bmb)}(\bmc)=h_{\alpha(\bmb)}(\bmb)$, since $0\le\alpha(\bmc),\alpha(\bmb),h_{\alpha(\bmb)}(\bmc),h_{\alpha(\bmb)}(\bmb)<P_m$. This contradicts \eqref{eq_hbnehc} and therefore, $g_m(\bmc)\ne g_m(\bmb)$ for all $\bmc\in\cB_{k}(\bmb)\setminus\{\bmb\}$.

Next, we analyze the complexity of computing $g_m(\bmb)$.
First, generate $\cB_k(\bmb)$ and test
$\alpha=0,1,\ldots,P_m-1$ in increasing order.
For each $\alpha$, compute $h_{\alpha}(\bmb)$ and
$h_{\alpha}(\bmc)$ for all $\bmc\in\cB_k(\bmb)\setminus\{\bmb\}$. Stop at the first $\alpha$ for which $h_{\alpha}(\bmc)\ne h_{\alpha}(\bmb)$ for all $\bmc\in\cB_k(\bmb)\setminus\{\bmb\}$.

For given $\alpha$ and $\bmc\in\cB_k(\bmb)\setminus\{\bmb\}$, the value $h_{\alpha}(\bmc)$ can be
computed by Horner's rule (see \cite[Section 5.1]{Higham2002}), which uses $m$ multiplications and $m$ additions in $\mathbb{F}_{P_m}$. Since there are at most $P_m$ choices of $\alpha$ and at most $D_m$ sequences to evaluate for each $\alpha$, the total number of field additions and multiplications is $O(P_mD_mm)=O_{q,k}(m^6)$.

Put $\ell=\lceil\log_2 P_m\rceil$. Recall that elements in $\mathbb{F}_{P_m}$ are represented by integers in $[0,P_m-1]$. Addition modulo $P_m$ takes $O(\ell)$ bit operations, and multiplication followed by reduction modulo $P_m$ takes $O(\ell^2)$ bit operations.
Comparing two field elements takes $O(\ell)$ bit operations. Thus, all evaluations and comparisons require
$O(P_mD_mm\ell^2)=O_{q,k}(m^6(\log m)^2)$ bit operations.

Generating $\cB_k(\bmb)$ costs $O_{q,k}(m^3\log m)$ bit operations by \Cref{lem_confusionball}. The prime search costs $O_{q,k}(m^3(\log m)^{O(1)})$ bit operations, as established above. Once $\alpha(\bmb)$ has been found, $h_{\alpha(\bmb)}(\bmb)$ is already available. Then computing $g_m(\bmb)$ requires only one integer multiplication and one addition, which take $O(\ell^2)$ bit operations. Combining these bounds gives the stated complexity $O_{q,k}(m^6(\log m)^{O(1)})$.
\end{IEEEproof}

\Cref{lem_gm} says that one can efficiently recover $\bmb$ from the syndrome $g_m(\bmb)$ and any $\bmb^\prime\in\cE_{1,k}(\bmb)$. This gives a length-$m$ $k$-substring edit correcting code with redundancy $\log\parenv{P_m^2}=6\log m+O_{q,k}(1)$. We next compress the syndrome $g_m(\bmb)$ by choosing a small prime modulus, which gives a length-$m$ $k$-substring edit correcting code with redundancy $4\log m+4\log\log m+O_{q,k}(1)$.

Let $p_j$ denote the $j$-th prime number. We have the following result.
\begin{lemma}\cite[eq.~(3.13)]{Rosser1962}\label{lem_jthprime}
    It holds that $p_j<j(\ln j+\ln\ln j)$ for every integer $j\ge6$.
\end{lemma}

Set
\begin{equation*}
 T_m=1+(D_m-1)\floorenv{\log(P_m^2-1)}.
\end{equation*}
Every integer in $\sparenv{P_m^2-1}$ has at most
$\floorenv{\log(P_m^2-1)}$ distinct prime divisors, because the product of $t$ distinct primes is at least $2^t$. Consequently, among $p_1,\ldots,p_{T_m}$, there is a prime that divides none of the integers in the set
\begin{equation*}
    \mathset{\abs{g_m(\bmb)-g_m(\bmc)}:\bmc\in\cB_{k}(\bmb)\setminus\{\bmb\}}.
\end{equation*}
Let $a(\bmb)$ be the least such prime, and put
\begin{equation}\label{eq_aandr}
 r(\bmb)=g_m(\bmb)\bmod a(\bmb)
 \text{ and } A_m=p_{T_m}+1.
\end{equation}
Then $0\le r(\bmb)<a(\bmb)<A_m$. Since
$\log P_m=O_{q,k}(\log m)$ and $D_m=O_{q,k}(m^2)$, we have $T_m=O_{q,k}(m^2\log m)$. Applying \Cref{lem_jthprime} to $p_{T_m}$ gives $A_m=O_{q,k}\parenv{m^2(\log m)^2}$. Define the syndrome
\begin{equation*}
 \mathsf{H}_m(\bmb)=\parenv{a(\bmb),r(\bmb)}.
\end{equation*}

\begin{lemma}\label{lem_localrecovery}
Let $\bmb\in\Sigma_q^m$. 
\begin{enumerate}[$(i)$]
    \item The syndrome $\mathsf{H}_m(\bmb)$ can be computed in $O_{q,k}\parenv{m^8(\log m)^{O(1)}}$ bit operations.
    \item Given $\bmw\in\cE_{1,k}(\bmb)$ and $\mathsf{H}_m(\bmb)$, the sequence $\bmb$ can be uniquely recovered in $O_{q,k}\parenv{m^7(\log m)^{O(1)}}$ bit operations.
\end{enumerate}
\end{lemma}
\begin{IEEEproof}
(i) To estimate the complexity of computing $\mathsf{H}_m(\bmb)$, we estimate the complexity of each step in the computation of $\mathsf{H}_m(\bmb)$. The syndrome $\mathsf{H}_m(\bmb)$ can be computed by the following steps:
\begin{itemize}
    \item \textbf{Step 1}: Generate $\cB_k(\bmb)$ using \Cref{lem_confusionball}. This step takes $O_{q,k}(m^3\log m)$ bit operations.
    \item \textbf{Step 2}: Compute $g_m(\bmc)$ for all $\bmc\in\cB_{k}(\bmb)$, including $\bmc=\bmb$. Since $\abs{\cB_{k}(\bmb)}\le D_m=O_{q,k}(m^2)$, it follows from \Cref{lem_gm} that this step takes $O_{q,k}\parenv{m^8(\log m)^{O(1)}}$ bit operations. Moreover, each $g_m(\bmc)$ has $O_{q,k}(\log m)$ bits.
    \item \textbf{Step 3}: For each $\bmc\in\cB_k(\bmb)\setminus\{\bmb\}$, put $\Delta_{\bmc}=\abs{g_m(\bmb)-g_m(\bmc)}$. Calculating all the differences $\Delta_{\bmc}$ takes $O_{q,k}(m^2\log m)$ bit operations.
    \item \textbf{Step 4}: Generate the first $T_m$ primes. To that end, test the integers $2,3,\ldots$ in order and retain those that are prime, stopping when $T_m$ primes have been obtained. By \Cref{lem_jthprime} and the bound $T_m=O_{q,k}(m^2\log m)$, we have $p_{T_m}=O_{q,k}(m^2(\log m)^2)$. At most $p_{T_m}-1$ integers are tested.  Each of these integers has $O_{q,k}(\log m)$ bits. By \Cref{lem_primetest}, each test uses $(\log m)^{O(1)}$ bit operations. Consequently, generating the first $T_m$ primes costs $O_{q,k}\parenv{m^2(\log m)^{O(1)}}$ bit operations.
    \item \textbf{Step 5}: Determine the prime $a(\bmb)$. For $1\le j\le T_m$, compute $\Delta_{\bmc}\bmod p_j$ for every $\bmc\in\cB_k(\bmb)\setminus\{\bmb\}$. The first prime for which all these remainders are nonzero is $a(\bmb)$. The analysis preceding this lemma guarantees that such a prime occurs among $p_1,\ldots,p_{T_m}$. There are at most $T_m(D_m-1)=O_{q,k}(m^4\log m)$ remainder computations. For every $\bmc$ and every $j\in[T_m]$, both $\Delta_{\bmc}$ and $p_j$ have $O_{q,k}(\log m)$ bits. Each remainder can be computed in $O_{q,k}((\log m)^2)$ bit operations. Thus, it takes $O_{q,k}(m^4(\log m)^3)$ bit operations to determine the prime $a(\bmb)$ in this step.
    \item \textbf{Step 6}: Finally, compute $r(\bmb)=g_m(\bmb)\bmod a(\bmb)$ and output $\mathsf{H}_m(\bmb)=(a(\bmb),r(\bmb))$. This last remainder computation takes $O_{q,k}((\log m)^2)$ bit operations.
\end{itemize}
Therefore, the computation of $\mathsf{H}_m(\bmb)$ costs
\begin{align*}
    &O_{q,k}(m^3\log m)+O_{q,k}\parenv{m^8(\log m)^{O(1)}}+O_{q,k}(m^2\log m)\\
    &+O_{q,k}\parenv{m^2(\log m)^{O(1)}}+O_{q,k}(m^4(\log m)^3)+O_{q,k}((\log m)^2)\\
    =&O_{q,k}\parenv{m^8(\log m)^{O(1)}}
\end{align*}
bit operations.

\vspace{5pt}
(ii) We next describe the recovery procedure.
Let $(a,r)=\mathsf{H}_m(\bmb)$ and define $\cL(\bmw)
=\cE_{1,k}(\bmw)\cap\Sigma_q^m$. By an argument similar to the proof of \Cref{lem_confusionball}, we can show that $\abs{\cL(\bmw)}=O_{q,k}(m)$, and that $\cL(\bmw)$ can be generated without repetitions in $O_{q,k}(m^2\log m)$ bit operations.

For each $\bmc\in\cL(\bmw)$, compute $g_m(\bmc)$
and test whether $g_m(\bmc)\bmod a=r$. The sequence $\bmb$ belongs to $\cL(\bmw)$ and passes this test. Moreover, every $\bmc\in\cL(\bmw)$ belongs to
$\cB_k(\bmb)$. If a sequence $\bmc\ne\bmb$ also passed the test, then $a$ would divide $g_m(\bmb)-g_m(\bmc)$, contrary to the definition of $a=a(\bmb)$. Therefore, exactly one sequence passes the test, and this sequence is $\bmb$.

By \Cref{lem_gm}, one evaluation of $g_m$ costs $O_{q,k}\parenv{m^6(\log m)^{O(1)}}$ bit operations.
There are $O_{q,k}(m)$ candidates. Evaluating $g_m$ on all of them therefore costs
$O_{q,k}\parenv{m^7(\log m)^{O(1)}}$ bit operations.
Since $g_m(\bmc)<P_m^2$ and $a\le p_{T_m}$,
each remainder computation and comparison with $r$
takes $O_{q,k}((\log m)^2)$ bit operations.
All these tests together take
$O_{q,k}(m(\log m)^2)$ bit operations.
Therefore, the recovery procedure costs 
$$O_{q,k}(m^2\log m)+O_{q,k}\parenv{m^7(\log m)^{O(1)}}+O_{q,k}(m(\log m)^2)=O_{q,k}\parenv{m^7(\log m)^{O(1)}}$$
bit operations.
\end{IEEEproof}

\subsection{The Code for Local Recovery}\label{sec_blockpartitions}
Let $W\ge k$ be a positive integer, and put
$m=2W$ and $s=W=m/2$. For each $\sigma\in\{0,s\}$, let $r_{\sigma}$ be the unique integer in $[m]$ for which
$\sigma+n+r_{\sigma}$ is divisible by $m$. For $\bmx\in\Sigma_q^n$, define
\begin{equation*}
 \bmx^{(\sigma)}=0^{\sigma}\bmx0^{r_{\sigma}},
 \qquad
 b_\sigma=\frac{\sigma+n+r_{\sigma}}{m}.
\end{equation*}
For $j\in[b_\sigma]$, the $j$-th length-$m$ block of $\bmx^{(\sigma)}$ is
\begin{equation*}
 \bmx_j^{(\sigma)}
 \triangleq\bmx^{(\sigma)}_{[(j-1)m+1,jm]}.
\end{equation*}

The following lemma explains the use of two shifted block partitions.
\begin{lemma}\label{lem_windowcontained}
Let $[i_1,i_2]\subseteq[n]$ with $i_2-i_1<W$. There is $\sigma\in\{0,s\}$ such that the shifted interval $[\sigma+i_1,\sigma+i_2]$ lies in one length-$m$ block of $\bmx^{(\sigma)}$. In other words, there is some $h\in[b_\sigma]$ such that $\sparenv{\sigma+i_1,\sigma+i_2]\subseteq[(h-1)m+1,hm}$.
\end{lemma}
\begin{IEEEproof}
For each $\sigma\in\{0,s\}$, compute
$\lfloor(i_1+\sigma-1)/m\rfloor$ and
$\lfloor(i_2+\sigma-1)/m\rfloor$. Choose a value of $\sigma$ for which these two integers are equal (such a $\sigma$ must exist), and let $h$ be their common value plus one. Then the index $h$ satisfies our requirement.
\end{IEEEproof}

Recall $a(\cdot)$, $r(\cdot)$ and $A_m$ from \eqref{eq_aandr}. For $\sigma\in\{0,s\}$, define
\begin{equation}\label{eq_fandg}
\begin{aligned}
 f_{\sigma,m}(\bmx)=
 \sum_{j=1}^{b_\sigma}a(\bmx_j^{(\sigma)})
 \pmod{A_m},\\
 g_{\sigma,m}(\bmx)=
 \sum_{j=1}^{b_\sigma}r(\bmx_j^{(\sigma)})
 \pmod{A_m}.
\end{aligned}
\end{equation}

\begin{lemma}\label{lem_completion}
Let $\bmx\in\Sigma_q^n$ and $\bmy$ be obtained from $\bmx$ by a $k$-substring edit. Suppose that an interval
$I=[i_1,i_2]\subseteq[n]$ of length at most $W$ is known and that $I$ covers the edit region in $\bmx$. Then $\bmx$ is uniquely determined by $\bmy$ and
\begin{equation*}
 \parenv{f_{0,m}(\bmx),g_{0,m}(\bmx), f_{s,m}(\bmx),g_{s,m}(\bmx)}.
\end{equation*}
The recovery takes
$O_{q,k}(nm^7(\log m)^{O(1)})$ bit operations. In particular, if
$W=(\log n)^{O_{q,k}(1)}$, this is at most 
$n(\log n)^{O_{q,k}(1)}$ for all sufficiently large $n$.
\end{lemma}
\begin{IEEEproof}
Since the interval $I$ is known, 
we can choose $\sigma\in\{0,s\}$ and $h\in[b_\sigma]$ as in \Cref{lem_windowcontained}. Let $\bmy^{(\sigma)}=0^\sigma\bmy0^{r_{\sigma}}$. The sequence $\bmy^{(\sigma)}$ is obtained from $\bmx^{(\sigma)}$ by a $k$-substring edit. Moreover, the edit region in $\bmx^{(\sigma)}$ is covered by the block $\bmx_h^{(\sigma)}$. Put $d=\abs{\bmy}-\abs{\bmx}=\abs{\bmy}-n$.

For $j<h$, the $j$-th block lies before the edit and therefore, $\bmx_j^{(\sigma)}=\bmy^{(\sigma)}_{[(j-1)m+1,jm]}$. For $j>h$, the $j$-th block lies after the edit and is shifted by $d$, implying $\bmx_j^{(\sigma)}=\bmy^{(\sigma)}_{[(j-1)m+1+d,jm+d]}$. Thus the receiver obtains every block except $\bmx_h^{(\sigma)}$ from $\bmy^{(\sigma)}$ and can compute $a$ and $r$ for all those blocks. It then follows from \eqref{eq_fandg} that
\begin{equation}\label{eq_amrm}
 \begin{aligned}
 a\parenv{\bmx_h^{(\sigma)}}
 &\equiv f_{\sigma,m}(\bmx)
 -\sum_{j\ne h}a\parenv{\bmx_j^{(\sigma)}}
 \pmod{A_m},\\
 r\parenv{\bmx_h^{(\sigma)}}
 &\equiv g_{\sigma,m}(\bmx)
 -\sum_{j\ne h}r\parenv{\bmx_j^{(\sigma)}}
 \pmod{A_m}.
 \end{aligned}
\end{equation}
Since $a\parenv{\bmx_h^{(\sigma)}},r\parenv{\bmx_h^{(\sigma)}}\in[0,A_m-1]$, the exact values of $a\parenv{\bmx_h^{(\sigma)}}$ and $r\parenv{\bmx_h^{(\sigma)}}$ are uniquely determined by \eqref{eq_amrm}.

Let $\ell=(h-1)m+1$. Since the edit region in $\bmx^{(\sigma)}$ is covered by $[\ell,\ell+m-1]$, the sequence $\bmy^{(\sigma)}_{[\ell,\ell+m-1+d]}$ is obtained from $\bmx_h^{(\sigma)}$ by the same $k$-substring edit. By \Cref{lem_localrecovery}, sequence $\bmy^{(\sigma)}_{[\ell,\ell+m-1+d]}$ and the recovered pair $\parenv{a\parenv{\bmx_h^{(\sigma)}},r\parenv{\bmx_h^{(\sigma)}}}$ determine
$\bmx_h^{(\sigma)}$ uniquely. Inserting that block among the already known blocks and deleting the $\sigma$ left zeros and the $r_{\sigma}$ right zeros recovers $\bmx$.

We now analyze the bit complexity of the recovery procedure. By the proof of \Cref{lem_windowcontained}, determining $\sigma$ and $h$ requires only a constant number of additions and integer divisions on $O(\log(n+m))$-bit integers, and hence takes $O((\log(n+m))^2)$ bit operations.

We need to compute the syndrome $\mathsf{H}_m$ for the $b_\sigma-1$ known blocks. Since $1\le r_\sigma\le m$ and
$\sigma\in\{0,m/2\}$, we have $b_\sigma-1=\floorenv{(n+\sigma)/m}\le \frac{2n}{m}$.
By \Cref{lem_localrecovery}, computing the syndrome
$\mathsf{H}_m\parenv{\bmx_j^{(\sigma)}}=\parenv{a\parenv{\bmx_j^{(\sigma)}},r\parenv{\bmx_j^{(\sigma)}}}$
for one known block takes $O_{q,k}\parenv{m^8(\log m)^{O(1)}}$ bit operations. There are at most $2n/m$ such blocks. Therefore, computing all these syndromes takes $O_{q,k}\parenv{nm^7(\log m)^{O(1)}}$ bit operations.

To compute the two sums in \eqref{eq_amrm}, we reduce
each partial sum modulo $A_m$ after every addition.
Thus, every partial sum remains in $[0,A_m-1]$.
Since $A_m=O_{q,k}(m^2(\log m)^2)$, each addition
involves $O_{q,k}(\log m)$-bit integers. Reduction requires at most one subtraction of $A_m$. Consequently, computing both sums and performing the two final modular subtractions in \eqref{eq_amrm} takes $O_{q,k}(b_\sigma\log m)$ bit operations.

With the recovered syndrome $\parenv{a\parenv{\bmx_h^{(\sigma)}},r\parenv{\bmx_h^{(\sigma)}}}$, \Cref{lem_localrecovery}
recovers the missing block $\bmx_h^{(\sigma)}$ in $O_{q,k}\parenv{m^7(\log m)^{O(1)}}$ bit operations. The padded sequences have length $O(n+m)$. Reading the known blocks, extracting the received sequence corresponding to the missing block, and assembling the recovered sequence can therefore be performed by sequential copying in $O_{q,k}(n+m)$ bit operations. If the parameters $P_m$, $T_m$, and $A_m$ have not already been computed, their computation takes $O_{q,k}\parenv{m^3(\log m)^{O(1)}}$ bit operations, as established previously. Combining these costs gives the total recovery complexity $O_{q,k}\parenv{nm^7(\log m)^{O(1)}}$.

Finally, if $W=(\log n)^{O_{q,k}(1)}$, then
$m=2W=(\log n)^{O_{q,k}(1)}$. The factor $m^7(\log m)^{O(1)}$ is consequently bounded by $(\log n)^{O_{q,k}(1)}$ for all sufficiently large $n$. This gives the stated bound $n(\log n)^{O_{q,k}(1)}$.
\end{IEEEproof}

\subsection{The Overall Code Construction}\label{sec_finalsubstingcode}
Recall the definition of $W_{n,k}$ from \eqref{eq_specificinterval}. In \Cref{sec_blockpartitions}, set
\begin{equation*}
 W=W_{n,k},\; m=2W_{n,k},\; s=W_{n,k}.
\end{equation*}
Recall the definition of $\mathsf{Loc}_{\rm edit}(\bmx)$ from \eqref{eq_LocSyn}. Define
\begin{equation*}
\mathsf{Syn}_{\rm edit}(\bmx)=\parenv{\mathsf{Loc}_{\rm edit}(\bmx),f_{0,m}(\bmx),g_{0,m}(\bmx),
f_{s,m}(\bmx),g_{s,m}(\bmx)}.
\end{equation*}
Recall that $\cD_{q,k}(n)$ is the set of all $\parenv{\bmp_k,\ell_{n,k}}$-dense sequences in $\Sigma_q^n$. For each possible value $\theta$ of this function, define
\begin{equation*}
 \cC_{\rm edit}(n;\theta)\triangleq
 \{\bmx\in\cD_{q,k}(n):
 \mathsf{Syn}_{\rm edit}(\bmx)=\theta\}.
\end{equation*}

\begin{theorem}\label{thm_codefine}
Fix $q\ge2$ and $k\ge1$. For all sufficiently large $n$, every nonempty set $\cC_{\rm edit}(n;\theta)$ is a $k$-substring edit correcting code. The decoding uses at most $n(\log n)^{O_{q,k}(1)}$ bit operations. Moreover, there is some $\theta$ such that
$$
 \red(\cC_{\rm edit}(n;\theta))
 \le \log n+8\log\log n+8\log\log\log n+O_{q,k}(1).
$$
\end{theorem}
\begin{IEEEproof}
Let $\bmx\in\cC_{\rm edit}(n;\theta)$, and let $\bmy$ be obtained from $\bmx$ by one $k$-substring edit. The syndrome $\mathsf{Loc}_{\rm edit}(\bmx)$ is part of the fixed value $\theta$. By \Cref{thm_specificlocator}, the receiver first determines whether $\bmy=\bmx$.  If equality holds, it outputs $\bmy$. Otherwise, \Cref{thm_specificlocator} gives an interval of length at most $W_{n,k}$ containing the edit region in $\bmx$.
The remaining four entries of $\theta$ are exactly the four values required in \Cref{lem_completion}, which then recovers $\bmx$ uniquely. Therefore, the code $\cC_{\rm edit}(n;\theta)$ is a $k$-substring edit correcting code.

Next, we estimate the decoding complexity. The receiver uses \Cref{thm_specificlocator} to test for no error and, otherwise, to locate the edit. This costs $O_{q,k}(n(\log n)^2)$ bit operations. For a nontrivial edit, \Cref{lem_completion} recovers the sequence using $O_{q,k}\parenv{nm^7(\log m)^{O(1)}}$ additional bit operations. That bound includes computing the known block syndromes, recovering the missing syndrome by modular subtraction, and recovering and inserting the missing block. For no error, copying the received sequence costs $O_q(n)$ bit operations. Hence the same bound covers both cases. Substituting $m=\Theta_{q,k}(\log n)$ gives the bound $n(\log n)^{O_{q,k}(1)}$.

To bound the redundancy, we need to bound the number of possible values of $\theta$. By \eqref{eq_LocSyn}, the first component $\mathsf{Loc}_{\rm edit}(\bmx)$ has at most $5M_{n,k}R_{k}\parenv{2\Lambda_{q,k}+1}^2$ values. Each of the remaining four components is a residue modulo $A_{m}$. Hence, the number of possible values of $\theta$ is at most $5M_{n,k}R_{k}\parenv{2\Lambda_{q,k}+1}^2A_m^4$.
By \Cref{lem_pdense}, $\abs{\cD_{q,k}(n)}\ge q^{n-1}$. Therefore, the pigeonhole principle ensures that there is a choice of $\theta$ such that
$$
 \abs{\cC_{\rm edit}(n;\theta)}\ge\frac{q^{n-1}}{5M_{n,k}R_{k}\parenv{2\Lambda_{q,k}+1}^2A_m^4},
$$
or equivalently,
$$
 \begin{aligned}
 \red\parenv{\cC_{\rm edit}(n;\theta)}
 &\le \log q+\log\parenv{5M_{n,k}R_{k}\parenv{2\Lambda_{q,k}+1}^2}
 +4\log A_{m}\\
 &\le \log n+8\log m
 +8\log\log m+O_{q,k}(1).
 \end{aligned}
$$
Indeed, $M_{n,k}=\Theta_{q,k}(n)$, whereas $R_k$ and
$2\Lambda_{q,k}+1$ are constants for fixed $q,k$, and
$A_m=O_{q,k}(m^2(\log m)^2)$. By \eqref{eq_ellnk}, we have $\ell_{n,k}=\Theta_{q,k}(\log n)$. It follows from \eqref{eq_Amu} and \eqref{eq_specificinterval} that $A_{n,k}=\Theta_{q,k}(\log n)$ and $W_{n,k}=\Theta_{q,k}(\log n)$. Hence
$m=\Theta_{q,k}(\log n)$, and therefore
$$
 \log m=\log\log n+O_{q,k}(1),
 \;
 \log\log m=\log\log\log n+O_{q,k}(1),
$$
which proves the redundancy bound.
\end{IEEEproof}

\section{Codes Correcting One Unequal-Length \texorpdfstring{$k$}{k}-Substring Edit}\label{sec_unequal}
In \Cref{thm_codefine}, we construct a $k$-substring edit correcting code with redundancy
$\log n+O_{q,k}\parenv{\log\log n}$.
For fixed $q$ and $k$, the lower bound in
\eqref{eq_lowredfixedqandb} is $\log n-O_{q,k}(1)$.
Thus the gap between these bounds is
$O_{q,k}\parenv{\log\log n}$. 

A $k$-substring edit may preserve the sequence length. In this section, we consider only edits that change it. We seek a code that allows us to recover any codeword $\bmx$ from any $\bmy\in\cE_{1,k}(\bmx)$ satisfying
$\abs{\bmy}\ne\abs{\bmx}$. We call such a code an unequal-length $k$-substring edit correcting code, as defined below.
\begin{definition}\label{def_unequalball}
For $\bmx\in\Sigma_q^n$, define
\begin{equation}\label{uneq_eq_model}
 \cE_{1,k}^{\ne}(\bmx)\triangleq\mathset{\bmz\in\cE_{1,k}(\bmx):\abs{\bmz}\ne n}.
\end{equation}
A nonempty set $\cC\subseteq\Sigma_q^n$ is an \emph{unequal-length $k$-substring edit correcting code} if $\cE_{1,k}^{\ne}(\bmx)\cap\cE_{1,k}^{\ne}(\bmy)=\varnothing$ for any distinct $\bmx,\bmy\in\cC$.
\end{definition}

In this section, we will construct a length-$n$ unequal-length $k$-substring edit correcting code with redundancy $\log n+O_{q,k}(1)$. Recall \Cref{def_substringedit} and the representation
in \eqref{eq_editbasic}. Recall that a $(\le t)$-burst-deletion or a $(t,T)$-localized deletion is a special $k$-substring edit by setting $k=t$ and $k=T$, respectively. Our construction also gives a $(\le k)$-burst-deletion correcting code and a $(t,k)$-localized deletion correcting code.

The main result of this section is presented in \Cref{thm_uneqfinalcode}. We first outline the main idea and explain the role of the unequal-length condition in \Cref{sec_outlineuneq}. The results needed for the construction are then proved in \Cref{sec_problemreduction,sec_markerrecovery,sec_labelrecovery}.

\subsection{Outline of the Idea}\label{sec_outlineuneq}
The construction is motivated by the following observation: a periodic sequence can be recovered from its length and a prefix containing one full period.

For a sequence $\bmu\in\Sigma_q^m$ and an integer $d\in[m]$, we say that $d$ is a period of $\bmu$ if
$u_i=u_{i+d}$ for every $i\in[m-d]$. Thus, a sequence with period $d$ is determined by its length and its first $d$ symbols. For example, let $\bmu=(01)^r$ with $r\ge2$. Deleting any of the $r$ displayed copies of $01$ gives $(01)^{r-1}$. Although the deleted copy cannot be identified, the original sequence is determined once its length and the prefix $01$ are known. This example shows that a long sequence need not require much recovery information if it has a fixed small period.

To use this observation for an arbitrary sequence $\bmx\in\Sigma_q^n$, we first identify where periodicity fails. Set $\ell_{\ast}=\lcm(1,\ldots,k)$ and call a coordinate $i$ \emph{special} if $x_i\ne x_{i+\ell_{\ast}}$, where we set $x_i=0$ for $i\notin[n]$. A special coordinate is a \emph{marker} if its length-$(\ell_{\ast}+1)$ substring is lexicographically smaller than those at all other special coordinates within distance $k$. As will be shown in \Cref{lem_markerproperty}, since $\ell_{\ast}$ is divisible by every integer in $[k]$, two length-$(\ell_{\ast}+1)$ substrings starting at special coordinates within distance $k$ cannot be equal. Consequently, any two markers are at least $k+1$ apart. Moreover, every special coordinate is within a constant distance of a marker, where the constant depends only on $q$ and $k$.

We partition $0\bmx$ into segments at these markers. The properties above ensure that after removing a prefix and a suffix of suitable constant lengths from a long segment, the remaining middle of that segment has period $\ell_{\ast}$. The known prefix can therefore be extended periodically to recover the middle. Short segments are covered entirely by their recorded prefixes and suffixes. Consequently, every segment is determined by its length and its recorded prefix and suffix, which form its \emph{label}. See \Cref{lem_label} for details. Recovering $\bmx$ therefore reduces to recovering the marker positions, which determine all segment lengths, and the segment labels.

We record the marker positions of $\bmx$ in the binary sequence $\bm{\chi}_k(\bmx)$: for $i\in[n]$, $\chi_k(\bmx)_i=1$ if and only if $i$ is a marker of $\bmx$. Then recovering the marker positions is equivalent to recovering $\bm{\chi}_k(\bmx)$. Let $\bmy$ be obtained from $\bmx$ by a $k$-substring edit, and set $h=n-\abs{\bmy}$. Since each marker decision depends only on a window of constant length, the $k$-substring edit can change marker decisions only near the edit region. Outside this neighborhood, the sequence $\bm{\chi}_k(\bmx)$ has an unchanged prefix and a suffix shifted by $h$. More precisely, \Cref{lem_inducedThburst} shows that $\bm{\chi}_k(\bmy)$ is obtained from $\bm{\chi}_k(\bmx)$ by replacing $T_{\ast}$ consecutive bits by $T_{\ast}-h$ bits, where
$T_{\ast}=\ell_{\ast}+3k$. We call this replacement a \emph{$(T_{\ast},T_{\ast}-h)$-burst error}. An important observation is that $\bm{\chi}_k(\bmy)$ is \emph{$k$-sparse}: any two ones are separated by at least $k$ zeros, since any two markers of $\bmy$ are at least $k+1$ positions apart.

Therefore, to recover $\bm{\chi}_k(\bmx)$, we need to correct a $(T_{\ast},T_{\ast}-h)$-burst error for fixed $T_{\ast}$ and every $h\in[-k,k]\setminus\{0\}$ as long as the received sequence is $k$-sparse. In \Cref{prop_markerrecovery}, we construct a syndrome representable by $\log n+O_k(1)$ bits that allows this recovery. This is the first point in the construction where we use the unequal-length condition $h\ne0$. See the proof for details.

Once the marker sequence $\bm{\chi}_k(\bmx)$ is recovered, all segment boundaries and lengths are known. It remains to recover the segment labels. We do this using a constant number of sums of labels. The key property is that any two possible original sequences with the recovered marker sequence and the same received sequence $\bmy$ can differ only among at most $2U_{\ast}$ consecutive segments, where $U_{\ast}=T_{\ast}+k$. This property also relies on the condition $h\ne0$. See \Cref{lem_affectedsegments} and its proof.
 
All segment labels are sequences over $\Sigma_q$ of the same length, which depends only on $q$ and $k$. To protect the labels, we assign them in order to $2U_{\ast}$ groups, returning to the first group after the last group. For each group, we store the sum of its labels, component-wise modulo $q$. Any $2U_{\ast}$ consecutive segments belong to different groups. Hence, when comparing two possible original sequences, each group contains at most one label that differs. Equality of the group sums therefore forces all labels to agree (see \Cref{prop_labelrecovery}). The labels and the known segment lengths then determine $\bmx$.

There are only a constant number of groups, and each label has constant length for fixed $q,k$. These label checks therefore use only $O_{q,k}(1)$ bits. Together with the marker syndrome, they give a complete syndrome representable by $\log n+O_{q,k}(1)$ bits, yielding the claimed redundancy.

\subsection{Reduction of the Problem}\label{sec_problemreduction}
In this section, we will frequently use the following parameters:
\begin{equation}\label{eq_uneqpara1}
 \ell_{\ast}\triangleq\lcm(1,\ldots,k),\quad
 D_{\ast}\triangleq q^{\ell_{\ast}}(q-1),\quad B_{\ast}\triangleq k(D_{\ast}-1)+\ell_{\ast}+1.
\end{equation}
Here, $\lcm(1,\ldots,k)$ denotes the least common multiple of $1,\ldots,k$. The choice of these parameters is justified in the proofs of \Cref{lem_markerproperty,lem_label}.

There is a natural order on symbols in $\Sigma_q$: $0<1<\cdots<q-1$. This order induces the lexicographic order on sequences of equal length over $\Sigma_q$. For any $\bmx\in\Sigma_q^{n}$, set $x_i=0$ for every $i\notin[n]$. Symbols with indices outside $[n]$ are defined only for computing the markers defined below.
\begin{definition}\label{def_marker}
Let $\bmx\in\Sigma_q^n$. A coordinate $i\in\mathbb Z$ is called \emph{special} (with respect to $\bmx$) if $x_i\ne x_{i+\ell_{\ast}}$. A special coordinate $i$ is called a \emph{marker} of $\bmx$ if $\bmx_{[i,i+\ell_{\ast}]}<\bmx_{[j,j+\ell_{\ast}]}$ for every special coordinate $j$ satisfying $0<\abs{j-i}\le k$. In particular, a special coordinate is a marker if there is no other special coordinate within distance $k$ from it.
\end{definition}

\begin{definition}
Let $\bmx\in\Sigma_q^n$. We associate with it a binary sequence $\bm{\chi}_k(\bmx)=\chi_1\chi_2\cdots\chi_n$, which is defined as follows:
$$
\chi_i=
\begin{cases}
    1,\mbox{ if }i\text{ is a marker},\\
    0,\mbox{ else}.
\end{cases}
$$
We call $\bm{\chi}_k(\bmx)$ the \emph{marker sequence} of $\bmx$.
\end{definition}

Since $x_i=0$ for every $i\notin [n]$, all special coordinates belong to $[1-\ell_{\ast},n]$. In particular, comparisons near the left end must include the special coordinates in $[1-\ell_{\ast},0]$, although these coordinates are not recorded in $\bm{\chi}_k(\bmx)$. We present some examples to illustrate the above two definitions.
\begin{example}
Let $q=2$, $k=1$. Then $\ell_{\ast}=1$.
\begin{itemize}
    \item For $\bmx=001100$, the special coordinates are $2$ and $4$, since $x_2\ne x_3$ and $x_4\ne x_5$. The distance between the two special coordinates is $2>k$, and so each is a marker. Hence $\bm{\chi}_1(\bmx)=010100$.
    \item For $\bmx=100$, the special coordinates are $0$ and $1$. Since $01<10$, coordinate $0$ is a marker and coordinate $1$ is not. Only markers in $[3]$ are considered in the definition of the marker sequence, so $\bm{\chi}_1(\bmx)=000$. This example shows that a coordinate outside $[n]$ can affect whether a coordinate inside $[n]$ is a marker.
\end{itemize}
\end{example}

\begin{example}\label{uneq_ex_markers}
Let $q=3$, $k=2$, and $\bmx=012012221$. Then $\ell_{\ast}=2$. The special coordinates and their corresponding sequences are given in the following table.
$$
\begin{array}{c|c|c|c|c|c|c|c|c|c}
\hline
 i&0&1&2&3&4&5&7&8&9\\
 \hline
 \bmx_{[i,i+2]}&001&012&120&201&012&122&221&210&100\\
 \hline
\end{array}
$$
At each of $0,4,9$, the sequence $\bmx_{[i,i+2]}$ is smaller than the sequence at every other special coordinate within distance two. Each of the remaining special coordinates has a smaller such neighbor. Thus the markers are $0,4,9$, and $\bm{\chi}_2(\bmx)=000100001$. Notice that $\bmx_{[1,3]}=\bmx_{[4,6]}$, but the distance between these two special coordinates is greater than $k$.
\end{example}

\begin{definition}
    A binary sequence is called \emph{$k$-sparse} if there are at least $k$ symbols $0$ between any two symbols $1$.
\end{definition}
 
The following lemma lists the three properties of markers that are critical in our construction.
\begin{lemma}\label{lem_markerproperty}
For every $\bmx\in\Sigma_q^n$, the following statements hold.
\begin{enumerate}[$(i)$]
\item If $i$ is special, then $\bmx_{[i,i+\ell_{\ast}]}\ne\bmx_{[j,j+\ell_{\ast}]}$ for any $j\in\mathbb{Z}$ with $0<\abs{i-j}\le k$. Moreover, $\bm{\chi}_k(\bmx)$ is $k$-sparse.
\item For every special coordinate $i$, there is a marker $j\in\mathbb Z$ such that $\abs{i-j}\le k(D_{\ast}-1)$.
\item Whether $i$ is a marker is determined by $x_{i-k},\ldots,x_{i+\ell_{\ast}+k}$.
\end{enumerate}
\end{lemma}
\begin{IEEEproof}
For (i), suppose
$\bmx_{[i,i+\ell_{\ast}]}=\bmx_{[i+a,i+a+\ell_{\ast}]}$ for some integer $a$ with $\abs{a}\in[k]$, where $i$ is special. Then $x_{i+r}=x_{i+a+r}$ for every $r\in[0,\ell_{\ast}]$. By \eqref{eq_uneqpara1}, we have $a\mid\ell_{\ast}$. If $a>0$, then taking $r=0,a,\ldots,\ell_{\ast}-a$ gives $x_i=x_{i+a}=\cdots=x_{i+\ell_{\ast}}$. If $a<0$, then taking $r=\ell_{\ast},\ell_{\ast}+a,\ldots,0$ gives $x_{i+\ell_{\ast}}=x_{i+\ell_{\ast}+a}=\cdots=x_i$. Both cases contradict the assumption that $i$ is special. Thus the length-$(\ell_{\ast}+1)$ sequence corresponding to a special coordinate must be different from any length-$(\ell_{\ast}+1)$ sequence within distance at most $k$ from it.

If two markers were at distance at most $k$, then,
by \Cref{def_marker}, each of their corresponding
sequences would have to be strictly smaller than the
other, which is impossible. Therefore, any two markers
have distance at least $k+1$, implying that
$\bm{\chi}_k(\bmx)$ is $k$-sparse.

For (ii), the number of length-$(\ell_{\ast}+1)$ sequences with different first and last symbols is $q\cdot q^{\ell_{\ast}-1}(q-1)=D_{\ast}$. If a special coordinate is not a marker, then, by (i) and \Cref{def_marker}, there is another special coordinate within distance $k$ whose corresponding sequence is strictly smaller. Starting from $i$, repeatedly choose such a coordinate. The corresponding sequences strictly decrease and belong to a set of size $D_{\ast}$. Hence this procedure stops at a marker after at most $D_{\ast}-1$ choices. The total distance from $i$ is at most $k(D_{\ast}-1)$. This proves (ii).

For (iii), all sequences compared with $\bmx_{[i,i+\ell_{\ast}]}$ start in $[i-k,i+k]$ and end in $[i-k+\ell_{\ast},i+k+\ell_{\ast}]$. Thus the symbols in $[i-k,i+\ell_{\ast}+k]$ determine the special coordinates being compared and all the comparisons.
\end{IEEEproof}

Let $\bmx\in\Sigma_q^n$, and $s_1<\cdots<s_w$ be the marker coordinates of $\bmx$ in $[n]$. Set $s_0=0$ and $s_{w+1}=n+1$. For $j\in[0,w]$, define the segment $\bmu_j^x\triangleq \bmx_{[s_j,s_{j+1}-1]}$, using $x_0=0$. These segments are nonempty and satisfy $\bmu_0^x\cdots\bmu_w^x=0\bmx$. Notice that if $w=0$, then there is only one segment: $0\bmx$.

For a sequence $\bmu\in\Sigma_q^{m}$ where $m\ge1$, let $d=\min\{m,B_{\ast}\}$ and define
\begin{equation*}
 \Lab(\bmu)\triangleq
 \bmu_{[1,d]}0^{B_{\ast}-d}
 \bmu_{[m-d+1,m]}0^{B_{\ast}-d},
\end{equation*}
which is a sequence in $\Sigma_q^{2B_{\ast}}$. We call $\Lab(\bmu)$ the \emph{label} of $\bmu$. Clearly, the label records the longest prefix and suffix of length at most $B_{\ast}$, with each padded on the right to length $B_{\ast}$.

\begin{example}
For the sequence in \Cref{uneq_ex_markers}, the markers in $[9]$ are $4$ and $9$, so the segments are $\bmu_0^x=0012$, $\bmu_1^x=01222$, and $\bmu_2^x=1$. Since $B_{\ast}=37$, the label of $\bmu_0^x$ is $\Lab\parenv{\bmu_0^x}=0012\,0^{33}0012\,0^{33}$.
\end{example}

In the above example, it is easy to see that the segment can be read immediately from its label. More generally, we can establish the following lemma.
\begin{lemma}\label{lem_label}
Every segment $\bmu_j^x$ is uniquely determined by its length $\abs{\bmu_j^x}$ and its label $\Lab\parenv{\bmu_j^x}$.
\end{lemma}
\begin{IEEEproof}
Let $m=\abs{\bmu_j^x}$ and $d=\min\mathset{m,B_{\ast}}$. If $m\le 2B_{\ast}$, then $m-d\le d$. Now the conclusion is clear from the definition of $\Lab\parenv{\bmu_j^x}$. Next, assume $m>2B_{\ast}$. Then $d=B_{\ast}$.

Recall that $\bmu_j^{x}=x_{s_j}\cdots x_{s_{j+1}-1}$. We can obtain $\bmx_{\sparenv{s_j,s_j+B_{\ast}-1}}$ and $\bmx_{\sparenv{s_{j+1}-B_{\ast},s_{j+1}-1}}$ immediately from $\Lab\parenv{\bmu_j^x}$. It remains to determine the unknown part $\bmx_{\sparenv{s_j+B_{\ast},s_{j+1}-B_{\ast}-1}}$. We claim that
\begin{equation}\label{eq_periodell}
    x_{i}=x_{i-\ell_{\ast}}
\end{equation}
for every $i\in\sparenv{s_j+B_{\ast},s_{j+1}-B_{\ast}-1}$. Suppose on the contrary that there is an $i\in\sparenv{s_j+B_{\ast},s_{j+1}-B_{\ast}-1}$ such that $x_{i}\ne x_{i-\ell_{\ast}}$. Then the index $i-\ell_{\ast}$ is special with respect to $\bmx$. According to \Cref{lem_markerproperty}~(ii), there is a marker in the interval $\sparenv{i-\ell_{\ast}-k(D_{\ast}-1),i-\ell_{\ast}+k(D_{\ast}-1)}$. Since $i\in\sparenv{s_j+B_{\ast},s_{j+1}-B_{\ast}-1}$ and $B_{\ast}=k(D_{\ast}-1)+\ell_{\ast}+1$, we have $\sparenv{i-\ell_{\ast}-k(D_{\ast}-1),i-\ell_{\ast}+k(D_{\ast}-1)}\subseteq\sparenv{s_j+1,s_{j+1}-1}$. However, by the definition of $s_0,s_1,\ldots,s_{w+1}$, there is no marker in $[s_{j}+1,s_{j+1}-1]$ for any $j\in[0,w]$. This proves \eqref{eq_periodell}.

Now the substring $\bmx_{\sparenv{s_j+B_{\ast},s_{j+1}-B_{\ast}-1}}$ can be recovered from $\bmx_{\sparenv{s_j,s_j+B_{\ast}-1}}$ by \eqref{eq_periodell}, since $B_{\ast}\ge\ell_{\ast}$. This completes the proof.
\end{IEEEproof}

The case $\abs{\bmu_j^x}>2B_{\ast}$ is illustrated in \Cref{fig_uxj}. 
\begin{figure}[!t]
\centering
\begin{tikzpicture}[x=1cm,y=1cm,font=\small]
\fill[green!12] (0,0) rectangle (2.8,.8);
\fill[green!12] (8.8,0) rectangle (11.6,.8);
\fill[red!12] (2.8,0) rectangle (8.8,.8);
\draw (0,0) rectangle (11.6,.8);
\draw (2.8,0)--(2.8,.8) (8.8,0)--(8.8,.8);
\node at (1.4,.4) {recorded prefix};
\node at (5.8,.4) {$x_i=x_{i-\ell_{\ast}}$};
\node at (10.2,.4) {recorded suffix};
\node[below] at (0,-.1) {$s_j$};
\node[below] at (2.8,-.1) {$s_j+B_{\ast}-1$};
\node[below] at (8.8,-.1) {$s_{j+1}-B_{\ast}$};
\node[below] at (11.6,-.1) {$s_{j+1}-1$};
\draw[decorate,decoration={brace,amplitude=4pt}] (0,1)--(2.8,1);
\node[above] at (1.4,1.15) {$B_{\ast}$ symbols};
\draw[decorate,decoration={brace,amplitude=4pt}] (8.8,1)--(11.6,1);
\node[above] at (10.2,1.15) {$B_{\ast}$ symbols};
\end{tikzpicture}
\caption{Recovery of a segment $\bmu_j^x$ of length greater than $2B_{\ast}$. The two green parts are recorded in the label. The symbols in the red part are recovered in increasing order using $x_i=x_{i-\ell_{\ast}}$.}
\label{fig_uxj}
\end{figure}
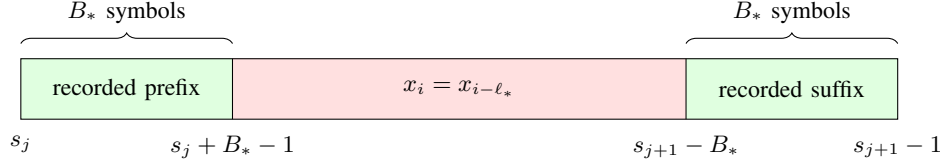

Suppose that $\bmx\in\Sigma_q^n$ and $\bmx\to_k\bmy$. Recovering $\bmx$ from $\bmy$ is equivalent to recovering the segments $\bmu_0^x,\bmu_1^x,\ldots,\bmu_w^x$ from $\bmy$. By \Cref{lem_label}, it is sufficient to recover $\abs{\bmu_j^x}$ and $\Lab\parenv{\bmu_j^x}$ for every $j\in[0,w]$. Since the lengths of these segments $\bmu_0^x,\bmu_1^x,\ldots,\bmu_w^x$ are determined by the markers $s_1,\ldots,s_w$ of $\bmx$ in $[n]$, and these markers are in turn determined by $\bm{\chi}_k(\bmx)$, recovering $\bmx$ reduces to recovering the marker sequence $\bm{\chi}_k(\bmx)$ and the labels $\Lab\parenv{\bmu_0^x},\Lab\parenv{\bmu_1^x},\ldots,\Lab\parenv{\bmu_w^x}$ from $\bmy$.

In the next subsection, we will design a code that corrects the errors in $\bm{\chi}_k(\bmx)$. The recovery of the labels $\Lab\parenv{\bmu_0^x}$, $\Lab\parenv{\bmu_1^x},\ldots,\Lab\parenv{\bmu_w^x}$ is done in \Cref{sec_labelrecovery}.

\subsection{Recovery of the Marker Sequence}\label{sec_markerrecovery}
In order to give a code that corrects the errors in $\bm{\chi}_k(\bmx)$, we need to understand how a $k$-substring edit changes the marker sequence. The following definition is necessary.
\begin{definition}\label{def_tsburst}
    Let $r_1,r_2\ge0$ be integers with $r_1+r_2>0$. Let $\bmalpha\in\Sigma^n$ where $n\ge r_1$, and $p\in[1,n-r_1+1]$. A sequence $\bmbeta\in\Sigma^{n-r_1+r_2}$ is obtained from $\bmalpha$ by an \emph{$(r_1,r_2)$-burst error} at position $p$, if $\beta_j=\alpha_j$ for every $j<p$ and $\beta_j=\alpha_{j+r_1-r_2}$ for every $j\ge p+r_2$. Thus the error replaces $r_1$ consecutive symbols by a sequence of length $r_2$ at the same position.
\end{definition}

In the rest of this section, we will use the following parameters:
\begin{equation}\label{eq_uneqpara2}
\begin{gathered}
  T_{\ast}\triangleq\ell_{\ast}+3k,\; U_{\ast}\triangleq T_{\ast}+k,\; E_{\ast}\triangleq 2U_{\ast}(U_{\ast}-1)+k(k-1),\\
 M_{n,\ast}\triangleq (U_{\ast}+k)n+E_{\ast},\; K_{\ast}\triangleq(k+1)E_{\ast}+k(k-1)+U_{\ast}.
\end{gathered}
\end{equation}
The choice of these parameters will be justified in the proofs of \Cref{lem_inducedThburst}, \Cref{lem_markervt} and \Cref{prop_markerrecovery}.

\begin{lemma}\label{lem_inducedThburst}
Let $n\ge T_{\ast}$, $\bmx\in\Sigma_q^n$ and $\bmx\to_k\bmy$. Represent $\bmx$ and $\bmy$ as in \eqref{eq_editbasic} with the convention in \eqref{eq_minimalends}. Denote $i=\abs{\bmu}+1$, $s=\abs{\bms}$ and $t=\abs{\bmt}$. Let $h=s-t$. Then $\bm{\chi}_k(\bmy)$ is obtained from $\bm{\chi}_k(\bmx)$ by a $(T_{\ast},T_{\ast}-h)$-burst error at some position $p\in[1,n-T_{\ast}+1]$. The position $p$ can be chosen so that $i\in[p,p+T_{\ast}-s]$ or equivalently, the interval $\sparenv{p,p+T_{\ast}-1}$ covers the edit region in $\bmx$.
\end{lemma}
\begin{IEEEproof}
By assumption, it holds that $y_j=x_j$ when $j<i$ and $y_j=x_{j+h}$ when $j\ge i+t$. Then by \Cref{lem_markerproperty}~(iii), we have
\begin{equation}\label{eq_chixy}
\chi_k(\bmy)_j=
\begin{cases}
   \bm{\chi}_k(\bmx)_j,\mbox{ if }j< i-\ell_{\ast}-k,\\
   \bm{\chi}_k(\bmx)_{j+h},\mbox{ if }j\ge i+t+k.
\end{cases}
\end{equation}
Let $u=\max\{1,i-\ell_{\ast}-k\}$ and $v=\min\{n,i+s+k-1\}$. Then $[u,v]=\sparenv{i-\ell_{\ast}-k,i+s+k-1}\cap[n]$ and $v-u+1\le\ell_{\ast}+2k+s\le T_{\ast}$. \Cref{eq_chixy} implies that all errors in $\bm{\chi}_k(\bmx)$ occurred in the interval $[u,v]$. Notice that $\abs{\bm{\chi}_k(\bmy)}=\abs{\bm{\chi}_k(\bmx)}-h$. Then by \Cref{def_tsburst}, to prove this lemma, it suffices to find an index $p\in[1,n-T_{\ast}+1]$, such that $[u,v]\subseteq\sparenv{p,p+T_{\ast}-1}$ and $i\in[p,p+T_{\ast}-s]$.

Set $p=\min\{u,n-T_{\ast}+1\}$. If $p=u$, then $v\le u+T_{\ast}-1=p+T_{\ast}-1$. If $p=n-T_{\ast}+1$, then $v\le n=p+T_{\ast}-1$. In either case, we have $[u,v]\subseteq[p,p+T_{\ast}-1]$.

Now it remains to show $i\in[p,p+T_{\ast}-s]$. First, notice that $u=\max\mathset{1,i-\ell_{\ast}-k}$, $p\le u$ and $i,\ell_{\ast},k\ge1$. This gives $p\le i$. By the definition of a $k$-substring edit, we have $1\le i+s\le n+1$. If $i+s\le n$, then $i+s\le v<p+T_{\ast}$, since $k\ge1$. If $i+s=n+1$, then $v=n$ and $p+T_{\ast}=n+1\ge i+s$. Now the proof is completed.
\end{IEEEproof}

By \Cref{lem_inducedThburst}, the recovery of $\bm{\chi}_k(\bmx)$ requires correction of a $(T_{\ast},T_{\ast}-h)$-burst error for every $h\in[-k,k]\setminus\{0\}$. For every fixed $r_1$ and $r_2$ with $r_1\ne r_2$, Sun \emph{et al.} constructed length-$n$ codes that can correct a $(r_1,r_2)$-burst-error \cite[Theorem 9]{Yubo2025tsburstIT}. The redundancy of this code is $\log n+O_{r_1,r_2}(1)$. Intersecting $2k$ of these codes (i.e., set $r_1=T_{\ast}$ and let $r_2$ range over $[T_{\ast}-k,T_{\ast}+k]\setminus\{T_{\ast}\}$) gives a trivial code satisfying our requirement. However, the redundancy of this trivial code is as large as $2k\log n+O_{k}(1)$.

In \Cref{prop_markerrecovery} below, we instead construct a syndrome $\mathsf{H}_n^{\ne}$ representable by $\log n+O_k(1)$ bits, which enables the recovery of $\bm{\chi}_k(\bmx)$ from $\bm{\chi}_k(\bmy)$ for every $1\le\abs{h}\le k$. This construction heavily relies on the $k$-sparse property of $\bm{\chi}_k(\bmy)$, which was proved in \Cref{lem_markerproperty}.

The next two lemmas are useful in the proof of \Cref{prop_markerrecovery}. Recall the functions $\vt(\cdot)$ and $\lwt{\cdot}$ defined in \eqref{eq_L1VT}.
\begin{lemma}\label{lem_markervt}
Suppose that $\bmbeta$ is obtained from $\bmalpha\in\Sigma_2^n$ by a $(T_{\ast},T_{\ast}-h)$-burst-error at position $p$, where $1\le\abs{h}\le k$. Put $\Delta_{L_1}=\lwt{\bmalpha}-\lwt{\bmbeta}$. Then
\begin{equation}\label{eq_markervt}
 \vt(\bmalpha)-\vt(\bmbeta)
 =p\Delta_{L_1}+h\sum_{i=p+1}^{n-h}\beta_i+\delta_p,
\end{equation}
where
\begin{equation*}
    \delta_p=
 \sum_{i=1}^{T_{\ast}-1}i\alpha_{p+i}
 -\sum_{i=1}^{T_{\ast}-h-1}(i+h)\beta_{p+i}.
\end{equation*}
Clearly, it holds that $\abs{\delta_p}\le E_{\ast}/2$.
\end{lemma}
\begin{IEEEproof}
By assumption, we have $\beta_i=\alpha_i$ for every $i<p$, and $\beta_i=\alpha_{i+h}$ for every $i\ge p+T_{\ast}-h$. Then it follows from the definition of the function $\vt$ that
\begin{equation*}
\begin{aligned}
\vt(\bmalpha)-\vt(\bmbeta)
 &=\sum_{i=p}^{p+T_{\ast}-1}i\alpha_i+\sum_{i=p+T_{\ast}}^{n}i\alpha_i-\sum_{i=p}^{p+T_{\ast}-h-1}i\beta_i-\sum_{i=p+T_{\ast}-h}^{n-h}i\beta_i\\
 &=p\parenv{\sum_{i=p}^{p+T_{\ast}-1}\alpha_i-\sum_{i=p}^{p+T_{\ast}-h-1}\beta_i}+\sum_{i=1}^{T_{\ast}-1}i\alpha_{p+i}-\sum_{i=1}^{T_{\ast}-h-1}i\beta_{p+i}+h\sum_{i=p+T_{\ast}-h}^{n-h}\beta_i\\
 &=p\Delta_{L_1}+\sum_{i=1}^{T_{\ast}-1}i\alpha_{p+i}-\sum_{i=1}^{T_{\ast}-h-1}(i+h)\beta_{p+i}+h\sum_{i=p+1}^{n-h}\beta_i\\
 &=p\Delta_{L_1}+h\sum_{i=p+1}^{n-h}\beta_i+\delta_p.
\end{aligned}
\end{equation*}

Finally, since
\begin{gather*}
   \abs{\sum_{i=1}^{T_{\ast}-1}i\alpha_{p+i}}\le\sum_{i=1}^{T_{\ast}-1}i=\binom{T_{\ast}}{2},\\
   \abs{\sum_{i=1}^{T_{\ast}-h-1}(i+h)\beta_{p+i}}\le\sum_{i=1}^{T_{\ast}+k-1}(i+k)=\frac{(T_{\ast}+3k)(T_{\ast}+k-1)}{2},
\end{gather*}
we conclude that $\abs{\delta_p}\le U_{\ast}(U_{\ast}-1)+\binom{k}{2}=E_{\ast}/2$.
\end{IEEEproof}

\begin{lemma}(\cfcite{Yubo2025tsburstIT}{Lemma 4})\label{lem_indexdistance}
Let $\bma,\bmb\in\mathbb{Z}_{\ge0}^n$, and let $K,Q$ be positive integers. Suppose that $\abs{a_i-b_i}<Q$ for every $i\in[n]$, and
$\sum_{i\in[n]:\,i\equiv r\pmod K}a_i\equiv\sum_{i\in[n]:\,i\equiv r\pmod K}b_i\pmod Q$
for every $r\in[0,K-1]$. If $\bma$ and $\bmb$ agree outside an interval containing at most $K$ coordinates, then $\bma=\bmb$.
\end{lemma}

Let $\bmalpha\in\Sigma_2^n$. If $w=\lwt{\bmalpha}\ge1$, let $L_1^{\alpha}<\cdots<L_w^{\alpha}$ be the positions of symbols $1$ in $\bmalpha$. For each $r\in[0,K_{\ast}-1]$, define
\begin{equation*}
 P_r(\bmalpha)\triangleq
 \parenv{\sum_{\substack{i\in[n]\\i\equiv r\pmod{K_{\ast}}}}\alpha_i}\bmod2,\quad
 R_r(\bmalpha)\triangleq
 \parenv{\sum_{\substack{j\in[w]\\j\equiv r\pmod{K_{\ast}}}}L_j^{\alpha}}\bmod 2U_{\ast}.
\end{equation*}
If $w=0$, define $R_r(\bmalpha)=0$ for any $\bmalpha$ and any $r$. The first sum uses the coordinate $i$ in the binary sequence. The second sum uses the index $j$ of a one in the ordered sequence of ones. Define the syndrome
\begin{equation}\label{eq_synHuneq}
 \mathsf{H}_{n}^{\ne}(\bmalpha)\triangleq
 \parenv{\vt(\bmalpha)\bmod M_{n,\ast},\
 \lwt{\bmalpha}\bmod 2U_{\ast},\
 (P_r(\bmalpha))_{r=0}^{K_{\ast}-1},\
 (R_r(\bmalpha))_{r=0}^{K_{\ast}-1}}.
\end{equation}

\begin{proposition}\label{prop_markerrecovery}
Let $n\ge T_{\ast}$ and $1\le\abs{h}\le k$. Suppose that a $k$-sparse sequence $\bmbeta\in\Sigma_2^{n-h}$ is obtained from $\bmalpha,\bmgamma\in\Sigma_2^n$ by $(T_{\ast},T_{\ast}-h)$-burst errors. If $\mathsf{H}_{n}^{\ne}(\bmalpha)=\mathsf{H}_{n}^{\ne}(\bmgamma)$, then $\bmalpha=\bmgamma$.

As a corollary, if $\bmy\in\cE_{1,k}^{\ne}(\bmx)$ for some $\bmx\in\Sigma_q^n$, then given $\mathsf{H}_{n}^{\ne}\parenv{\bm{\chi}_k(\bmx)}$, one can uniquely recover $\bm{\chi}_k(\bmx)$ from $\bm{\chi}_k(\bmy)$.
\end{proposition}
\begin{IEEEproof}
Suppose that the $(T_{\ast},T_{\ast}-h)$-burst errors in $\bmalpha$ and $\bmgamma$ occurred at positions $p$ and $r$, respectively. Without loss of generality, we may assume $1\le p\le r\le n-T_{\ast}+1$. Let $I=[p,r+T_{\ast}-1]$. Then it holds that
\begin{equation}\label{eq_albega}
    \begin{gathered}
        \bmalpha_{[n]\setminus I}=\bmgamma_{[n]\setminus I},\\
        \bmbeta_{[1,p-1]}=\bmalpha_{[1,p-1]},\\
        \beta_i=\alpha_{i+h},\text{ for }i\ge r+T_{\ast}-h.
    \end{gathered}
\end{equation}

The assumption $\lwt{\bmalpha}\equiv\lwt{\bmgamma}\pmod{2U_{\ast}}$ gives $\lwt{\bmalpha}-\lwt{\bmbeta}\equiv\lwt{\bmgamma}-\lwt{\bmbeta}\pmod{2U_{\ast}}$. Since $\bmbeta$ is obtained from both $\bmalpha$ and $\bmgamma$ by $(T_{\ast},T_{\ast}-h)$-burst-errors, we have $\lwt{\bmalpha}-\lwt{\bmbeta},\lwt{\bmgamma}-\lwt{\bmbeta}\in\sparenv{-(T_{\ast}-h),T_{\ast}}\subseteq\sparenv{-U_{\ast},U_{\ast}-1}$. Then it follows that $\lwt{\bmalpha}-\lwt{\bmbeta}=\lwt{\bmgamma}-\lwt{\bmbeta}$. Denote $\Delta_{L_1}=\lwt{\bmalpha}-\lwt{\bmbeta}$.

Let $L=r-p$ and $S=\sum_{i=p+1}^{r}\beta_i$. Applying \eqref{eq_markervt} to the two pairs $\parenv{\bmalpha,\bmbeta}$ and $\parenv{\bmgamma,\bmbeta}$ yields
\begin{equation}\label{eq_alphagamma}
    \vt(\bmgamma)-\vt(\bmalpha)=L\Delta_{L_1}-hS+\delta_r-\delta_p.
\end{equation}
Since $S\le L\le n-T_{\ast}<n$, $\abs{\Delta_{L_1}}\le U_{\ast}$ and $\abs{h}\le k<U_{\ast}$, we have
\begin{align*}
    \abs{\vt(\bmgamma)-\vt(\bmalpha)}&\le L\abs{\Delta_{L_1}}+\abs{h}S+\abs{\delta_r}+\abs{\delta_p}\\
    &<(U_{\ast}+k)n+E_{\ast}=M_{n,\ast}.
\end{align*}
Then the condition $\vt(\bmalpha)\equiv \vt(\bmgamma)\pmod{M_{n,\ast}}$ implies $\vt(\bmalpha)=\vt(\bmgamma)$. Combining this with \eqref{eq_alphagamma} gives
\begin{equation}\label{eq_LDelta}
 L\Delta_{L_1}=hS+\delta_p-\delta_r.
\end{equation}
Since $\bmbeta$ is $k$-sparse, an interval of length $L\ge1$ contains at most $S\le(L+k)/(k+1)$ ones. For $L=0$, we have $S=0$. Thus the bound $S\le(L+k)/(k+1)$ always holds.

Suppose first that $\Delta_{L_1}\ne0$. It follows from \eqref{eq_LDelta} that $L\le L\abs{\Delta_{L_1}}\le kS+E_{\ast}\le k(L+k)/(k+1)+E_{\ast}$.
Hence $L\le k^2+(k+1)E_{\ast}$ and
$\abs{I}=L+T_{\ast}\le k^2+(k+1)E_{\ast}+U_{\ast}-k=K_{\ast}$.
Recall from \eqref{eq_albega} that $\alpha_i\ne\gamma_i$ only if $i\in I$. Therefore, the condition $P_r(\bmalpha)=P_r(\bmgamma)$ for every $r\in[0,K_{\ast}-1]$ gives $\bmalpha=\bmgamma$ by \Cref{lem_indexdistance}, with $K=K_{\ast}$ and $Q=2$.

Now suppose $\Delta_{L_1}=0$. Then $\lwt{\bmalpha}=\lwt{\bmbeta}=\lwt{\bmgamma}$.  Denote $w=\lwt{\bmalpha}$. If $w=0$, then $\bmalpha=\bmgamma=0^n$. Otherwise, let $L_j^{\alpha},L_j^{\beta},L_j^{\gamma}$ be the positions of the $j$-th one in $\bmalpha,\bmbeta,\bmgamma$, respectively. Let $w_1=\lwt{\bmalpha_{[1,p-1]}}$ and $w_2=\lwt{\bmalpha_{I}}$. It follows from \eqref{eq_albega} that $\lwt{\bmbeta_{[1,p-1]}}=\lwt{\bmgamma_{[1,p-1]}}=w_1$, $\lwt{\bmbeta_{[p,r+T_{\ast}-h-1]}}=\lwt{\bmgamma_{I}}=w_2$ and $\lwt{\bmbeta_{[r+T_{\ast}-h,n-h]}}=\lwt{\bmalpha_{[r+T_{\ast},n]}}=\lwt{\bmgamma_{[r+T_{\ast},n]}}=w-w_1-w_2$. Consequently,
\begin{equation*}
\begin{gathered}
     L_j^{\beta}=L_j^{\alpha}=L_j^{\gamma},\text{ for }j\le w_1,\\
     L_j^{\beta}=L_{j}^{\alpha}-h=L_{j}^{\gamma}-h,\text{ for }j>w_1+w_2.
\end{gathered}
\end{equation*}
This gives $\abs{L_j^{\alpha}-L_j^{\gamma}}=0$ for $j\in [w_1]\cup[w_1+w_2+1,w]$. For every $j\in[w_1+1,w_1+w_2]$, since $\abs{L_j^{\alpha}-L_j^{\beta}},\abs{L_j^{\gamma}-L_j^{\beta}}<\max\mathset{T_{\ast},T_{\ast}-h}\le U_{\ast}$, it holds that $\abs{L_j^{\alpha}-L_j^{\gamma}}<2U_{\ast}$.

If $L_j^{\alpha}= L_j^{\gamma}$ for every $j\in[w]$, the conclusion follows. Otherwise, let $j_1$ and $j_2$ be the smallest and largest indices $j$ such that $L_j^{\alpha}\ne L_j^{\gamma}$, respectively. To apply \Cref{lem_indexdistance}, it remains to bound $j_2-j_1$. First, notice that $j_1,j_2\in[w_1+1,w_1+w_2]$ and thus, $j_2-j_1<w_2$. The assumption $\Delta_{L_1}= 0$ and \eqref{eq_LDelta} give $\abs{h}S\le E_{\ast}$. Since $h\ne0$, it follows that $S\le E_{\ast}$. Therefore,
$$
 w_2=\beta_p+\sum_{i=p+1}^{r}\beta_i+
 \sum_{i=r+1}^{r+T_{\ast}-h-1}\beta_i
 \le S+T_{\ast}-h\le E_{\ast}+U_{\ast}<K_{\ast}.
$$
Now the condition $R_r(\bmalpha)=R_r(\bmgamma)$ for every $r\in[0,K_{\ast}-1]$, together with \Cref{lem_indexdistance} with $Q=2U_{\ast}$ and $K=K_{\ast}$, gives $L_j^{\alpha}=L_j^{\gamma}$ for every $j\in[w]$, which implies $\bmalpha=\bmgamma$. Now the proof is completed.
\end{IEEEproof}

The following remarks clarify the hypotheses of \Cref{prop_markerrecovery}.
\begin{remark}
    \begin{itemize}
        \item The hypotheses in \Cref{prop_markerrecovery} require the common received sequence $\bmbeta$ to be $k$-sparse. The two original binary sequences need not be $k$-sparse.
        \item The condition $h\ne0$ is used to deduce $S\le E_{\ast}$ when $\Delta_{L_1}=0$. If $h=\Delta_{L_1}=0$, \eqref{eq_LDelta} gives no bound on $S$.
    \end{itemize}
\end{remark}

\subsection{Recovery of the Labels}\label{sec_labelrecovery}
By \Cref{prop_markerrecovery}, the marker sequence
$\bm{\chi}_k(\bmx)$ is uniquely determined by
$\bm{\chi}_k(\bmy)$ and the syndrome in \eqref{eq_synHuneq}. It remains to recover the labels of the segments defined in \Cref{sec_problemreduction}. We first bound the number of consecutive segment indices at which two sequences with the same marker sequence and a common received sequence can differ.

\begin{lemma}\label{lem_affectedsegments}
Let $\bmx,\bmz\in\Sigma_q^n$ where $n\ge T_{\ast}$. Suppose $\bm{\chi}_k(\bmx)=\bm{\chi}_k(\bmz)$ and $\cE_{1,k}^{\ne}(\bmx)\cap\cE_{1,k}^{\ne}(\bmz)\ne\emptyset$. Let $w=\lwt{\bm{\chi}_k(\bmx)}=\lwt{\bm{\chi}_k(\bmz)}$. Denote the respective segments of $\bmx$ and $\bmz$ by $\bmu_0^x,\ldots,\bmu_w^x$ and $\bmu_0^z,\ldots,\bmu_w^z$. Then there exist integers $0\le a\le b\le w$ such that $b-a+1\le2U_{\ast}$ and $\bmu_j^x=\bmu_j^z$ for every $j\in[0,w]\setminus[a,b]$. In other words, the two sequences $\bmx$ and $\bmz$ differ in at most $2U_{\ast}$ segments.
\end{lemma}
\begin{IEEEproof}
Let $\bmy\in\cE_{1,k}^{\ne}(\bmx)\cap\cE_{1,k}^{\ne}(\bmz)$ and $h=n-\abs{\bmy}$. Then $1\le\abs{h}\le k$.
Suppose that the $k$-substring edits in $\bmx$ and $\bmz$ start at positions $i_x$ and $i_z$ and remove $s_x$ and $s_z$ symbols, respectively, where $s_x,s_z\in[0,k]$.

Denote $\bmalpha=\bm{\chi}_k(\bmx)=\bm{\chi}_k(\bmz)$ and
$\bmbeta=\bm{\chi}_k(\bmy)$. Since $\bmy\in\cE_{1,k}^{\ne}(\bmx)$, by \Cref{lem_inducedThburst}, there exists $p\in[1,n-T_{\ast}+1]$ such that $\bmbeta$ is obtained from $\bmalpha$ by a $(T_{\ast},T_{\ast}-h)$-burst error at position $p$, where the index $p$ satisfies $i_x\in\sparenv{p,p+T_{\ast}-s_x}$. By the same argument, there exists $r\in[1,n-T_{\ast}+1]$ such that
$\bmbeta$ is obtained from $\bmalpha$ by
a $(T_{\ast},T_{\ast}-h)$-burst error at position $r$, and the index $r$ satisfies $i_z\in\sparenv{r,r+T_{\ast}-s_z}$.

Without loss of generality, assume $p\le r$. Then we have
$$
 x_j=z_j=
 \begin{cases}
  y_j,     & j\in[1,p-1],\\
  y_{j-h}, & j\in[r+T_{\ast},n].
 \end{cases}
$$
Consequently, $x_j=z_j$ for every $j\in[n]\setminus I$,
where $I\triangleq[p,r+T_{\ast}-1]$.

We next prove that $\sum_{j\in I}\alpha_j<2U_{\ast}$.
Put $L=r-p$. If $L\le U_{\ast}$, then, since $T_{\ast}=U_{\ast}-k$ and $k\ge1$, $\abs{I}=L+T_{\ast}\le2U_{\ast}-k<2U_{\ast}$. The desired bound follows because $\bmalpha$ is binary.

Now assume $L>U_{\ast}$. The $(T_{\ast},T_{\ast}-h)$-burst error at position $r$ gives $\beta_j=\alpha_j$ for $j\in[1,r-1]$, whereas the $(T_{\ast},T_{\ast}-h)$-burst
error at position $p$ gives $\beta_j=\alpha_{j+h}$ for
$j\in[p+T_{\ast}-h,n-h]$. Hence, for every $j\in[p+T_{\ast}-h,r-1]$, it holds that $\alpha_j=\beta_j=\alpha_{j+h}$. All these coordinates are valid: the bounds on $p,r$ and $T_{\ast}>\abs{h}$ imply that $j\in[n]\cap[n-h]$ and $j+h\in[n]$.
If $\alpha_j=1$, then $\alpha_{j+h}=1$. This contradicts the $k$-sparsity of $\bmalpha$ established in \Cref{lem_markerproperty}~(i), since $1\le\abs{(j+h)-j}=\abs{h}\le k$. Therefore, $\alpha_j=0$ throughout
$[p+T_{\ast}-h,r-1]$. Then it follows that
$$
 \begin{aligned}
 \sum_{j\in I}\alpha_j
 &=\sum_{j=p}^{p+T_{\ast}-h-1}\alpha_j
   +\sum_{j=r}^{r+T_{\ast}-1}\alpha_j\\
 &\le (T_{\ast}-h)+T_{\ast}
 <2U_{\ast}.
 \end{aligned}
$$
The three subintervals of $I$ used here are illustrated in
\Cref{uneq_fig_interval}.

Let $s_1<\cdots<s_w$ be the common markers in $[n]$, and put $s_0=0$ and $s_{w+1}=n+1$. Let $a,b\in[0,w]$ be the unique indices satisfying $s_a\le p<s_{a+1}$ and
$s_b\le r+T_{\ast}-1<s_{b+1}$. Then $a\le b$. The set of markers in $[p+1,r+T_{\ast}-1]$ is $\mathset{s_{a+1},\ldots,s_b}$. By definition of $\bmalpha$, we have $\alpha_j=1$ if and only if $j\in\mathset{s_1,\ldots,s_w}$. Therefore,
$$
 b-a+1
 =1+\sum_{j=p+1}^{r+T_{\ast}-1}\alpha_j
 \le1+\sum_{j\in I}\alpha_j
 \le2U_{\ast}.
$$
It follows from the choice of $a$ and $b$ that $I$ is a subset of $[s_a,s_{b+1}-1]$. Since $x_i=z_i$ for every $i\in[n]\setminus I$, we conclude that $\bmu_j^x=\bmu_j^z$ for every $j\in[0,w]\setminus[a,b]$.
This completes the proof.
\end{IEEEproof}
\begin{figure}[!t]
\centering
\begin{tikzpicture}[x=1cm,y=1cm,font=\small]
\fill[red!12] (0,0) rectangle (3,.8);
\fill[red!12] (8.7,0) rectangle (11.7,.8);
\draw (0,0) rectangle (11.7,.8);
\draw (3,0)--(3,.8) (8.7,0)--(8.7,.8);
\node at (1.5,.4) {$T_{\ast}-h\le U_{\ast}$ symbols};
\node at (5.85,.4) {all symbols are zero};
\node at (10.2,.4) {$T_{\ast}\le U_{\ast}$ symbols};
\node[below] at (0,-.1) {$p$};\node[below] at (3,-.1) {$p+T_{\ast}-h-1$};
\node[below] at (8.7,-.1) {$r$};\node[below] at (11.7,-.1) {$r+T_{\ast}-1$};
\node[above] at (5.85,1) {$\alpha_{j+h}=\beta_j=\alpha_j,\qquad 1\le\abs{h}\le k$};
\end{tikzpicture}
\caption{The common marker sequence on $[p,r+T_{\ast}-1]$ when $r-p> U_{\ast}$. Markers in this interval can occur only in the two red parts. The white part consists of zeros and can be arbitrarily long.}
\label{uneq_fig_interval}
\end{figure}
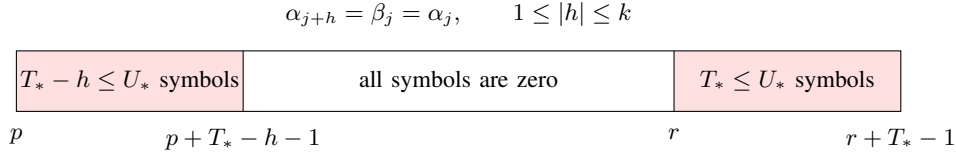

For $r\in[0,2U_{\ast}-1]$, define
\begin{equation}\label{eq_bmGr}
\bmG_{r}^{\ne}(\bmx)\triangleq\parenv{\sum_{\substack{j\in[0,w]\\ j\equiv r\pmod{2U_{\ast}}}}\Lab\parenv{\bmu_j^x}}\bmod q,
\end{equation}
where $\bmu_0^x,\ldots,\bmu_w^x$ are the segments of $0\bmx$. Each component is added over the integers and reduced to $[0,q-1]$.

\begin{proposition}\label{prop_labelrecovery}
Let $n\ge T_{\ast}$ and $\bmx,\bmz\in\Sigma_q^n$ satisfy
$\bm{\chi}_k(\bmx)=\bm{\chi}_k(\bmz)$. Denote their respective segments by $\bmu_0^x,\ldots,\bmu_w^x$ and
$\bmu_0^z,\ldots,\bmu_w^z$. Suppose that $\cE_{1,k}^{\ne}(\bmx)\cap\cE_{1,k}^{\ne}(\bmz)\ne\emptyset$ and that $\bmG_{r}^{\ne}(\bmx)=\bmG_{r}^{\ne}(\bmz)$ for every $r\in[0,2U_{\ast}-1]$. Then $\Lab(\bmu_j^x)=\Lab(\bmu_j^z)$ for every $j\in[0,w]$. 

Consequently, given $\bmy\in\cE_{1,k}^{\ne}(\bmx)$, the marker sequence $\bm{\chi}_k(\bmx)$, and
$(\bmG_{r}^{\ne}(\bmx))_{r=0}^{2U_{\ast}-1}$, all segment labels of $\bmx$ are uniquely determined.
\end{proposition}
\begin{IEEEproof}
By \Cref{lem_affectedsegments}, there exist integers
$0\le a\le b\le w$ such that $b-a+1\le 2U_{\ast}$ and
$ \bmu_j^x=\bmu_j^z$ for every $j\in[0,w]\setminus[a,b]$.
In particular, we have $\Lab\parenv{\bmu_j^x}=\Lab\parenv{\bmu_j^z}$ for every $j\in[0,w]\setminus[a,b]$.

For every $r\in[0,2U_{\ast}-1]$, the equality
$\bmG_{r}^{\ne}(\bmx)=\bmG_{r}^{\ne}(\bmz)$ and
\eqref{eq_bmGr} imply
$$
 \sum_{\substack{j\in[0,w]\\
                  j\equiv r\pmod{2U_{\ast}}}}
 \parenv{\Lab(\bmu_j^x)-\Lab(\bmu_j^z)}
 \equiv\bm{0}\pmod{q}.
$$
All vector congruences in this proof are componentwise.
The summands with indices outside $[a,b]$ are zero.
Therefore,
$$
 \sum_{\substack{j\in[a,b]\\
                  j\equiv r\pmod{2U_{\ast}}}}
 \parenv{\Lab(\bmu_j^x)-\Lab(\bmu_j^z)}
 \equiv\bm{0}\pmod{q}
$$
for every $r\in[0,2U_{\ast}-1]$. This implies $\Lab\parenv{\bmu_j^x}\equiv\Lab\parenv{\bmu_j^z}\pmod{q}$ for every $j\in[a,b]$, since $\abs{[a,b]}\le 2U_{\ast}$. Recall that $\Lab\parenv{\bmu_j^x},\Lab\parenv{\bmu_j^z}\in\Sigma_q^{2B_{\ast}}$. Now we conclude that $\Lab\parenv{\bmu_j^x}=\Lab\parenv{\bmu_j^z}$ for every $j$. Now the proof is completed.
\end{IEEEproof}

\subsection{The Code Construction}\label{sec_uneqcode}
Define the syndrome
\begin{equation}\label{eq_fullsyndrome}
 \mathsf{Syn}_{n}^{\ne}(\bmx)\triangleq
 \parenv{\mathsf{H}_{n}^{\ne}(\bm{\chi}_k(\bmx)),
 (\bmG_{r}^{\ne}(\bmx))_{r=0}^{2U_{\ast}-1}}.
\end{equation}
By \eqref{eq_synHuneq} and \eqref{eq_bmGr}, the values of this syndrome belong to
$$
\cV_{\mathsf{Syn}}\triangleq[0,M_{n,\ast}-1]\times[0,2U_{\ast}-1]\times\Sigma_2^{K_{\ast}}\times[0,2U_{\ast}-1]^{K_{\ast}}\times(\Sigma_q^{2B_{\ast}})^{2U_{\ast}}.
$$
This set has $2U_{\ast}(4U_{\ast})^{K_{\ast}}q^{4B_{\ast}U_{\ast}}M_{n,\ast}$ elements.

\begin{proposition}\label{prop_complexitySynne}
    For $n\ge1$ and any $\bmx\in\Sigma_q^n$, the syndrome $\mathsf{Syn}_{n}^{\ne}(\bmx)$ can be computed in $O_{q,k}\parenv{n\log n}$ bit operations.
\end{proposition}
The proof of this proposition is deferred to Appendix~\ref{appendix_pfSynne}.

For $\theta\in\cV_{\mathsf{Syn}}$, let
\begin{equation*}
    \cC^{\ne}(n;\theta)\triangleq\mathset{\bmx\in\Sigma_q^n:\mathsf{Syn}_{n}^{\ne}(\bmx)=\theta}.
\end{equation*}

The main result of this section is given below.
\begin{theorem}\label{thm_uneqfinalcode}
For $n\ge T_{\ast}$ and $\theta\in\cV_{\mathsf{Syn}}$, every nonempty $\cC^{\ne}(n;\theta)$ is an unequal-length $k$-substring edit correcting code. For any $\bmx\in\cC^{\ne}(n;\theta)$ and any $\bmy\in\cE_{1,k}^{\ne}(\bmx)$, we can recover $\bmx$ from $\bmy$ in $O_{q,k}\parenv{n^2\log n}$ bit operations by brute force. Moreover, there exists some $\theta$ such that $\red\parenv{\cC^{\ne}(n;\theta)}\le\log n+O_{q,k}(1)$.
\end{theorem}
\begin{IEEEproof}
Suppose $\bmx,\bmz\in\cC^{\ne}(n;\theta)$ and
$\bmy\in\cE_{1,k}^{\ne}(\bmx)\cap\cE_{1,k}^{\ne}(\bmz)$. We will prove $\bmx=\bmz$. 

Put $h=n-\abs{\bmy}$. Then $1\le\abs{h}\le k$. By \Cref{lem_markerproperty,lem_inducedThburst}, the sequence $\bm{\chi}_k(\bmy)$ is $k$-sparse and is obtained from both $\bm{\chi}_k(\bmx)$ and $\bm{\chi}_k(\bmz)$ by $(T_{\ast},T_{\ast}-h)$-burst errors. \Cref{eq_fullsyndrome} gives $\mathsf{H}_{n}^{\ne}(\bm{\chi}_k(\bmx))=\mathsf{H}_{n}^{\ne}(\bm{\chi}_k(\bmz))$. By \Cref{prop_markerrecovery}, we have $\bm{\chi}_k(\bmx)=\bm{\chi}_k(\bmz)$. Then it follows from \Cref{prop_labelrecovery} that $\Lab(\bmu_j^x)=\Lab(\bmu_j^z)$ for every $j\in\sparenv{0,\lwt{\bm{\chi}_k(\bmx)}}$, since $\bmG_{r}^{\ne}(\bmx)=\bmG_{r}^{\ne}(\bmz)$ for every $r\in[0,2U_{\ast}-1]$ by \eqref{eq_fullsyndrome}. Notice that for each $j\in\sparenv{0,\lwt{\bm{\chi}_k(\bmx)}}$, the lengths of $\bmu_j^x$ and $\bmu_j^z$ are determined by $\bm{\chi}_k(\bmx)$. Now the equality $\bmx=\bmz$ follows from \Cref{lem_label}.

We next give the decoding procedure and bound its complexity. Write $m=\abs{\bmy}=n-h$ and define
$\cT_h=\{t\in[0,k]:0\le t+h\le k\}$. For every $t\in\cT_h$, every $p\in[1,m-t+1]$, and every
$\bmv\in\Sigma_q^{t+h}$, form the candidate $\bmz=\bmy_{[1,p-1]}\bmv\bmy_{[p+t,m]}$. Each candidate has length $m-t+(t+h)=n$. Every length-$n$ sequence that can give $\bmy$ by a $k$-substring edit occurs among these candidates. In particular, the transmitted codeword $\bmx$ occurs in this list.

Since $m\le n+k$, the number of candidates is at most
\begin{equation*}
    \sum_{t\in\cT_h}(m-t+1)q^{t+h}
    \le (m+1)\sum_{s=0}^{k}q^s
    =O_{q,k}(n).
\end{equation*}
Writing one candidate costs
$O_{q,k}(n)$ bit operations. By
\Cref{prop_complexitySynne}, computing the syndrome $\mathsf{Syn}_n^{\ne}$ of each candidate
costs $O_{q,k}\parenv{n\log n}$ bit operations.
The syndrome has $O_{q,k}(\log n)$ bits, so
comparing it with $\theta$ costs at most that many
bit operations. The decoder returns the first candidate whose syndrome equals $\theta$. Such a candidate exists, and the correction property proved above ensures that it is $\bmx$. The total decoding cost is therefore $O_{q,k}\parenv{n^2\log n}$ bit operations.

Finally, we give the redundancy bound. By the pigeonhole principle, there exists $\theta\in\cV_{\mathsf{Syn}}$ such that
$$
\abs{\cC^{\ne}(n;\theta)}\ge\frac{q^n}{2U_{\ast}(4U_{\ast})^{K_{\ast}}q^{4B_{\ast}U_{\ast}}M_{n,\ast}},
$$
or equivalently,
\begin{align*}
    \red\parenv{\cC^{\ne}(n;\theta)}&\le\log\parenv{M_{n,\ast}}+\log\parenv{2U_{\ast}(4U_{\ast})^{K_{\ast}}q^{4B_{\ast}U_{\ast}}}\\
    &=\log n+O_{q,k}(1).
\end{align*}
This completes the proof.
\end{IEEEproof}

The constant term $O_{q,k}(1)$ in $\red\parenv{\cC^{\ne}(n;\theta)}$ could be very large, as shown in the following example.
\begin{example}
For $k=2$, the parameters in \eqref{eq_uneqpara1} and \eqref{eq_uneqpara2} are $\ell_{\ast}=2$, $B_{\ast}=2q^2(q-1)+1$, $T_{\ast}=8$, $U_{\ast}=10$, $E_{\ast}=182$, $K_{\ast}=558$ and $M_{n,\ast}=12n+182$.
\Cref{thm_uneqfinalcode} gives a code with redundancy at most
$\log(12n+182)+\log20+558\log40+40B_{\ast}\log q$.
\end{example}

As mentioned in \Cref{sec_relatedmodel}, a $(\le k)$-burst-deletion and a $(t,k)$-localized deletion are particular unequal-length substring edits. We therefore obtain the following consequence from \Cref{thm_uneqfinalcode}.
\begin{corollary}\label{uneq_cor_deletion}
For fixed $q\ge2$ and $k\ge t\ge1$, \Cref{thm_uneqfinalcode} gives a $(\le k)$-burst-deletion correcting code with redundancy at most $\log n+O_{q,k}(1)$ and a $(t,k)$-localized deletion correcting code with redundancy at most $\log n+O_{q,t,k}(1)$.
\end{corollary}

\section{Conclusion}\label{sec_conlcusion}
In this paper, we developed a new method to construct a $q$-ary $k$-substring edit correcting code with redundancy $\log n + 8\log\log n + o(\log\log n)$, for any fixed $q\ge2$ and $k\ge1$. This improves upon the previous binary construction of Li \emph{et al.}, whose redundancy was $\log n + 16k\log\log n + o(\log\log n)$. When the deleted substring and the inserted sequence have different lengths, we constructed $q$-ary codes with redundancy $\log n + O_{q,k}(1)$, which is optimal up to an additive constant. This gives $q$-ary $(\le t)$-burst-deletion correcting codes and $(t,T)$-localized deletion correcting codes with redundancies $\log n + O_{q,t}(1)$ and $\log n + O_{q,T}(1)$, respectively, for all fixed $q\ge2$ and $1\le t\le T$.

Several questions remain. Recall that the lower bound on the redundancy of a $k$-substring edit correcting code is $\log n-O_{q,k}(1)$ for fixed $q$ and $k$. For general $k$-substring edits, it is natural to ask whether the additive $O_{q,k}(\log\log n)$ redundancy can be reduced to $O_{q,k}(1)$. A natural idea is to extend the unequal-length construction to include the equal-length setting. However, we showed in Proposition~B.1 that a straightforward extension fails for a particular syndrome class, indicating that new ideas are required. For the unequal-length model, reducing the parameter-dependent constants in the redundancy and improving the decoding complexity would make the construction more useful at moderate lengths.

\appendices
\section{Proof of \texorpdfstring{\Cref{prop_complexitySynne}}{the proposition}}\label{appendix_pfSynne}
\begin{IEEEproof}
For fixed $q,k$, the parameters
$\ell_{\ast},D_{\ast},B_{\ast},T_{\ast},U_{\ast},E_{\ast}$,
and $K_{\ast}$ are constants, whereas
$M_{n,\ast}=\Theta_k(n)$. We compute the marker sequence,
the marker syndrome $\mathsf{H}_{n}^{\ne}$, and the label checks as follows.

First, compute $\bmalpha=\bm{\chi}_k(\bmx)$.
Pad $\bmx$ with $k$ zeros on the left and
$\ell_{\ast}+k$ zeros on the right. Scan the resulting sequence using a sliding window of length $\ell_{\ast}+2k+1$. Its $n$ successive windows correspond to $\bmx_{[i-k,i+\ell_{\ast}+k]}$ for $i\in[n]$, under the zero extension.

By \Cref{lem_markerproperty}~(iii), each window determines whether its corresponding coordinate is a marker. Perform the required tests and comparisons, and output $1$ if the coordinate is a marker and $0$ otherwise. Advance the window by removing its first symbol and appending the next symbol. Since $q,k$ are fixed, updating one window and deciding whether the corresponding coordinate is a marker cost $O_{q,k}(1)$ bit operations. There are exactly $n$ windows, so computing $\bmalpha$ costs $O_{q,k}(n)$ bit operations.

Next, compute $\mathsf{H}_{n}^{\ne}(\bmalpha)$ by scanning $\bmalpha$ from left to right. For the VT residue, maintain $V\in[0,M_{n,\ast}-1]$, initially zero, and a position counter $i$, initially one. Whenever $\alpha_i=1$, add $i$ to $V$ and subtract $M_{n,\ast}$ if the sum is at least $M_{n,\ast}$. Since $i\le n<M_{n,\ast}$, the sum is less than $2M_{n,\ast}$. Therefore, at most one subtraction is needed. After reading each bit, increase $i$ by one. The counter $i$, the residue $V$, and the temporary sum have $O_k(\log n)$ bits. Each increment, addition, comparison, or subtraction therefore costs $O_k(\log n)$ bit operations. There are $n$ increments and at most $n$ residue updates. Therefore, computing $\vt\parenv{\bmalpha}\bmod{M_{n,\ast}}$ costs $O_k\parenv{n\log n}$ bit operations.

The $2K_{\ast}+1$ residues $W:=\lwt{\bmalpha}\bmod 2U_{\ast}$, $P_r$, and $R_r$ are computed as we read $\bmalpha$ from left to right. Use $a=i\bmod K_{\ast}$ and $b=i\bmod 2U_{\ast}$ for the current position $i$. Let $c$ be the number of ones already read modulo $K_{\ast}$. Initially, set $a=b=1$, $c=W=0$, and all $P_r,R_r$ to zero. If $\alpha_i=1$, first set $c=(c+1)\bmod K_{\ast}$. Then update $W\gets(W+1)\bmod 2U_{\ast}$, $P_a\gets(P_a+1)\bmod 2$, and
$R_c\gets(R_c+b)\bmod 2U_{\ast}$. For the $j$-th one, the updated value of $c$ is $j\bmod K_{\ast}$, so its position contributes to the required $R_r$. After each bit, increase $a$ modulo $K_{\ast}$ and $b$ modulo $2U_{\ast}$. These counters and residues have $O_k(1)$ bits, and each input bit requires only a constant number
of updates. Thus the computation of these $2K_{\ast}+1$ residues costs $O_k(n)$ bit operations. Together with the VT computation, computing $\mathsf{H}_{n}^{\ne}(\bmalpha)$ costs $O_k\parenv{n\log n}$ bit operations.

Finally, compute the $2U_{\ast}$ label syndromes $\bmG_{r}^{\ne}(\bmx)$ by scanning $0\bmx$ and $\bmalpha$ in order. Initialize all group sums $\bmG_{r}^{\ne}(\bmx)$ to zero and a group counter $g$ to zero. Start the first segment with the initial symbol $0$. Whenever $\alpha_i=1$, finish the current segment before $x_i$, add its label to $\bmG_g^{\ne}(\bmx)$ componentwise modulo $q$, increase $g$ modulo $2U_{\ast}$, and start the next segment at $x_i$. After reading $x_n$, finish the last segment and add its label to the corresponding $\bmG_{r}^{\ne}(\bmx)$. For segment $\bmu_j^x$, the group counter is $g=j\bmod 2U_{\ast}$, as required.

For each segment, start with $d=0$ and two empty strings $\Lab_1$ and $\Lab_2$. Read the segment from left to right. For each new symbol, set $d=\min\{d+1,B_{\ast}\}$.
Add the symbol to the end of $\Lab_1$ if $\abs{\Lab_1}<B_{\ast}$. Also add the symbol to the end of $\Lab_2$. If this makes $\Lab_2$ longer than $B_{\ast}$, remove the first symbol of $\Lab_2$. At the end of a segment of length $m$, we have $d=\min\{m,B_{\ast}\}$, and the two length-$d$ sequences $\Lab_1$ and $\Lab_2$ are respectively the prefix and suffix of that segment.
Add $B_{\ast}-d$ zeros to the end of each of $\Lab_1$ and $\Lab_2$, and then form the label $\Lab=\Lab_1\Lab_2$. Since $q,k$ are fixed, the counter and both strings have $O_{q,k}(1)$ bits. Thus each symbol update and each label computation costs $O_{q,k}(1)$ bit operations.

Adding one label to a group sum requires $2B_{\ast}$
component updates modulo $q$. Each update adds two
integers in $[0,q-1]$ and subtracts $q$ if necessary,
so it costs $O_q(1)$ bit operations. Increasing $g$ modulo $2U_{\ast}$, setting $d=0$, and resetting $\Lab_1$ and $\Lab_2$ to empty sequences also cost $O_{q,k}(1)$ bit operations per segment. Reading $0\bmx$ processes $n+1$ symbols and gives at most $n+1$ segments. Updating $d$ and the two stored strings costs
$O_{q,k}(1)$ bit operations for each input symbol.
After a segment has been read, forming its label,
updating the group sum and counter, and preparing
the stored strings for the next segment cost
$O_{q,k}(1)$ bit operations.
There are $n+1$ input symbols and at most $n+1$
segments, so all label checks cost $O_{q,k}(n)$
bit operations.

Computing $\bmalpha$ and the label syndromes costs
$O_{q,k}(n)$ bit operations, while computing
$\mathsf{H}_{n}^{\ne}(\bmalpha)$ costs
$O_k\parenv{n\log n}$ bit operations.
The total cost is therefore
$O_{q,k}\parenv{n\log n}$, as claimed.
\end{IEEEproof}

\section{Equal-Length Edits}\label{sec_counterexample}
In \Cref{thm_uneqfinalcode}, we show that every nonempty set $\cC^{\ne}(n;\theta)$ is an unequal-length $k$-substring edit correcting code. The example below shows that the correction property in \Cref{thm_uneqfinalcode} cannot be extended to all equal-length edits.
\begin{proposition}\label{prop_example}
Let $q\ge3$ and $k=2$. At length $n=84$, there is a
syndrome value $\theta\in\cV_{\mathsf{Syn}}$ such that
$\cC^{\ne}(n;\theta)$ contains two distinct sequences
with a common received sequence after one substitution.
\end{proposition}
\begin{IEEEproof}
Here $\ell_{\ast}=2$ and $U_{\ast}=10$, so the segment
labels are grouped by their indices modulo $2U_{\ast}=20$.
Consider the following three sequences:
\begin{equation*}
    \bmy=(0010)^{21},\qquad
    \bmx=(0020)(0010)^{20},\qquad
    \bmz=(0010)^{20}(0020).
\end{equation*}
Clearly, $\bmx\ne\bmz$. All three sequences have length $n=84$. Replacing the symbol $2$ by $1$ at position $3$ of
$\bmx$, or at position $83$ of $\bmz$, gives $\bmy$.
Thus, $\bmx$ and $\bmz$ have a common received sequence after one substitution.

We next check their markers. In each block $00c0$,
where $c\in\{1,2\}$, only the first and third positions are special. Their corresponding length-three sequences are $00c$ and $c00$. A first position is a marker: every other special position within distance two is a third position, whose corresponding sequence starts with $1$ or $2$, whereas $00c$ starts with $0$. A third position is not a marker because the first position of its own block is at distance two and has a smaller corresponding sequence. There are no special coordinates outside $[84]$, since the first two symbols are zero and the sequences are extended by zeros. Hence the set of markers of each of $\bmx$, $\bmy$ and $\bmz$ is $\mathset{1,5,9,\ldots, 81}$, and $\bm{\chi}_k(\bmx)=\bm{\chi}_k(\bmz) =\bm{\chi}_k(\bmy)=(1000)^{21}$. Their marker syndromes in \eqref{eq_synHuneq} are
therefore equal.

The markers split $0\bmx$ and $0\bmz$ into the initial
segment $0$, followed by the $21$ displayed blocks
of length four.
For $\bmx$, segments $1$ and $21$ are $0020$ and $0010$;
for $\bmz$, they are $0010$ and $0020$.
All other segments are equal.
Since $1\equiv21\pmod{20}$, both labels are added to
$\bmG_1^{\ne}$.
For both sequences, this group sum contains one label
for $0020$ and one label for $0010$, so the sums are equal.
Every other group sum is also equal because its segments
are unchanged. Thus
$\bmG_r^{\ne}(\bmx)=\bmG_r^{\ne}(\bmz)$ for every
$r\in[0,19]$.

By \eqref{eq_fullsyndrome}, we have
$\mathsf{Syn}_{84}^{\ne}(\bmx)
=\mathsf{Syn}_{84}^{\ne}(\bmz)$.
Take $\theta$ to be this common value.
Then $\bmx,\bmz\in\cC^{\ne}(84;\theta)$, and both
give $\bmy$ after one substitution, as required.
\end{IEEEproof}

\bibliographystyle{IEEEtran}
\bibliography{ref}

@article{Rosser1962,
  author  = {Rosser, J. Barkley and Schoenfeld, Lowell},
  title   = {Approximate formulas for some functions of prime numbers},
  journal = {Ill. J. Math.},
  volume  = {6},
  number  = {1},
  pages   = {64--94},
  year    = {1962},
  month={Mar.}
}

@article{Levenshtein1966,
  author  = {Vladimir I. Levenshtein},
  title   = {Binary {C}odes {C}apable of {C}orrecting {D}eletions, {I}nsertions, and {R}eversals},
  journal = {Soviet Physics Doklady},
  volume  = {10},
  pages   = {707--710},
  year    = {1966}
}

@article{Levenshtein1967Burst,
  author  = {Vladimir I. Levenshtein},
  title   = {Asymptotically optimum binary code with correction for losses of one or two adjacent bits},
  journal = {Problemy Kibernetiki},
  volume  = {19},
  pages   = {293--298},
  year    = {1967},
  note    = {In Russian; English translation in Systems Theory Research, vol. 19, pp. 298--304, 1970}
}

@article{Tenengolts1984,
  author  = {G. Tenengolts},
  title   = {Nonbinary codes, correcting single deletion or insertion},
  journal = {IEEE Trans. Inf. Theory},
  volume  = {30},
  number  = {5},
  pages   = {766--769},
  year    = {1984},
  month   = {Sep.}
 }

@article{Agrawal2004Annals,
  author  = {Manindra Agrawal and Neeraj Kayal and Nitin Saxena},
  title   = {{PRIMES} is in {P}},
  journal = {Ann. Math.},
  volume  = {160},
  number  = {2},
  pages   = {781--793},
  year    = {2004},
  month={Sept.}
}

@ARTICLE{Schoeny2017IT,
  author={Schoeny, Clayton and Wachter-Zeh, Antonia and Gabrys, Ryan and Yaakobi, Eitan},
  journal={IEEE Trans. Inf. Theory}, 
  title={Codes {C}orrecting a {B}urst of {D}eletions or {I}nsertions}, 
  year={2017},
  month={Jan.},
  volume={63},
  number={4},
  pages={1971--1985}
}

@article{Gabrys2018Damerau,
  author  = {Ryan Gabrys and Eitan Yaakobi and Olgica Milenkovic},
  title   = {Codes in the {D}amerau {D}istance for {D}eletion and {A}djacent {T}ransposition {C}orrection},
  journal = {IEEE Trans. Inf. Theory},
  volume  = {64},
  number  = {4},
  pages   = {2550--2570},
  year    = {2018},
  month   = {Apr.}
  }

@article{Agrawal2019Annals,
  author  = {Manindra Agrawal and Neeraj Kayal and Nitin Saxena},
  title   = {Errata: {PRIMES} is in {P}},
  journal = {Ann. Math.},
  volume  = {189},
  number  = {1},
  pages   = {317--318},
  year    = {2019},
  month={Jan.}
}

@article{SunXiGe2023IT,
  author  = {Yubo Sun and Yuanxiao Xi and Gennian Ge},
  title   = {Sequence {R}econstruction {U}nder {S}ingle-{B}urst-{I}nsertion/{D}eletion/{E}dit {C}hannel},
  journal = {IEEE Trans. Inf. Theory},
  volume  = {69},
  number  = {7},
  pages   = {4466--4483},
  year    = {2023},
  month   = {Jul.}
  }

@article{Song2023IT,
 title = {{N}on-binary {T}wo-{D}eletion {C}orrecting {C}odes and {B}urst-{D}eletion {C}orrecting {C}odes},
  author    = {Wentu Song and Kui Cai},
  journal   = {IEEE Trans. Inf. Theory},
  volume    = {69},
  number={10},
  year      = {2023},
  month={Oct.},
  pages={6470--6484}
}

@ARTICLE{Ryan2023IT,
  author={Gabrys, Ryan and Guruswami, Venkatesan and Ribeiro, Jo{\~a}o and Wu, Ke},
  journal={IEEE Trans. Inf. Theory}, 
  title={Beyond {S}ingle-{D}eletion {C}orrecting {C}odes: {S}ubstitutions and {T}ranspositions}, 
  year={2023},
  month={Jan.},
  volume={69},
  number={1},
  pages={169--186}
  }

@article{LuZhang2023IT,
  author  = {Ziyang Lu and Yiwei Zhang},
  title   = {$t$-{D}eletion-$s$-{I}nsertion-{B}urst {C}orrecting {C}odes},
  journal = {IEEE Trans. Inf. Theory},
  volume  = {69},
  number  = {10},
  pages   = {6401--6413},
  year    = {2023},
  month   = {Oct.}
  }

@article{KhuatKim2023,
  author  = {Thi-Huong Khuat and Sunghwan Kim},
  title   = {Novel $q$-ary {C}ode {D}esign for {O}ne {E}dit or {A}djacent {T}ransposition {E}rror},
  journal = {IEEE Open J. Commun. Soc.},
  volume  = {4},
  pages   = {2830--2837},
  year    = {2023},
  doi     = {10.1109/OJCOMS.2023.3327435}
}

@article{Shuche2024IT,
 title     = {{N}on-binary {C}odes for {C}orrecting a {B}urst of at {M}ost $t$ {D}eletions},
  author    = {Shuche Wang and Yuanyuan Tang and Jin Sima and Ryan Gabrys and Farzad Farnoud},
  journal   = {IEEE Trans. Inf. Theory},
  volume    = {70},
  number={2},
  year= {2024},
  month={Feb.},
  pages={964--979}
}

@ARTICLE{Tuan2024IT,
  author={Thanh Nguyen, Tuan and Cai, Kui and Siegel, Paul H.},
  journal={IEEE Trans. Inf. Theory}, 
  title={A {N}ew {V}ersion of q-{A}ry {V}arshamov-{T}enengolts {C}odes {W}ith {M}ore {E}fficient {E}ncoders: {T}he {D}ifferential {VT} {C}odes and {T}he {D}ifferential {S}hifted {VT} {C}odes}, 
  year={2024},
  month={Oct.},
  volume={70},
  number={10},
  pages={6989--7004}
}

@article{Song2025Entropy,
 title= {Some {N}ew {C}onstructions of $q$-ary {C}odes for {C}orrecting a {B}urst of at {M}ost $t$ {D}eletions},
  author = {Wentu Song and Kui Cai and Tony Q. S. Quek},
  journal= {Entropy},
  volume= {27},
  number={1},
  pages={85},
  year= {2025},
  month={Jan.}
}

@article{YutingLi2025IT,
  author  = {Yuting Li and Yuanyuan Tang and Hao Lou and Ryan Gabrys and Farzad H. Farnoud},
  title   = {Optimal {C}odes {C}orrecting a {S}ubstring {E}dit},
  journal = {IEEE Trans. Inf. Theory},
  volume  = {71},
  number  = {7},
  pages   = {5178--5191},
  year    = {2025},
  month={Jul.}
}

@article{YuanyuanTang2025TCOMM,
  author  = {Yuanyuan Tang and Sarvin Motamen and Hao Lou and Kallie Whritenour and Shuche Wang and Ryan Gabrys and Farzad H. Farnoud},
  title   = {Correcting a {S}ubstring {E}dit {E}rror of {B}ounded {L}ength},
  journal = {IEEE Trans. Commun.},
  volume  = {73},
  number  = {1},
  pages   = {12--21},
  year    = {2025},
  month={Jan.}
}

@ARTICLE{LiGabrysFarnoud2026IT,
  author={Li, Yuting and Gabrys, Ryan and Farnoud, Farzad},
  journal={IEEE Trans. Inf. Theory}, 
  title={Constructing {L}ow-{R}edundancy {C}odes via {D}istributed {G}raph {C}oloring}, 
  year={2026},
  month={Oct.},
  volume={72},
  number={10},
  pages={7344-7358}
  }

@article{Yubo2025tsburstIT,
  author = {Yubo Sun and Ziyang Lu and Yiwei Zhang and Gennian Ge},
  title = {Asymptotically {O}ptimal {C}odes for {$(t,s)$}-{B}urst {E}rror},
  journal = {IEEE Trans. Inf. Theory},
  volume = {71},
  number = {3},
  pages = {1570--1584},
  month ={Mar.},
  year = {2025}
}

@ARTICLE{Yubosun2025IT,
  author={Yubo Sun and Gennian Ge},
  journal={IEEE Trans. Inf. Theory}, 
  title={Codes for {C}orrecting a {B}urst of {E}dits {U}sing {W}eighted-{S}ummation {VT} {S}ketch}, 
  year={2025},
  month={Mar.},
  volume={71},
  number={3},
  pages={1631--1646}
}

@ARTICLE{ZuoYe2026IT,
  author={Ye, Zuo and Sun, Yubo and Ge, Gennian},
  journal={IEEE Trans. Inf. Theory}, 
  title={Correcting {B}ursty/{L}ocalized {D}eletions: {A} {N}ew {E}rror-{P}osition-{E}stimation {C}ode}, 
  year={2026},
  month={Mar.},
  volume={72},
  number={3},
  pages={1667--1682}
  }

@book{Higham2002,
  author    = {Higham, Nicholas J.},
  title     = {Accuracy and Stability of Numerical Algorithms},
  edition = {2},
  publisher = {Society for Industrial and Applied Mathematics},
  address   = {Philadelphia, PA, USA},
  year      = {2002}
}

@book{Shoup2008,
  author    = {Shoup, Victor},
  title     = {A Computational Introduction to Number Theory and Algebra},
  edition   = {2},
  publisher = {Cambridge University Press},
  year      = {2008},
  url       = {https://www.shoup.net/ntb/}
}

@INPROCEEDINGS{Lenz2020ISIT,
  author={Lenz, Andreas and Polyanskii, Nikita},
  booktitle={Proc. Int. Symp. Inf. Theory (ISIT)}, 
  title={Optimal {C}odes {C}orrecting a {B}urst of {D}eletions of {V}ariable {L}ength}, 
  year={2020},
  month={Jun.},
  volume={},
  number={},
  pages={757--762},
  address={Los Angeles, CA, USA}
  }

@INPROCEEDINGS{RawadBitar2021ISIT,
  author={Rawad Bitar and Serge Kas Hanna and Nikita Polyanskii and Ilya Vorobyev},
  booktitle={Proc. Int. Symp. Inf. Theory (ISIT)},  
  title={Optimal {C}odes {C}orrecting {L}ocalized {D}eletions}, 
  year={2021},
  month={Jul.},
  volume={},
  number={},
  pages={1991--1996},
  address={Melbourne, Australia}
}
\end{document}